\documentclass[11pt]{article}

\usepackage[margin=1.25in]{geometry}
\usepackage[T1]{fontenc}
\usepackage[utf8]{inputenc}
\usepackage[english]{babel}
\usepackage{lmodern}
\usepackage{microtype}
\usepackage{setspace}
\usepackage{indentfirst}
\usepackage{abstract}
\usepackage[bottom]{footmisc}

\usepackage{amsmath,amssymb,amsfonts,amsthm,mathtools,bm,bbm}
\usepackage{dsfont}
\usepackage{mathrsfs}
\usepackage{IEEEtrantools}

\usepackage{graphicx}
\usepackage{booktabs}
\usepackage{array}
\usepackage{multirow}
\usepackage{tabularx}
\usepackage{longtable}
\usepackage{threeparttable}
\usepackage{caption}
\usepackage{subcaption}
\usepackage{float}
\usepackage{enumerate}
\usepackage{adjustbox}
\usepackage{rotating}
\usepackage{pdflscape}
\usepackage{lscape}
\usepackage{tablefootnote}
\usepackage{verbatim}
\usepackage{comment}

\usepackage{tikz}
\usetikzlibrary{calc}

\usepackage{xcolor}
\usepackage{url}
\usepackage[colorlinks=true,linkcolor=blue,citecolor=blue,urlcolor=blue]{hyperref}
\usepackage[nameinlink,capitalize,noabbrev]{cleveref}
\usepackage[round,colon,authoryear]{natbib}
\usepackage[normalem]{ulem}

\definecolor{burno}{rgb}{0.8, 0.33, 0.0}
\definecolor{jade}{rgb}{0.0, 0.66, 0.42}
\definecolor{gainsboro}{rgb}{0.78, 0.78, 0.78}
\definecolor{tealblue}{rgb}{0.25, 0.4, 0.96}
\definecolor{lava}{rgb}{0.81, 0.06, 0.13}

\makeatletter
\renewcommand{\maketitle}{%
  \begin{center}%
    \vspace*{0.35in}%
    {\LARGE \@title \par}%
    \vspace{1.05em}%
    {\large \@author \par}%
    \vspace{0.75em}%
    {\large \@date \par}%
  \end{center}%
  \vspace{0.65em}%
  \@thanks
}
\makeatother

\usepackage{titlesec}
\titleformat{\section}{\normalfont\Large\bfseries}{\thesection}{1em}{}
\titleformat{\subsection}{\normalfont\large\bfseries}{\thesubsection}{1em}{}
\titleformat{\subsubsection}{\normalfont\normalsize\bfseries}{\thesubsubsection}{1em}{}
\titlespacing*{\section}{0pt}{2.0ex plus 0.6ex minus 0.2ex}{1.1ex plus 0.2ex}
\titlespacing*{\subsection}{0pt}{1.7ex plus 0.5ex minus 0.2ex}{0.9ex plus 0.2ex}
\titlespacing*{\subsubsection}{0pt}{1.4ex plus 0.4ex minus 0.2ex}{0.7ex plus 0.2ex}

\theoremstyle{plain}
\newtheorem{theo}{Theorem}
\newtheorem{prop}{Proposition}
\newtheorem{lemma}{Lemma}

\theoremstyle{definition}
\newtheorem{assum}{Assumption}

\theoremstyle{remark}

\newcommand{\appendixnumbering}{%
  \numberwithin{equation}{section}%
  \numberwithin{figure}{section}%
  \numberwithin{table}{section}%
  \numberwithin{theo}{section}%
  \numberwithin{prop}{section}%
  \numberwithin{lemma}{section}%
  \numberwithin{corol}{section}%
  \numberwithin{assum}{section}%
  \numberwithin{defi}{section}%
  \renewcommand{\thetheo}{\Alph{section}.\arabic{theo}}%
  \renewcommand{\theprop}{\Alph{section}.\arabic{prop}}%
  \renewcommand{\thelemma}{\Alph{section}.\arabic{lemma}}%
  \renewcommand{\thecorol}{\Alph{section}.\arabic{corol}}%
  \renewcommand{\theassum}{\Alph{section}.\arabic{assum}}%
  \renewcommand{\thedefi}{\Alph{section}.\arabic{defi}}%
}

\DeclareUrlCommand\ULurl{%
  \renewcommand\UrlLeft{\bgroup}%
  \renewcommand\UrlRight{\egroup}}

\usepackage{xparse}
\let\OriginalIncludeGraphics\includegraphics
\RenewDocumentCommand{\includegraphics}{O{} m}{%
  \IfFileExists{#2}{\OriginalIncludeGraphics[#1]{#2}}{%
    \fbox{\begin{minipage}[c][0.28\textheight][c]{0.86\linewidth}\centering
    Missing figure: \texttt{\detokenize{#2}}\\
    Place this file in the same folder as \texttt{main.tex}.
    \end{minipage}}%
  }%
}

\title{\scshape Causality versus Serial Correlation:\\[0.35em] an Asymmetric Portmanteau Test}

\author{Amedeo \textsc{Andriollo}\thanks{\href{mailto:amedeo.andriollo@dauphine.psl.eu}{amedeo.andriollo@dauphine.psl.eu}. Universit\'e Paris-Dauphine, Universit\'e PSL, CNRS, DRM, Finance, 75016 Paris, France. I am sincerely indebted to my advisor Eric Renault. I am especially thankful to Bertille Antoine, Ao Wang, Giovanni Ricco, Luis Candelaria, Kenichi Nagasawa and Leonardo Melosi. I am grateful for comments from Anna Mikusheva, Carlos Velasco, Mikkel Plagborg-M\o ller, Majid Al Sadoon, Evgenia Passari, Michele Piffer, Giulio Rossetti, Emanuele Savini, and seminar and conference participants at several venues, including the University of Luxembourg, CREST, ESWC 2025, ICEEE 2025, the Virtual Workshop for Junior Researchers in Time Series, and the Virtual Time Series Seminar (VTTS). I would also like to thank SIdE, the Bank of Italy, and the committee of the Carlo Giannini Prize.}}

\date{\today}

\begin{document}

\maketitle

\begin{abstract}
\noindent This paper studies specification testing in dynamic linear models in the presence of omitted variables. The null hypothesis of interest is weak exogeneity: shocks have zero conditional expectation given their own past and the past of omitted variables. Existing tests based on quadratic forms of serial cross-correlations suffer from size distortions because their variance incorporates symmetric dependence in both directions, including causality from past shocks to present omitted variables (inverse causality). This paper proposes an asymmetric Portmanteau test that isolates violations of weak exogeneity from inverse causality, is asymptotically normal under the null, and does not require a parametric specification of the joint dynamics. An empirical application examines the Economic Policy Uncertainty shock series and rejects its weak exogeneity. Addressing this failure by controlling for omitted variables changes the estimated inflation response from negative to positive, suggesting a supply-side shock interpretation.
\vspace{0.05in}
\end{abstract}

\thispagestyle{empty}
\clearpage
\setcounter{page}{1}

\section{Introduction}\label{sec-1}

Estimating dynamic causal effects through structural models, such as Structural Vector Autoregressions (SVARs), is a central tool in applied macroeconometrics. Following \citet{sims1980macroeconomics}, identification assumptions allow the innovations of a multivariate time series model to be interpreted as linear functions of the underlying structural shocks.
Economists routinely exploit this link by estimating the residuals of a structural model, thus investigating the shocks' propagation through impulse response analysis \citep{kilian2017structural}. Related arguments extend to univariate approaches, such as local projections \citep{dufour1998short, plagborg2021local, jorda2023local}, combined with external instruments.

The validity of such analyses hinges critically on the correct specification of the structural dynamics, both for variables included in and external to the model. Structural shocks must be ``internally'' exogenous to the other current and lagged endogenous variables in the model \citep{ramey2016macroeconomic}, and this property cannot be undermined by variables omitted from the model. If the latter do interfere, the shocks fail to be ``externally'' exogenous and no longer represent truly \textit{unanticipated} movements in the macroeconomic system.
As a consequence, the history of the observed internal variables is insufficient to recover the shocks, violating invertibility or fundamentalness \citep{lippi1994var, nakamura2018identification}.

Existing approaches to assessing this issue have cast it either as a problem of Granger causality testing \citep{giannone2006does, forni2014sufficient, plagborg2022instrumental, miranda2023identification} or as tests of conditional mean independence \citep{chen2017testing}. From an econometric standpoint, both approaches can be viewed as addressing dynamic specification testing in the presence of omitted variables. This paper brings these perspectives together by formulating the problem in terms of weak exogeneity, understood here as the property of the structural shocks having zero conditional expectation given the past of both internal and external variables \citep{mikusheva2025linear}.

The literature on specification testing in the presence of omitted variables can be organized according to whether the practitioner explicitly models the joint dynamics of internal and external variables. When the joint system is estimated, classical specification tests are available (\citealp[e.g., Wald-type procedures, see][]{hong2005generalized, escanciano2006generalized}). Inference in this case depends on how accurately the interaction between internal and external variables is specified and estimated. Parametric joint modeling is therefore sensitive to misspecification, while strategies based on semi-parametric or nonparametric estimation face the well-known curse of dimensionality and finite-sample problems. When the joint dynamics are left unspecified, practitioners typically rely on nonparametric tests 
(\citealp[e.g., Portmanteau tests, see][]{hong1996testing, lobato2002testing}).
Despite not requiring augmentation, these tests are designed to detect symmetric forms of dependence. However, weak exogeneity imposes a directional restriction. In our context, this implies that both classes of approaches may produce rejections when there is dependence from past structural shocks to current omitted variables. Since it reflects the causal direction opposite to the one being tested, this paper refers to this dependence as \textit{inverse causality}: broadly speaking from past structural shocks to present omitted variables rather than from past omitted variables to present shocks.
This channel becomes particularly salient in applied macroeconometrics because structural shocks, which capture primitive fluctuations of the macroeconomic system, are expected to influence (external) macro variables over time. As a result, inverse causality can trigger rejections that researchers may mistakenly interpret as evidence against shocks' exogeneity.

This paper addresses these limitations by proposing an asymmetric Portmanteau test that isolates violations of weak exogeneity from inverse causality. The benchmark Portmanteau statistic tests the null of zero cross-correlation between current estimated shocks and lagged omitted variables by aggregating squared sample cross-correlations.\footnote{For univariate processes, \cite{hong1996testing} defines the Portmanteau statistic as the weighted sum of squared cross-correlation between univariate time series at positive and negative lags, with weights determined by a kernel function. Following \cite{hong2001test} and \cite{bouhaddioui2006generalized}, this paper regards as representative of the class of tests based on the serial cross-correlation function its one-sided multivariate formulation. Specifically, the benchmark is the weighted sum of the $\ell_2$ norm of the cross-correlation between the two multivariate processes at positive lags.} A key difficulty is that the quadratic norm introduces a \textit{symmetry}: the variance of the resulting statistic depends not only on dependence from lagged omitted variables to current shocks, which is relevant under the null of weak exogeneity, but also on dependence from lagged shocks to current omitted variables, that is inverse causality. The newly proposed statistic subtracts out the contribution of this latter channel, leading to a modified statistic that targets only the dependence implied by violations of weak exogeneity. By construction, the procedure avoids parametric modeling of the joint dynamics and remains robust to misspecification of how past shocks affect the present of omitted variables, which is particularly useful when prior knowledge of the interaction between omitted variables and the dynamic system is limited. 

The asymptotic distribution of the proposed test statistic is studied under the null that estimated shocks are weakly exogenous. Establishing asymptotic normality requires additional assumptions on the shocks' second and fourth conditional moments that, however, are testable. The second moment condition serves to isolate the contribution of squares to the mean vs. variance of the statistic. In technical terms, this restriction ensures correct centering and scaling of the statistic under the null. The fourth moment condition is mainly used to establish asymptotic normality.\footnote{When relaxing the assumption of independence, \cite{hong2001test}'s footnote 8 briefly discussed the condition of conditional homokurtosis for establishing the asymptotic normality of his testing procedure.} The main asymptotic result is derived for observed processes.\footnote{Extension to settings where shocks and omitted variables are estimated can be found in the Online Appendix.} Under the fixed alternatives of nonzero cross-correlation, this paper proves that the asymmetric Portmanteau achieves equivalent asymptotic power to the benchmark. 
Since the statistic has inherently limited power against alternatives involving nonlinear non-pairwise dependence, a discussion on possible generalizations of the statistic is provided. Finite-sample properties of the statistics are examined through Monte Carlo simulations, with a summary of the main findings provided in the appendix and the full results reported in the online appendix.

As an empirical application, the paper examines the exogeneity of the
\cite{baker2016measuring}'s Economic Policy Uncertainty (EPU) shocks.
The analysis controls for economic conditions using the \cite{mccracken2016fred}'s macroeconomic factors. EPU shocks fail the exogeneity test with respect to lagged macro factors. Building on this, the paper revisits \cite{diercks2024rains} and strengthens their conclusions: when these additional controls are included, the response of inflation to EPU shocks shifts from modestly negative to markedly positive. Combined with contractionary responses in other variables, this evidence suggests that the EPU structural shock operates as a supply-side negative shock, similar to the `expectational' shocks discussed in \cite{ascari2023endogenous}.

\textsc{Literature.} This paper contributes to three strands of literature. First, it relates to specification testing in dynamic linear models. Early contributions include \cite{ljung1978measure}, \cite{hosking1980multivariate} and \cite{li1981distribution}. Rather than modeling the joint process, \cite{haugh1976checking} developed a two-step procedure to test for independence between time series: first fitting univariate models, then examining cross-correlations at different lags. \cite{hong1996consistent, hong1996testing} generalized these tests to all lags using kernel-weighted schemes, with \cite{bouhaddioui2006generalized} extending the framework to multivariate processes. This paper contributes to \cite{haugh1976checking}'s approach by introducing a new version of the Portmanteau statistic that isolates a specific direction of causality. Second, the paper relates to tests of the martingale difference hypothesis used for specification testing, which indeed require modeling the conditional mean of the joint process \citep{durlauf1991spectral, hong2005generalized, escanciano2006generalized}. This class of tests can be viewed as extending \cite{ljung1978measure}'s rather than \cite{haugh1976checking}'s. The newly proposed method improves upon these by avoiding the need to model joint conditional means and variances. Third, this paper contributes to the literature on testing invertibility of structural shocks, by providing a new testing strategy.

\textsc{Outline.} Section \ref{sec2_l2norm} introduces the benchmark Portmanteau statistic and shows how inverse causality is incorporated into it. The definition of the asymmetric Portmanteau statistic is in Section \ref{chp1_sec2_subsec_corrected}. Section \ref{sec3_asymptotics} develops the asymptotic theory for the proposed test statistic. Section \ref{empirics} presents the empirical application. Section \ref{conclusion} concludes. Appendixes \ref{appendixA}-\ref{appendixC} contain all proofs and the summary of the simulation evidence.

\section{The Testing Strategies Based on $\ell_2$-norm}\label{sec2_l2norm}
 
Section \ref{chp1_sec2_subsec_preliminaries} establishes the framework and discusses the class of Portmanteau statistics. Section \ref{chp1_sec2_subsec_corrected} defines the asymmetric Portmanteau statistic. Proof of the proposition appears in Appendix \ref{appendixA}.

\subsection{Preliminaries and a Discussion on Portmanteau Statistics}\label{chp1_sec2_subsec_preliminaries}

Let $\{X_{t}, Z_{t}; t=1,..,T\}$ denote two zero-mean multivariate square-integrable jointly stationary processes of respective finite dimensions $d_1,d_2\in\mathbb{N_+}$. Let $\mathcal{I}(t-1)$ be the information set available at period $t-1$ comprising the joint past, $\{X_{s},Z_{s}; s< t\}$. Unless stated otherwise, these processes are standardized.\footnote{$\text{Var}[X_t]=I_{d_1}$ and $\text{Var}[Z_t]=I_{d_2}$. It is relaxed in the Online Appendix.}

Throughout, $X$ represents the structural shock series and $Z$ represents the omitted variables. The null hypothesis of interest is weak exogeneity, defined as the shocks having zero conditional expectation given their own past and the past of omitted variables:
\begin{align}\label{h0grangernoncausality}
    \mathcal{H}_0: \mathbb{E}[X_t| \mathcal{I}(t-1)]=0
\end{align}
Before introducing the modified test statistic in Eq.(\ref{corrected_statistic}) (Section \ref{chp1_sec2_subsec_corrected}),  we discuss testing strategies based on squared serial cross-correlations following \cite{hong1996testing}. The benchmark is the one-sided statistic based on weighted quadratic forms:
\begin{align}
\label{statistic1_hong}
\mathcal{T}_\omega & = \sum_{j=1}^{T-1} \omega(j)Q(j) \\
Q(j) & = ||\widehat{\Gamma}_{XZ}(j)||_F^2 =\text{tr}\left[ \widehat{\Gamma}_{XZ}(j)^\prime \widehat{\Gamma}_{XZ}(j)\right]=\left\vert\left\vert\text{vec}\left[\widehat{\Gamma}_{XZ}(j)\right] \right\vert\right\vert^2  \label{quadraticform}
\end{align}  
for some nonrandom non-negative weights $\{\omega(j)\}$, where $\widehat{\Gamma}_{XZ}(j)$ is the sample cross-correlation between the processes:\footnote{For a clearer exposition of the asymptotic theory, we do not consider the finite-sample corrected cross-correlation functions (i.e., scaled by $T-j$ rather than $T$), as the conclusions remain valid \citep{hong2001test, hong2005generalized}. For the finite-sample statistics, refer to Appendix \ref{appendixC}.}
\begin{align*}
\widehat{\Gamma}_{XZ}(j)=\frac{1}{T}\sum_{t=j+1}^{T}X_t Z_{t-j}^\prime, \; \; \; 
\Gamma_{XZ}(j)=\mathbb{E}[X_t Z_{t-j}^\prime], \; \; \; \; \;  j=1,...,T-1 
\end{align*}
The statistic is one-sided ($j>0$) because $\mathcal{H}_0$ concerns a particular direction of causality: the influence of the past of $Z$, $\{Z_s; s< t\}$, on the present $X$, $\{X_t\}$. The quadratic forms, $\{Q(j)\}$, correspond to the squared $\ell_2$-norms of vectorized sample cross-correlation matrices, equivalently their squared Frobenius norms.\footnote{This generalization from univariate to multivariate analysis dates to \cite{li1981distribution}. See \cite{bouhaddioui2006generalized} for further discussion. For the equivalence between Euclidean norm and trace, see chapter 4 of \cite{lutkepohl1997handbook}.} For its connection to kernel estimation of the spectrum, see \cite{hong1996consistent}. 

By construction, the quadratic forms treat causality directions symmetrically. This becomes apparent when decomposing the test statistic: the inner product of cross-correlation matrices generates cross-product terms where $X$ and $Z$ enter symmetrically across different time lags. Formally, by means of some algebra: 
\begin{align}
&\mathcal{T}_\omega = \mathcal{T}_{1\omega} + \mathcal{T}_{2\omega}  \label{decompositionsintosums} \\
&\mathcal{T}_{1\omega} =\frac{1}{T^2} \sum_{j=1}^{T-1} \omega(j) \sum_{t=j+1}^T || X_t ||^2|| Z_{t-j} ||^2 \nonumber\\
& \mathcal{T}_{2\omega} =\frac{1}{T^2} \sum_{j=1}^{T-2} \omega(j) \sum_{s,t=j+1, s\not=t}^T <X_t,X_s><Z_{t-j},Z_{s-j}> \nonumber 
\end{align}
The Portmanteau statistic, $\mathcal{T}_\omega$, consists of two components: the ``sum of squares'', $\mathcal{T}_{1\omega}$, and the ``sum of cross-products'', $\mathcal{T}_{2\omega}$. While the sum of squares preserves temporal ordering with respect to the tested causal direction (from past omitted variables to present shocks), the sum of cross-products incorporates the interaction between two time indexes, $s$ and $t$, thus blending both directions of causality. This symmetry due to the norm suggests that when testing the null, Portmanteau statistics may not effectively distinguish between violations of weak exogeneity and dependence from past shocks to present omitted variables (inverse causality).

The distinction between components matters for understanding asymptotic properties of the testing strategies based on \cite{hong1996testing} and subsequent work. Intuitively, the sum of cross-products, $\mathcal{T}_{2\omega}$, dominates under the null and thus controls test size, whereas the sum of squares, $\mathcal{T}_{1\omega}$, dominates under the alternatives, and so regulates power.\footnote{
In Proposition \ref{prop_moments}, the asymptotic properties of the test under the null are described in detail. For the power of the test, refer to Theorem \ref{theorem_power} and the discussion that follows.}

In a stylized setting, the following proposition clarifies how cross-product terms incorporate both directions of dependence. We impose two simplifying assumptions: i) marginal independence of the processes, i.e., $X_t \perp\!\!\!\perp X_k, Z_t \perp\!\!\!\perp Z_k, t\not=k $; ii) conditional homoskedasticity of the shocks, $\mathbb{E}[||X_t||^2|\mathcal{I}(t-1)]=\mathbb{E}[||X_t||^2]$.

\begin{prop} \label{prop_toyexample1}
Let $\{X_t, Z_t\}$ be marginally i.i.d. processes with finite fourth moments, such that the process $\{X_t\}$ is homoskedastic  conditionally on the joint past, $\mathcal{I}(t-1)$. Under the null hypothesis in Eq.(\ref{h0grangernoncausality}), the variance of the benchmark statistic, $\mathcal{T}_{\omega}$, depends on the inverse causality through the variance of the sum of cross-products, $\mathcal{T}_{2\omega}$. In particular, let us consider the time indexes such that $t>s$, then we write:
\begin{align*}
\mathbb{E}[(<X_t,X_s><Z_{t-j},Z_{s-j}>)^2]=\begin{cases}
 d_1 d_2, \; \; \; \; \; \; \; \; \; \;  \; \; \; \; \; \; \; \; \; \;  \; \; \; \; \; \; \; \; \; \;  \; \; \; s> t-j \\
 \mathbb{E}[ ||X_s ||^2<Z_{t-j},Z_{s-j}>^2], \; \; s\leq t-j
 \end{cases}
\end{align*}
Under mutual independence of the two processes, $X_t \perp\!\!\!\perp Z_s, \forall s,t $ , we have:
\begin{align*}
 &\mathbb{E}[(<X_t,X_s><Z_{t-j},Z_{s-j}>)^2]=d_1 d_2.
\end{align*}
\end{prop}
Proposition \ref{prop_toyexample1} demonstrates that under the null hypothesis,
the variance of the statistic $\mathcal{T}_\omega$ incorporates, through the cross-products in $\mathcal{T}_{2\omega}$, the dependence from past $X$ to present $Z$ captured by cross-moments of the joint process (i.e., the inverse causality channel). This vanishes when either the processes are independent, $X_t \perp\!\!\!\perp Z_s, \forall t,s$ (strict exogeneity between shocks and omitted variables at all lags/leads), or when a specific time ordering holds, $s>t-j$ (with $t>s$). Two remarks follow. First, in the univariate case $(d_1=d_2=1)$, it is evident that inverse causality operates through the conditional variance of $Z$. Second, the presence of such dependencies in the variances arises even: a) under the restrictive assumption on the univariate processes, namely marginally i.i.d. time series, b) under past independence of shocks, $X_t \perp\!\!\!\perp Z_{s_1}, X_{s_2},$ with $s_1,s_2< t$.\footnote{Past independence would imply conditional homoskedasticity. Similar assumptions to past independence appear in \cite{hong2009granger}, and \cite{candelon2016nonparametric}. This condition, together with the assumption of marginally i.i.d., is weaker than statistical independence, $X_t \perp\!\!\!\perp Z_s, \forall s,t $.}

Lemma \ref{coroll_toyexample1} in the Appendix \ref{lemma_inversecausality} explicitly derives how dependencies from shocks to omitted variables affect the variance of the sum of cross-products, $\mathcal{T}_{2\omega}$, under a general class of DGPs where inverse causality is present and the null of interest holds true.

\subsection{An Asymmetric Portmanteau Statistic}\label{chp1_sec2_subsec_corrected}
Motivated by the previous proposition, we introduce a modified statistic that accounts for the following cross-products of unordered time indexes $(s,t)$:
\begin{align}
    \mathcal{T}_\omega^c & = \mathcal{T}_\omega - \frac{1}{T^2} \sum_{j=1}^{T-2} \omega(j) \sum_{s,t=j+1, s\not=t, j\leq |t-s|}^T \langle X_t,X_s\rangle \langle Z_{t-j},Z_{s-j}\rangle=\mathcal{T}_{\omega}-\mathcal{C}_{\omega} \nonumber\\
    &  = \mathcal{T}_{1\omega}+\frac{2}{T^2}\sum_{t=2}^{T}\sum_{s=1}^{t-1}\sum_{j=t-s+1}^{s-1}\omega(j)\langle X_t,X_s\rangle \langle Z_{t-j},Z_{s-j}\rangle=\mathcal{T}_{1\omega}+\mathcal{T}_{2\omega}^c \label{corrected_statistic}
\end{align}
The two formulations are associated with either the inverse causality channel, $\mathcal{C}_{\omega}$, or the corrected causality channel, $\mathcal{T}_{2\omega}^c$. Proposition \ref{prop_toyexample1} motivates removing cross-products whose time ordering prevents a martingale
structure, or equivalently retaining the one satisfying $j>|t-s|$ for the unordered indexes $(s,t)$. Since the correction term, $\mathcal{C}_{\omega}$, differences out the influence of inverse causality and so breaks the symmetry, the proposed statistic becomes an \textit{asymmetric} Portmanteau test. To further motivate the correction term, we offer two perspectives.

The first perspective interprets a subset of the cross-products as bias in the variance of sample covariance estimators, under weak exogeneity and conditional homoskedasticity (Proposition \ref{prop_toyexample1}). Along this heuristic, a jackknife solution would then employ a block-deletion scheme targeting observations associated with the temporal ordering $j\leq |t-s|$ (for a given $j$), precisely those observations linked to the inverse causality channel.

The second perspective provides deeper insight by viewing moments of cross-products as coefficients in predictive regressions. Consider the bivariate process $\{X_t,Z_t\}$ for time indexes $t>s$. After applying the correction, the first moment of remaining cross-products, $\mathcal{T}_{2\omega}^c$, is proportional to coefficients in regressions of the form: $X_tX_s = \sum_{j=t-s+1}^{s-1} \phi_{j}^{(1)}Z_{t-j}Z_{s-j}+e_t$ (for a fixed $s>0$), where $e_t$ is an error term, $Z_{t-j},Z_{s-j}\in\mathcal{I}(t-j)$ and $X_s\not\in\mathcal{I}(t-j)$.
Conversely, the first moment of cross-products constituting the correction term, $\mathcal{C}_{\omega}$, is proportional to coefficients in autoregressions: $X_tZ_{t-j} = \sum_{l=1}^s \varphi_{l}^{(1)}X_{l}Z_{l-j}+\epsilon_t$ (for a fixed $j>0$), where $\epsilon_t$ is an error term, $X_{l},Z_{l-j}\in\mathcal{I}(s)$ and $Z_{t-j}\not\in\mathcal{I}(s)$.
These coefficients assess distinct implications of the null. By law of iterated expectations, weak exogeneity implies: i) $\mathbb{E}[X_tX_s]=0, \; \forall s<t$; ii) $\mathbb{E}[X_tZ_{t-j}]=0, \; \forall j>0$, captured by the first and second regression sets, respectively. When looking at the second moment of $\mathcal{T}_{2\omega}^c$ and  $\mathcal{C}_\omega$, then those are proportional to coefficients of the following regressions: 
\begin{align*}
X_t^2X_s^2 &= \sum_{j=t-s+1}^{s-1}  \phi_{j}^{(2)}Z_{t-j}^2Z_{s-j}^2+e_t^{(2)}, \quad s>t-j, \quad \text{(for a fixed $s>0$)} \\
X_t^2Z_{t-j}^2 &= \sum_{l=1}^s \varphi_{l}^{(2)}X_{l}^{2}Z_{l-j}^{2}+\epsilon_t^{(2)}, \quad l\leq t-j, \quad \text{(for a fixed $j>0$)}
\end{align*}
Conditional homoskedasticity of $X$ has the following implication:\\ 
$ \phi_{j}^{(2)}\propto  \mathbb{E}[(X_t^2X_s^2)M_{Z} (Z_{t-j}^2Z_{s-j}^2)]/\text{Var}[Z_{t-j}^2Z_{s-j}^2]=d_1 \mathbb{E}[M_{Z}Z_{t-j}^2 Z_{s-j}^2]/\text{Var}[Z_{t-j}^2Z_{s-j}^2]$,
for an appropriate projection operator, $M_Z \in \mathcal{I}(s)$, following a Frisch–Waugh–Lovell argument. Coefficients $\{\phi_{j}^{(2)}\}$ from the first regression set depend solely on higher moments of the marginal process $\{Z_t\}$. In contrast, coefficients from the second set depend on higher-order moments of the joint process $\{X_t,Z_t\}$. Consistent with Proposition \ref{prop_toyexample1}, this distinction motivates the correction term.

\section{Asymptotic Theory}\label{sec3_asymptotics}

Section \ref{chp1_sec3_subsec_asymptotics} establishes the asymptotic properties of the proposed statistic under the null of interest. Section \ref{chp1_sec3_subsec_consistencyunder} derives asymptotic properties under a general class of alternatives and discusses settings where the test has limited power. Proofs of the Proposition and Theorems appear in Appendix \ref{appendixB}. A discussion on the finite-sample properties of the statistics can be found in Appendix \ref{appendixC}.
In the Online Appendix, the main result of Section \ref{chp1_sec3_subsec_asymptotics} is extended to estimated processes, and the full results of the simulations are reported.

\subsection{Asymptotics of the Statistic under the Null}\label{chp1_sec3_subsec_asymptotics}

This paper considers the following condition on the weighting scheme in Eq.(\ref{statistic1_hong}):

\begin{assum}\label{assumption_kernel1} Let the sequence of weights $\{\omega(j)\}$ be a function of some sequence of integers $M=M(T)$ for which there exists an appropriate square-integrable kernel $k(\cdot):\mathbb{R}\rightarrow [-1,1]$, continuous at 0 and at all points except for a finite number of points, such that: $\omega(j)=k^2(j/M)$, $k(0)=1$.
\end{assum}
This assumption is standard in nonparametric spectral density estimation via kernel functions \citep{hong2001test}. The sequence of integers $M$, growing with sample size $T$, characterizes the kernel estimation window: larger $M$ incorporates more lags when summing cross-correlations in the statistic.\footnote{See the discussion at the end of Section \ref{chp1_sec3_subsec_asymptotics} regarding the smoothing parameter.} Define:
\begin{equation} \label{kernel1}
\begin{aligned}
     \mu_{\omega,T}&=\mu[\{\omega(j)\},T]=d_1d_2\sum_{j=1}^{T-1}\left(1-\dfrac{j}{T}\right) \omega(j) \\
     D_{\omega,T}^{(Hete)} &=D^{(Hete)}[\{\omega(j)\},T] \\
     & \quad 
        =\frac{4d_1}{T^2}
        \sum_{j=1}^{T-2}\sum_{\ell=1}^{T-2}\omega(j)\omega(\ell)
        \sum_{s=\max\{j,\ell\}+1}^{T-1}
        \sum_{t=s+1}^{\min\{T,\ s+\min(j,\ell)-1\}}
        \gamma_{t,s}(j,\ell) \\
     D_{\omega,T}&=D[\{\omega(j)\},T]=2d_1d_2\sum_{j=1}^{T-2}\left(1-\dfrac{j}{T}\right)\left(1-\dfrac{j+1}{T}\right) \omega^2(j)
\end{aligned}
\end{equation}
with $\gamma_{t,s}(j,l)=\mathbb E\!\left[\langle Z_{t-j},Z_{s-j}\rangle \langle Z_{t-\ell},Z_{s-\ell}\rangle\right]$, assuming these last moments exist.
The first two quantities approximately match the mean and the variance of the proposed statistic under the null hypothesis, scaled by the sample size (Proposition \ref{prop_moments}). The last quantity, $D_{\omega,T}$, represents the asymptotic variance of the benchmark statistic scaled by the sample size, under mutual independence of $X$ and $Z$ \citep[pg.191-2]{hong2001test}. Under Assumption \ref{assumption_kernel1} and appropriate conditions on $\{\gamma_{t,s}(j,l)\}$, these quantities are of same order $M$: $\mu_{\omega,T}=O(M)$, $D_{\omega,T}=O(M)$, and $D_{\omega,T}^{(Hete)}=O(M)$. The following proposition characterizes the moments of the asymmetric Portmanteau statistic.
\begin{prop} \label{prop_moments}
Suppose Assumption \ref{assumption_kernel1} holds, with $\frac{M}{T}\rightarrow 0$, as $T,M\rightarrow \infty$.
\begin{enumerate}[i)] 
\item Suppose $Z$ has finite fourth moments and:\\ $|\text{Cov}[|| Z_{1}||^2,|| Z_{1+h}||^2]|=O(h^{-1-\epsilon})$ for $\epsilon>0$.\\
If: $\mathbb{E}[X_tX_t^\prime|\mathcal{I}(t-1)]= \mathbb{E}[X_tX_t^\prime]$, then: $\mathbb{E}[T \cdot \mathcal{T}_{1\omega}]=\mu_{\omega,T}$.\\
In addition, if: $\mathbb{E}[(X_tX_t^\prime)\otimes (X_tX_t^\prime) |\mathcal{I}(t-1)]= \mathbb{E}[(X_tX_t^\prime)\otimes (X_tX_t^\prime) ]$,\\ then: $\text{Var}[T\cdot \mathcal{T}_{1\omega}]=O(M/T)$, implying mean-squared convergence:\\  $\lim_{T\rightarrow\infty} (T \cdot\mathcal{T}_{1\omega}-\mu_\omega)\left(D_{\omega,T}\right)^{-1/2}=0.$
\item Under $\mathcal{H}_0$ stated in Eq.(\ref{h0grangernoncausality}), we have: $\mathbb{E}[T\cdot\mathcal{T}_{2\omega}^c]=0$.\\
If additionally: $\mathbb{E}[X_tX_t^\prime|\mathcal{I}(t-1)]= \mathbb{E}[X_tX_t^\prime]$, then: $\text{Var}[T\cdot\mathcal{T}_{2\omega}^c]= D_{\omega,T}^{(Hete)}$.
\end{enumerate}
\end{prop}
Proposition \ref{prop_moments} formalizes the benefits of the correction term. Under conditional homoskedasticity and conditional homokurtosis of $X$ with respect to the joint past, the first two moments of the asymmetric Portmanteau statistic do not incorporate dependencies running from past $X$ to present $Z$. Specifically, these moments depend only on the weighting scheme, $\{\omega(j)\}$ or, at most, on particular higher moments of the marginal process $Z$, $\{\gamma_{t,s}\}$, rather than on moments of the joint process that would reflect inverse causality. In fact, by breaking the symmetry of the quadratic form, the correction term restores a martingale structure that permits separating variances via law of iterated expectations. Consequently, the variance of the proposed statistic reduces to (nonparametrically) estimating the long-run variance of second-order moments of the process $Z$ (i.e., the cross-products $\{\langle Z_{t-j},Z_{s-j}\rangle\}$).

Propositions \ref{prop_toyexample1}-\ref{prop_moments} highlight a trade-off in terms of restrictions on marginal vs. joint processes: maintaining directional inference while avoiding to specify the dynamics between shocks and omitted variables. This paper prioritizes an agnostic stance toward inverse causality, imposing minimal restrictions on how past shocks influence present omitted variables, at the cost of stronger moment restrictions on the structural shocks themselves. Indeed, these moment restrictions serve to isolate \textit{effects} rather than \textit{causes}: they ensure the statistic correctly detects violations of weak exogeneity without constraining the inverse causality channel.
Conditional homoskedasticity ensures proper centering by isolating weak exogeneity violations to the mean of cross-products, $\mathcal{T}_{2\omega}^c$, rather than the sum of squares, $\mathcal{T}_{1\omega}$. Conditional homokurtosis bounds the variance of the latter sum, ensuring that the cross-products dominate under the null.

These moment restrictions on structural shocks are testable and offer practical advantages for empirical application. Since structural shocks are estimated rather than observed, conditional moment restrictions can guide towards sharper identification (e.g., \citealp{hafner2022identification}, or related heteroskedastic identification schemes). Under additional parametric assumptions on inverse causality, these conditions could be relaxed.\footnote{These moment restrictions are weaker than mutual independence \citep{hong1996testing} or past independence \citep{candelon2016nonparametric}, and comparable to approximate q-dependence \citep{hong2005generalized}, while being more directly testable.} Notably, analogous conditions have been proposed for testing in Proxy-SVAR frameworks \citep{bruns2024testing}.

Using the quantities defined in Eq.(\ref{corrected_statistic})-(\ref{kernel1}), let:
\begin{align*}
    \mathcal{T}^{(Hete)}:=T\left(\dfrac{ \mathcal{T}_{1\omega}-\mu_{\omega,T}}{\sqrt{D_{\omega,T}}}+\dfrac{\mathcal{T}_{2\omega}^c}{\sqrt{D_{\omega,T}^{(Hete)}}}\right), \quad \mathcal{T}^{(Hong)}:=T \left(\dfrac{\mathcal{T}_{1\omega}-\mu_{\omega,T}}{\sqrt{D_{\omega,T}}}+\dfrac{\mathcal{T}_{2\omega}}{\sqrt{D_{\omega,T}}}\right),
\end{align*}
where $\mathcal{T}^{(Hete)}$ is the centered and scaled version of the proposed asymmetric Portmanteau statistic, whereas $\mathcal{T}^{(Hong)}$ corresponds to centered and scaled version of the benchmark (e.g., \citealp[]{hong1996consistent}). To state the dependence conditions used below, define the triangular array: $\Lambda_{s,t}= \sum_{j=t-s+1}^{s-1}\omega(j)X_{s}\langle Z_{t-j},Z_{s-j}\rangle$, with $1\leq s\leq t$, and for $r\geq 1$, define the fourth-order dependence coefficient: \\
$C_{r,4}^{\Lambda}:= \sup_{t\geq 1}\sup_{m=1,2,3}\sup_{\mathcal I_{m,r,t}}\left\|\operatorname{Cov}\left(\Lambda_{i_1,t}\otimes..\otimes\Lambda_{i_m,t},\Lambda_{i_{m+1},t}\otimes..\otimes\Lambda_{i_4,t}\right)\right\|_{F}$,\\
with: $\mathcal I_{m,r,t}
=\left\{(i_1,..,i_4):1\leq i_1\leq..\leq i_m<i_{m+1}\leq..\leq i_4\leq t,i_{m+1}-i_m\geq r
\right\}.$\\
Additionally, we consider the following assumption on the process $Z$:
\begin{assum}\label{assumptionZmixing} The process $\{Z_t\}$ is strictly stationary, has finite $(8+\delta)$-order moments, with $\alpha(h)$ such that:
$\sum_{h=1}^\infty \alpha(h)^{\delta/(8+\delta)}<\infty$.
\end{assum}

\begin{theo}\label{theorem1}
Suppose the process $\{X_t\}$ is such that:
\begin{align*}
    \mathbb{E}[X_tX_t^\prime|\mathcal{I}(t-1)]= \mathbb{E}[X_tX_t^\prime], \; \; \;\mathbb{E}[(X_tX_t^\prime)\otimes (X_tX_t^\prime) |\mathcal{I}(t-1)]= \mathbb{E}[(X_tX_t^\prime)\otimes (X_tX_t^\prime) ]
\end{align*}
Suppose the time series $\{Z_t\}$ satisfies Assumption \ref{assumptionZmixing}. Further, suppose the joint process $\{X_t,Z_t\}$ is strictly stationary, and Assumption \ref{assumption_kernel1} holds with $\frac{M^2}{T}\rightarrow 0$, as both $T,M\rightarrow \infty$. \\
Under the null $\mathcal{H}_0$ in Eq.(\ref{h0grangernoncausality}), we have: $\mathcal{T}^{(Hete)}\xrightarrow{d} \mathcal{N}(0,1)$.\\
If additionally the joint process $\{X_t,Z_t\}$ satisfies: $C_{r,4}^{\Lambda}=O(r^{-2})$, for $r\rightarrow\infty$, the asymptotic normality holds with $\frac{M}{T}\rightarrow 0$, as both $T,M\rightarrow \infty$.
\end{theo}
Theorem \ref{theorem1} offers two notable improvements over existing testing strategies.

First, Portmanteau statistics following \cite{hong1996consistent, hong1996testing} are typically studied under statistical independence. In the presence of inverse causality, benchmark tests based on squared cross-correlations may therefore exhibit size distortions under weak exogeneity, as their higher moments incorporate dependencies from past $X$ to present $Z$ (Proposition \ref{prop_toyexample1}). The asymmetric Portmanteau statistic addresses this issue directly: Theorem \ref{theorem1} establishes that its asymptotic normality depends essentially on the martingale properties of $X$ with respect to the joint past.

Second, tests for the martingale difference property \citep{hong2005generalized, escanciano2006generalized} could, in principle, be used to test $\mathcal{H}_0$, but at a high cost. These procedures would require either (i) the joint process $\{X_t,Z_t\}$ to be a martingale difference sequence (stronger than weak exogeneity), or (ii) explicit modeling of the joint conditional mean. The latter approach would involve additional high-level assumptions (e.g., Assumptions A2-A3 in \citealp{hong2005generalized}) whose implications and testability are less transparent than the primitive moment conditions imposed in Theorem \ref{theorem1}. By contrast, the proposed framework avoids modeling the joint dynamics altogether, while maintaining interpretable and testable restrictions on the moments of the structural shocks.

The main cost relative to \cite{hong1996testing} is a more stringent rate condition on the smoothing parameter $M$, which must diverge slower than $\sqrt{T}$. This slower rate reflects the need to control the variance of the sum of squares under minimal assumptions on the marginal process $Z$ (i.e., finite eighth moments). However, the second part of Theorem \ref{theorem1} shows that, under mild additional dependence conditions \citep{dedecker2007weak}, the standard rate $M/T\rightarrow 0$ suffices. 

The smoothing parameter $M$ governs the number of lags considered in the test statistic and thus the rate at which the weighted sum of covariance terms converges to a Gaussian limit. For small/finite $M$, the limiting distribution is a weighted sum of chi-squared variables \citep{box1970distribution, francq2007multivariate}. As $M$ increases with $T$, the sum converges to normality by a standard central limit argument. The choice of $M$ involves a familiar trade-off: sufficiently fast growth ensures asymptotic normality under the null, while sufficiently slow growth preserves power against alternatives (Theorem \ref{theorem_power}). In addition, since the correction term permits separating the variances, the smoothing parameter $M$ also governs the effective window used to estimate the long-run variance of second-order moments of $Z$, thus relating its choice to the standard size-power tradeoff in HAR inference \citep{lazarus2021size}.

\subsection{Consistency under the Alternatives}\label{chp1_sec3_subsec_consistencyunder}

This section establishes the asymptotic power under a general class of alternatives. Let $\kappa_{mrmr,XY}(j,k,l)$ denote the fourth-order cumulant of $\{X_{m,t},Z_{r,t-j},X_{m,t-k},Z_{r,t-l}\}$, where $X_{m,t}$ and $Z_{r,t}$ are the $m^{th}$ and $r^{th}$ entries of $X_{t}$ and $Z_{t}$.
We require absolute summability of fourth-order cumulants: $\sum_{m,r=1}^{d_1,d_2}\sum_{j,k,l=-\infty}^\infty \vert \kappa_{mrmr,XZ}(j,k,l)\vert <\infty $. 
\begin{theo}\label{theorem_power}
Suppose $\{X_{t},Z_{t}\}$ is jointly fourth-order stationary process with absolute summability of the fourth-order cumulants, with: $|\text{Cov}[|| Z_{1}||^2,|| Z_{1+h}||^2]|=O(h^{-1-\epsilon})$ for $\epsilon>0$. Suppose further: $ \exists j>0\;, ||\Gamma_{XZ}(j)||\not= 0 \;$, with $\; \sum_{j=1}^\infty|| \Gamma_{XZ}(j)||^2< \infty$. Suppose Assumption \ref{assumption_kernel1} holds with $\frac{M}{T}\rightarrow 0$ as both $T,M\rightarrow \infty$.\\ We have: $ (M^{1/2}T^{-1})\mathcal{T}^{(Hete)}\xrightarrow{p} \Delta 
\sum_{j=1}^{\infty} \left\vert\left\vert \text{vec}\left[\Gamma_{XZ}(j)\right] \right\vert\right\vert^2$, for a finite $\Delta>0$.\\ Consequently, for any fixed positive $K\in \mathbb{R} $: $\lim_{T,M\rightarrow\infty}Pr\left(\left\vert\mathcal{T}^{(Hete)}\right\vert>K\right)\xrightarrow[]{} 1.$
\end{theo}
Theorem \ref{theorem_power} establishes a consistency result: the proposed statistic, properly centered and scaled, converges in probability to the sum of squared cross-correlations across lags (up to a positive scalar). Under fixed alternatives with nonzero cross-correlation, the proposed statistic explodes at rate $(M^{1/2}/T)^{-1}$. Slower growth of $M$ yields faster divergence and higher power, consistent with the discussion concluding Section \ref{chp1_sec3_subsec_asymptotics}. Fourth-order stationarity and absolute summability of joint cumulants are standard conditions imposed when studying the power of tests following \cite{hong1996testing}, as they accommodate a wide class of processes \citep[p.846][]{hong1996consistent}.\footnote{More general conditions are discussed in \cite{lobato2002testing}, see pg.731-3.} Two technical remarks clarify the role of fourth-order cumulants in our framework.

First, the proof exploits that under alternatives the sum of squares, $\mathcal{T}_{1\omega}$, stochastically dominates the sum of cross-products, $\mathcal{T}_{2\omega}^c$, as anticipated in Section \ref{chp1_sec2_subsec_preliminaries}. This follows from Theorem 6 in \cite{hannan1970multiple} (pg.210) via an Isserlis-type argument \citep[see Eq.(5.3.20)][]{priestley1981spectral}, based on establishing $\ell_2$-convergence of covariance estimators to their population counterparts under fourth-order stationarity and absolute summability. These conditions ensure the process is sufficiently close to a multivariate normal (or to a generalized linear process) for mean-square convergence to take effect. Crucially, while fourth-order cumulants are asymptotically negligible under alternatives (Theorem \ref{theorem_power}), they govern the asymptotic behavior under the null (Theorem \ref{theorem1}). When $X$ is a martingale with respect to higher moments, cumulants drive the distribution of the test statistic under $\mathcal{H}_0$ (Proposition \ref{prop_moments}). Hence, the correction term in Eq.(\ref{corrected_statistic}) specifically targets the subset of cumulants associated with inverse causality.

Second, our asymptotic approach differs fundamentally from \cite{escanciano2006generalized}. Their strategy first establishes that sample autocovariances converge weakly in $\ell_2$-norm to a Gaussian process under the null of martingale difference (their Theorem 1).\footnote{They generalize autocovariances to measure conditional mean dependence nonparametrically (pg.155). Convergence occurs in the Hilbert space of square-integrable functions; see pg.158–159 for details.} Second, since covariances enter their statistic ``squared'', they show their statistic converges in distribution to a weighted sum of independent $\chi^2_1$ variables. In contrast, \cite{hong1996consistent} focuses on quadratic forms directly, showing that sums of cumulants converge to normality. The distinction is one of convergence order: \cite{escanciano2006generalized} achieve convergence of covariances \textit{before} squaring, while \cite{hong1996consistent}'s inferential theory with cross-products is \textit{after} squaring. Our framework follows the latter approach, underscoring the importance of fourth-order cumulants when inverse causality is left unrestricted, precisely the setting where practitioners lack information about how omitted variables $Z$ interact with the structural shock and the structural dynamics.

The proposed test has no power against uncorrelated but non-martingale processes. This reflects a well-known limitation of Portmanteau tests: nonlinear dependencies from past $Z$ to present $X$ that leave linear associations unaffected cannot be detected. One potential extension addresses this limitation through generalized spectral analysis \citep{hong2005generalized}. Rather than summing squared covariances between $X$ and $Z$, their approach considers squared covariances between $X$ and the empirical characteristic function of $Z$, thereby capturing nonlinear dependencies. Since their statistic involves quadratic forms \citep[Eq.3.11-3 in][]{hong2005generalized}, an analogous correction term can be constructed. Adapting their framework, define:
\begin{align*}
         &\mathcal{T}_\omega^{(HL)} = \sum_{j=1}^{T-1}\omega(j)\int_{\mathbb{R}}\left\lvert\widehat{\Psi}_{XZ}(j;v)\right\rvert^2Q(dv), \quad \quad \widehat{\Psi}_{XZ}(j;v)= \dfrac{1}{T}\sum_{t=j+1}^T X_t \breve{Z}_{t-j}(v)\\
         & \breve{Z}_{s}(v)= e^{iv Z_s}- T^{-1}\sum_{m=1}^T e^{iv Z_s}
     \end{align*}
 where, in this instance, $i$ denotes the imaginary unit and $Q: \mathbb{R}\rightarrow \mathbb{R}^+$ a nondecreasing weighting function symmetric about zero. Parallel to Eq.(\ref{corrected_statistic}), the asymmetric version of \cite{hong2005generalized}'s statistic is: 
\begin{align*}
         \mathcal{T}_\omega^{(HL),c} & = \mathcal{T}_\omega^{(HL)} - \left(\frac{1}{T^2} \sum_{j=1}^{T-2} \omega(j) \int_{\mathbb{R}}\left(  \sum_{s,t=j+1, j\leq |t-s|}^TX_t X_s \breve{Z}_{t-j}(v) \breve{Z}_{s-j}(v)\right) W(dv)\right) \nonumber
\end{align*}
A second limitation concerns the nature of detectable dependencies. Portmanteau statistics capture temporal dependence in a pairwise manner, detecting relationships of the form $\mathbb{E}[X_{t}|Z_{t-j}]$ but not necessarily non-pairwise interactions such as, for instance, \\ $\mathbb{E}[X_{t}|Z_{t-j}, Z_{t-j-1}]$. A standard solution is to augment $Z$ with its first $R$ lags: \\ $\{Z_t^\ddag=(Z_t^\prime,Z_{t-1}^\prime,...,Z_{t-R}^\prime)^\prime\}$, and then testing the null using cross-covariances between $X$ and $Z^\ddag$ \citep[][]{dominguez2004consistent, kuan2004new, wang2022testing}. Albeit potentially applicable here, this type of solution complicates inference: the variance of the resulting statistic involves substantially more cross-product terms, requiring more elaborate correction terms to difference out the inverse causality effects. Developing such extensions remains an avenue for future research.

Finite-sample properties under the null and fixed alternatives are studied via
Monte Carlo experiments, summarized in Appendix \ref{appendixC}. Appendix \ref{appendixC} also define the finite-sample versions of the statistics that are used in the simulations and empirical application. The full set of results is reported in the Online Appendix.

\section{Empirical application} \label{empirics}

This section presents an empirical application of the proposed testing procedures. Section \ref{subsec_invertibility_structural} introduces the concept of invertibility or fundamentalness of structural shocks and relates it to the null of interest. Section \ref{subsec_EPU} studies the exogeneity property of the uncertainty shock of \cite{baker2016measuring} by revisiting the empirical analysis of \cite{diercks2024rains}. 

\subsection{Testing Invertibility of Structural Shocks}\label{subsec_invertibility_structural}

In applied macroeconometrics, by a Wold-type of argument, common practice is to assume that the macroeconomic (stationary) multivariate time series, $\{W_{t}\}$, admits a Moving Average (MA) representation driven by mutually orthogonal structural shocks: $ W_t=B(L)\epsilon_t = \sum_{j=0}^{\infty}B_j\epsilon_{t-j}$, with $\epsilon_{t}\sim (0,I)$,
where $B(L)$ captures the propagation of the structural disturbances.\footnote{This rationale is supported by a twofold motivation: i) by the Wold Representation theorem, if the time series is covariance-stationary then it admits a MA($\infty$) representation \citep{brockwell1987time}, with the Wold innovations  being the reduced-form residuals of the linear projection of $W$ onto its infinite past; ii) the linear (or linearized) dynamic stochastic economic model, based on the variables $W$, usually admits a VARMA solution, whose structural shocks are assumed to be mutually orthogonal \citep{fernandez2007abcs}.} When $W$ is causal and invertible, structural shocks can be recovered from current and lagged values of $W$, up to a rotation matrix which governs the instantaneous relationships among the components of $W$, thus in turn motivating the use of SVAR models. Invertibility may however fail when economic agents' information differs from the econometrician's \citep{hansen2019two} and, in such cases, the MA representation is said to be non-fundamental \citep{lippi1994var, nakamura2018identification}. In practice, the issue of invertibility spells out as a problem of VAR misspecification, due to omitted variables or insufficient set of lagged controls \citep{chen2017testing, miranda2023identification}.

To fix ideas, following \cite{giannone2006does}, partition $W$ into two blocks $W_1$ and $W_2$ of dimensions $d_1$ and $d_2$, with reduced-form residuals $(X_{t},Z_t)^\prime$ and structural shocks $(\epsilon_{1,t},\epsilon_{2,t})^\prime$, where the process $Z$ collects the innovations of the block $W_2$. For simplicity, suppose $B_0=I$, so that structural shocks coincide with reduced-form residuals:
\begin{align*}
       \begin{pmatrix} W_{1,t} \\ W_{2,t} \end{pmatrix}= \begin{pmatrix} A_{1,1}(L) & A_{1,2}(L) \\  A_{2,1}(L) & A_{2,2}(L)
        \end{pmatrix}\begin{pmatrix} X_{t} \\ Z_{t} \end{pmatrix}
\end{align*}
Suppose the econometrician is interested in recovering the structural shocks, $\{\epsilon_{1,t}=X_{t}\}$, but omits from the empirical analysis the block $W_2$ (and so all the linear space spanned by the history of $Z$). The MA representation is fundamental only if $A_{1,2}(L)=0$ holds or, equivalently, whether there is Granger noncausality from the omitted variables $W_2$ to $X$. Vice versa, if $A_{1,2}(L)\not=0$, recovering the structural shock requires the enlarged information set $\{W_{1,t},W_{2,t}\}$ or, equivalently, $\{W_{1,t},Z_{t}\}$.

As the example makes clear, testing fundamentalness reduces to testing Granger non-causality or conditional lagged exogeneity \citep{giannone2006does, forni2014sufficient, plagborg2022instrumental, miranda2023identification}. Complementarily, \cite{chen2017testing} show that, when the DGP is a VARMA process generated by non-Gaussian i.i.d. shocks, fundamentalness holds if and only if the reduced-form innovations are m.d.s. \citep[][Theorem 1]{chen2017testing}. Maintaining this assumption, the shocks $X$ are fundamental or invertible if jointly: i) the structural shocks are non-Gaussian, ii) $\{X_{t}\}$ is m.d.s. with respect to its own past, $\sigma(X_{t-1},...)$, that is invertibility of $W_1$ alone, and iii) $\{X_{t}\}$ is m.d.s. with respect to the past of the omitted innovations, $\sigma(Z_{t-1},...)$. Taken together, these three conditions amount to the structural shocks $X$ having zero conditional mean given the past of both internal and external variables, $\sigma(X_{t-1},Z_{t-1},...)$,  which is the null of weak exogeneity of Eq.(\ref{h0grangernoncausality}).

\subsection{\cite{baker2016measuring}'s EPU Shocks}\label{subsec_EPU}

\cite{baker2016measuring} construct an index of Economic Policy Uncertainty (EPU) based on the frequency of newspaper articles that contain terms related to uncertainty, economy, and policy. Following their Section IV.D, the EPU structural shocks are estimated at monthly frequency by fitting a VAR(6) to five U.S. time series from Jan. 1985 to Dec. 2019, imposing a Cholesky ordering with the EPU index first, followed by the log S\&P 500, federal funds rate, log employment and log industrial production. As noted by \cite{diercks2024rains}, the estimated shock series is serially uncorrelated.\footnote{The shock series are provided in the replication package of \cite{diercks2024rains}.} Both the Lilliefors and Jarque--Bera tests reject the null hypothesis of Gaussianity at the 5\% level.

The small VAR system, however, arguably does not control for all relevant macroeconomic conditions. This turns out to be critical when questioning the exogeneity of the estimated shock, especially since the EPU index captures \textit{``uncertainties related to the economic ramifications of ``noneconomic'' policy matters [...] both near-term concerns [...] and longer term concerns''} \citep{baker2016measuring}. To assess these properties, this paper tests whether the EPU shock is weakly exogenous with respect to its own past and the past of omitted variables, using both the benchmark and the asymmetric Portmanteau statistics.

For omitted variables, this paper considers the first 8 principal components of large macroeconomic datasets, motivated by the intuition that, under a state space representation of the economic system, these estimated factors should approximate the relevant state variables \citep{forni2014sufficient}. Specifically, we consider \cite{mccracken2016fred}'s (McK Ng) macroeconomic factors from FRED-MD, a database of 134 monthly U.S. macroeconomic indicators.\footnote{Spanning until Jun. 2021, 8 static factors are estimated by PCA allowing for missing values via EM algorithm 
\citep{stock2002forecasting}.} 

\begin{figure}[h!]
\begin{subfigure}{\textwidth}
 \centering 
  \includegraphics[width=0.85\linewidth]{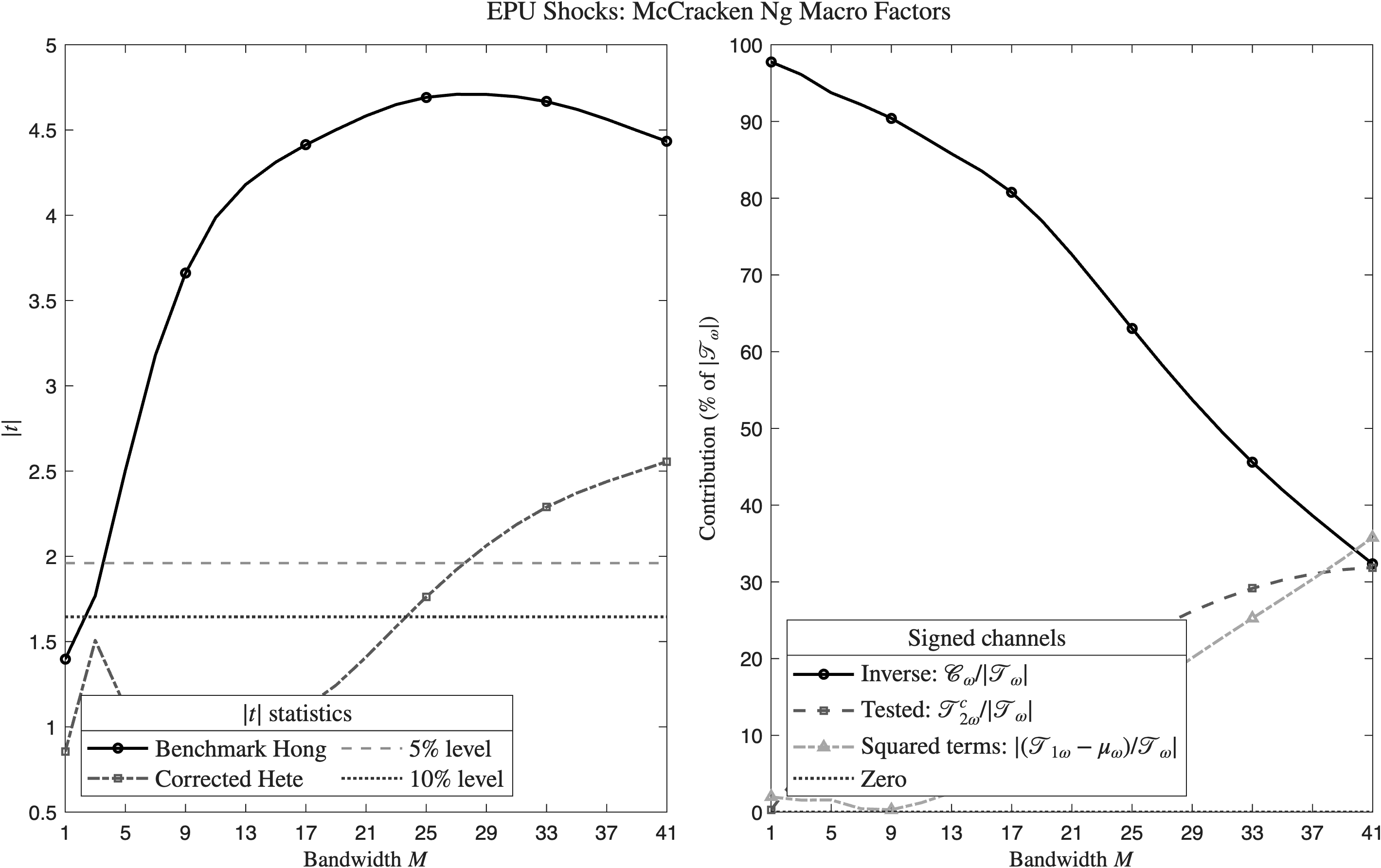}
\end{subfigure}
\caption[Empirical rates]{
				\cite{baker2016measuring}'s shocks and \cite{mccracken2016fred}'s factors.\\ Comparison between the two testing strategies.
                \normalsize
                 \textsc{Left panel}: y-axis reports the absolute value of the stats; x-axis lets the bandwidths range from $M=1$ to $41\sim$3Y+4M. The statistics are: Benchmark (\textit{Hong}, black solid and circles), Proposed (\textit{Hete}, grey dashed and squares). Horizontal dashed and dotted lines are the nominal 5\% and 10\% critical values, respectively. 
                \textsc{Right panel}: y-axis reports the share of each channel, as percentages of the absolute value of the centered Benchmark, $|\mathcal{T}_{\omega}-\mu_\omega|$, according to Eq.(\ref{corrected_statistic}): the inverse causality channel (black solid and circles), remaining tested channel (grey dashed and squares), and the absolute contribution of the centered sum of squares (grey dashed and triangles). Dotted horizontal line marks zero. All reported statistics are finite-sample versions and quadratic spectral kernel (see Appendix \ref{appendixC_statistics}). 
			}
\label{fig:EPU_McNg}
\end{figure}

The left panel of Figure \ref{fig:EPU_McNg} reports the statistics as the bandwidth
$M$ varies. Both procedures reject weak exogeneity at the 5\% level, but for
different reasons. The benchmark rejects at very short horizons ($M\leq2$). By contrast, the asymmetric Portmanteau statistic rejects only for bandwidths larger than about two years. Hence, while the benchmark suggests strong short-run evidence against weak exogeneity, the proposed statistic points to a rejection mainly driven by medium- and longer-run cycles.

The right panel of Figure \ref{fig:EPU_McNg} sheds light on this discrepancy by
decomposing the benchmark statistic as in Eq.~(\ref{corrected_statistic}), with
each component expressed as a share of the absolute value of the centered
benchmark statistic, $|T_\omega-\mu_\omega|$. At short horizons, the inverse
causality channel, $\mathcal{C}_\omega$, accounts for almost all the magnitude of
the benchmark statistic. Therefore, the short-horizon rejections are driven
primarily by dependence from past EPU shocks to current omitted factors,
rather than by the direction of weak exogeneity, namely dependence from
past omitted factors to current EPU shocks. As $M$ increases, the
contribution of the inverse channel declines, while the corrected channel becomes
relatively more important.

The conclusion is therefore not merely about finite-sample differences. Figure \ref{fig:EPU_McNg} shows that the benchmark rejections are largely due to inverse causality, whereas the rejections delivered by the asymmetric statistic are tied to the channel targeted by weak exogeneity. This distinction has practical consequences: by relying on the benchmark statistic alone, a practitioner could be misled into focusing on short-lag controls, while neglecting longer-lag dependencies. From this perspective, the decomposition in Eq.~(\ref{corrected_statistic}) is pivotal, since it prevents the short-horizon rejection from being incorrectly interpreted as evidence against weak exogeneity.  

In light of these results, the EPU shocks cannot be deemed fundamental or invertible when considering \cite{mccracken2016fred}'s factors. Impulse response analysis can therefore benefit from augmenting the system with such controls. To illustrate this point, I revisit \cite{diercks2024rains}'s analysis of the superadditive effects of uncertainty shocks, focusing on the impulse response of inflation and the stock market to EPU shocks.
\cite{diercks2024rains} estimate the following set of state-dependent local projections:
\begin{align*}
y_{t+h}=\alpha_h+\big(\beta_{0,h}+\beta_{1,h} \mathds{1}\{\epsilon_{unc,t-1}>0,..,\epsilon_{unc,t-L}>0\}\big) \epsilon_{unc,t}+\sum_{i=1}^{p}\gamma_{i,h}w_{t-i}+u_{t+h}    
\end{align*}
where $h$ sets the predictive horizon, ranging from 0 to 36 months, $p$ are the lags of the control variables, $\{w_t\}$, and the indicator function takes value 1 if each one of the previous $L$ shocks $\{\epsilon_{unc,t-1},..,\epsilon_{unc,t-L}\}$ has been positive. The state-multiplier coefficient $\beta_{1,h}$ captures the superadditive effect of uncertainty: the impact of a cascade of positive uncertainty shocks is more severe than the isolated sum of them. In their Appendix B.1 (Figure B.2), the shock of interest $\epsilon_{unc}$ is the EPU shock, the outcome variable $y$ is inflation, the number of consecutive positive shocks is 2 ($L=1$), the number of lags for the controls is 6, and the control set follows \cite{baker2016measuring}. Figure \ref{fig:DHT_inflation_original} reproduces their baseline results. 

Upon adding lags of \cite{mccracken2016fred}'s macroeconomic factors to the controls in the local projections, the unconditional linear impulse response remains largely unchanged (dashed vs. solid black lines in Figure \ref{fig:DHT_inflation_augmented}), while the state-dependent response, $\{\beta_{1,h}\}_{h=0,...,H}$, becomes positive and more pronounced as the number of factor lags increases. This pattern is consistent with the diagnostics in Figure \ref{fig:EPU_McNg}. Once longer-lag macroeconomic information is included,
the superadditive response of inflation to EPU shocks becomes clearer, strengthening \cite{diercks2024rains}'s conclusions: a sequence of consecutive positive uncertainty shocks leads to a marked increase of inflation. This finding connects to two strands of the literature. First, it aligns with \cite{fernandez2015fiscal}, who show that unexpected changes in fiscal policy uncertainty can lead to increased inflation within a standard New Keynesian model. Second, it relates closely to \cite{ascari2023endogenous}, who demonstrate that, in a rich DSGE model with firm dynamics, an `expectational' shock that raises short-term inflation expectations results in negative macroeconomic effects, causing inflation to rise while output declines.

Turning to the remaining variables of the system, while the impulse responses of industrial production and short rates are consistent with \cite{diercks2024rains} (not reported), the notable exception is the stock market. The state-dependent response of the real S\&P 500 to consecutive positive EPU shocks becomes significantly more negative at longer horizons (Figure \ref{fig:DHT_stockmkt_augmented}), consistent with \cite{berger2020uncertainty}'s finding that uncertainty shocks are linked to declines in stock returns (refer to their discussion in Section 5.2). 

A comparison across uncertainty shocks reveals an interesting heterogeneity. The inflation response to shocks from the other two considered measures (\cite{ludvigson2021uncertainty}'s and \cite{berger2020uncertainty}'s) is strongly negative, in contrast to the positive inflationary effect documented for EPU shocks. Meanwhile, the response of industrial production is negative across all three shock series, confirming their countercyclical nature \citep{diercks2024rains}. This pattern suggests that EPU shocks might operate primarily through a supply-side channel, whereas the financial uncertainty and realized volatility uncertainty shocks might be better characterized as demand-side disturbances.

\begin{figure}[h]
\centering
\begin{subfigure}{0.85\textwidth}
 \centering 
  \includegraphics[width=\linewidth]{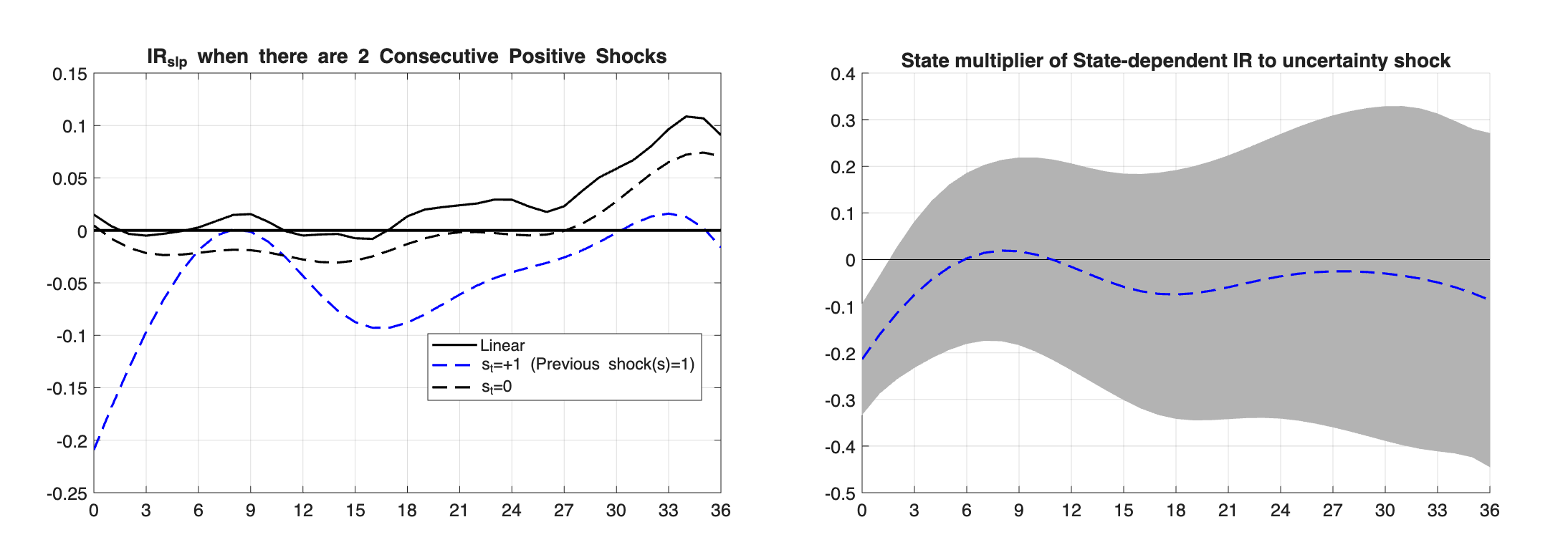}
  \caption{\normalsize{Figure B.2.(e)-(f) in  \cite{diercks2024rains}}}
  \label{fig:DHT_inflation_original}
\end{subfigure}
\par\medskip
\begin{subfigure}{0.85\textwidth}
 \centering
  \includegraphics[width=\linewidth]{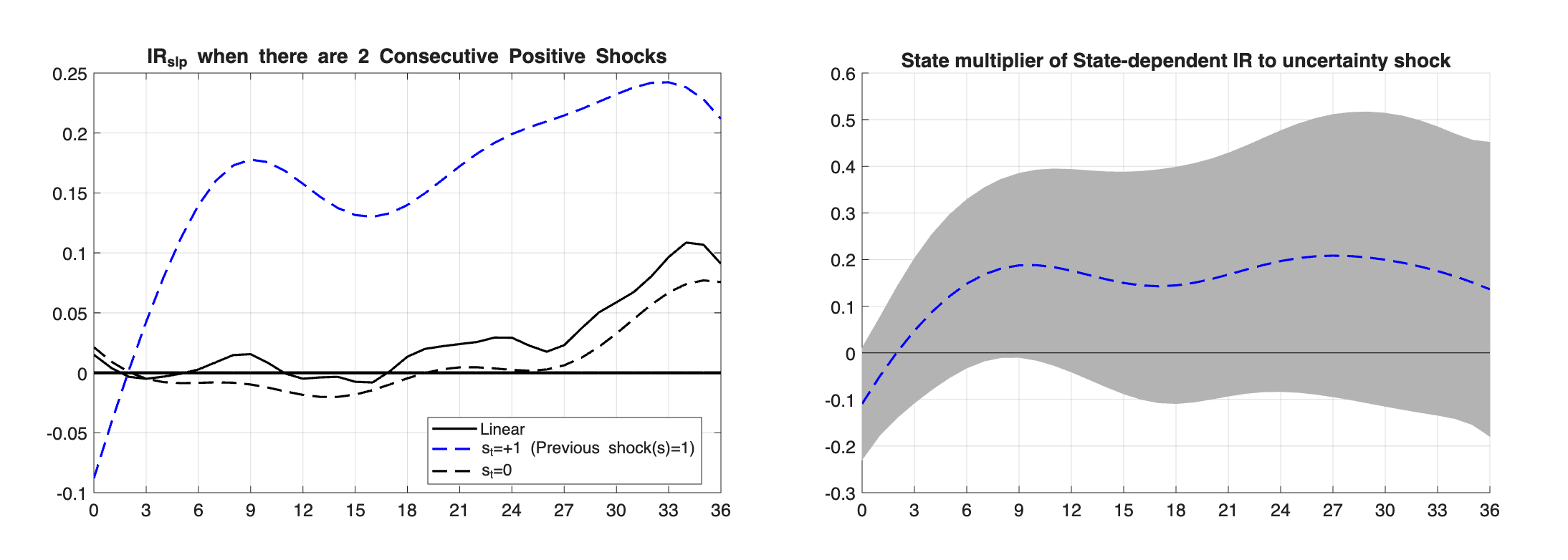}
  \caption{\normalsize{Controlling by one lag of  \cite{mccracken2016fred}'s factors}}
  \label{fig:s_power_corrected2}
\end{subfigure}
\par\medskip
\begin{subfigure}{0.85\textwidth}
 \centering
  \includegraphics[width=\linewidth]{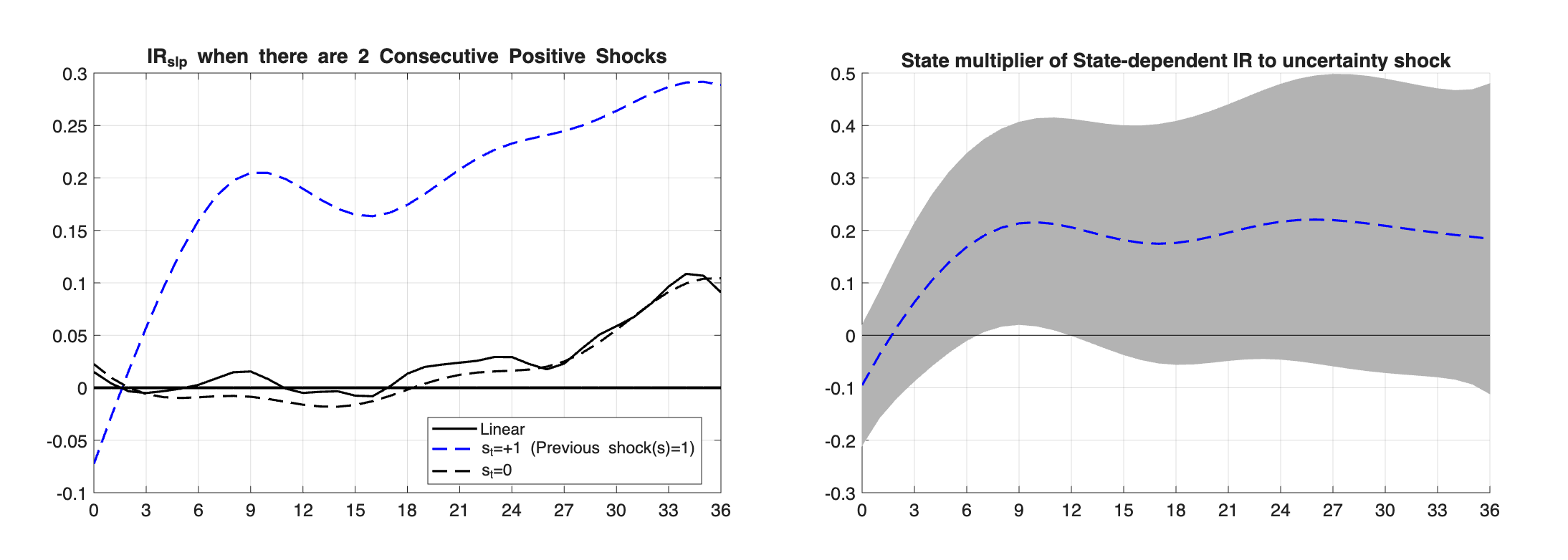}
  \caption{\normalsize{Controlling by three lags of  \cite{mccracken2016fred}'s factors}}
  \label{fig:s_power_corrected2}
\end{subfigure}
\caption[Empirical rates]{
				Response of price level to consecutive positive EPU uncertainty shocks: \\ \normalsize{ \textsc{Left panels}: empirical state-dependent impulse responses (estimated as in \citealp{diercks2024rains}) to two consecutive positive shocks (dashed blue line) and contrast it to the response to a single shock (dashed black line), and in the linear model (solid black line). \textsc{Right panels}: the incremental effect of the second shock, i.e. $\{\beta_{1,h}\}_{h=1,..,H}$, with 90\% confidence intervals (shaded area). The y-axes report the level of impulse responses, while the x-axes the horizons, $h$. Price level is by Personal Consumption Expenditures (PCE) Price Index.
			}}
\label{fig:DHT_inflation_augmented}
\end{figure}

\begin{figure}[h]
\centering
\begin{subfigure}{0.85\textwidth}
 \centering 
  \includegraphics[width=\linewidth]{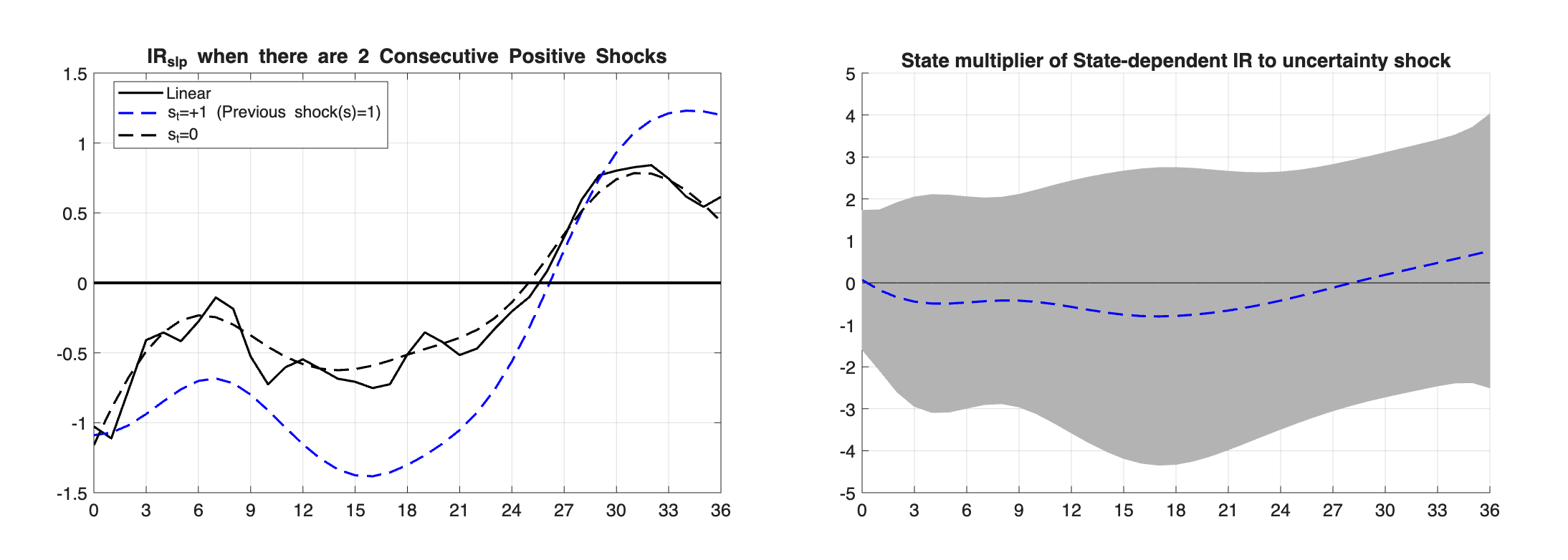}
  \caption{\normalsize{Figure B.3(e)-(f) in  \cite{diercks2024rains}}}
\end{subfigure}
\par\medskip
\begin{subfigure}{0.85\textwidth}
 \centering
  \includegraphics[width=\linewidth]{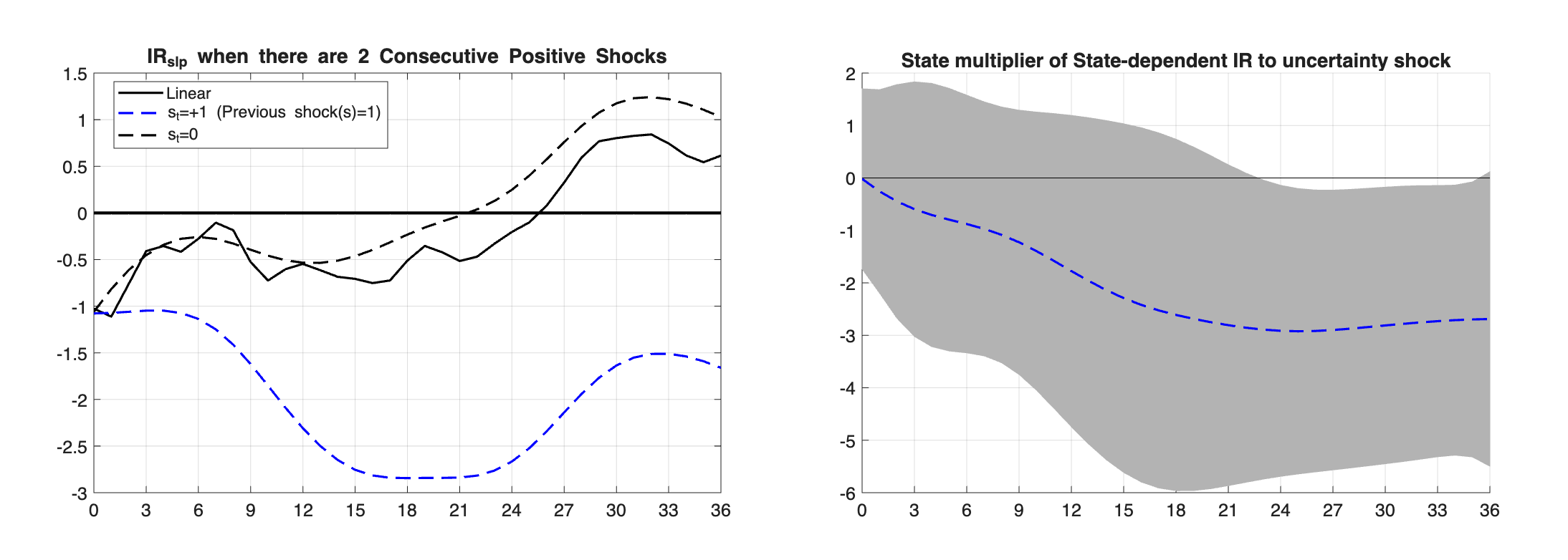}
  \caption{\normalsize{Controlling by one lag of \cite{mccracken2016fred}'s factors}}
\end{subfigure}
\par\medskip
\begin{subfigure}{0.85\textwidth}
 \centering
  \includegraphics[width=\linewidth]{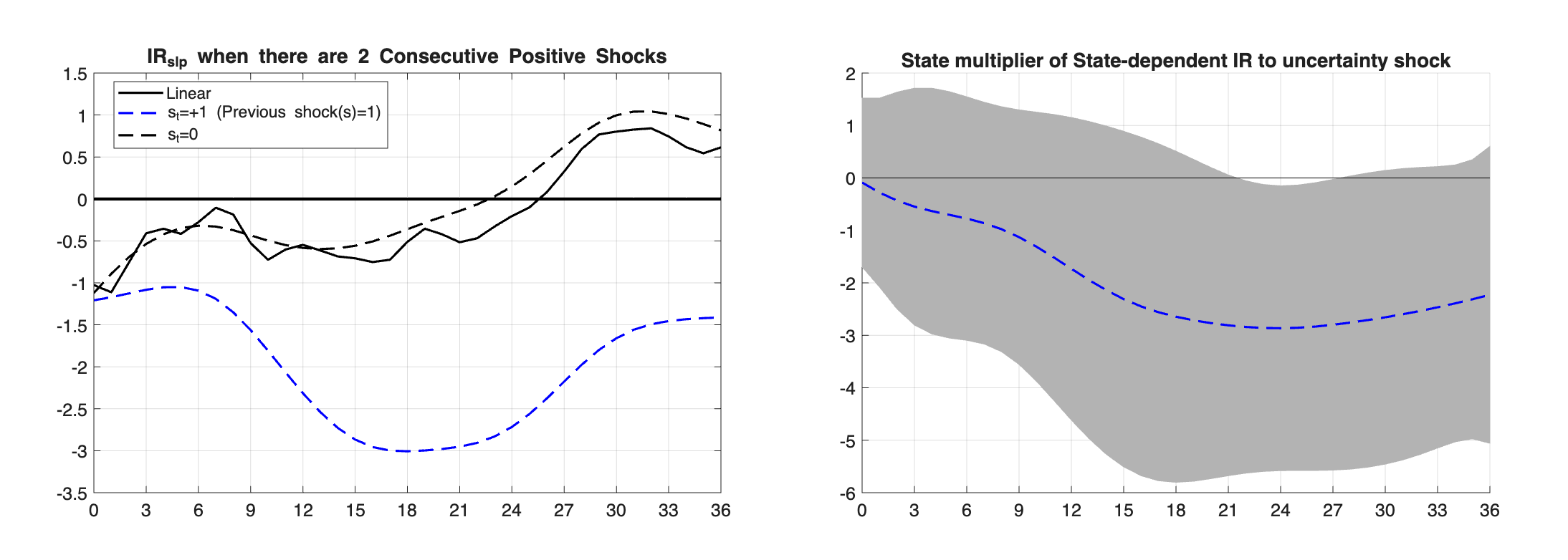}
  \caption{\normalsize{Controlling by three lags of \cite{mccracken2016fred}'s factors}}
\end{subfigure}
\caption[Empirical rates]{
				Response of stock market to consecutive positive EPU uncertainty shocks: \normalsize{ \textsc{Left panels}: empirical state-dependent impulse responses (estimated as in \citealp{diercks2024rains}) to two consecutive positive shocks (dashed blue line) and contrast it to the response to a single shock (dashed black line), and in the linear model (solid black line). \textsc{Right panels}: the incremental effect of the second shock, i.e. $\{\beta_{1,h}\}_{h=1,..,H}$, with 90\% confidence intervals (shaded area). The y-axes report the level of impulse responses, while the x-axes the horizons, $h$. Stock market is by real S\&P 500.
			}}
\label{fig:DHT_stockmkt_augmented}
\end{figure}

\clearpage

\section{Conclusion} \label{conclusion}
This paper studies specification testing in dynamic linear models in the presence of omitted variables, with its focus on testing weak exogeneity of structural shocks. Portmanteau statistics based on quadratic forms of serial cross-correlations might confound violations of weak exogeneity with dependence running from past shocks to current omitted variables (inverse causality), potentially producing misleading rejections.
To address this issue, this paper proposes an asymmetric Portmanteau statistic by introducing a correction term which removes the influence of inverse causality. The proposed statistic isolates the directional restriction implied by weak exogeneity without requiring parametric modeling of the joint dynamics. Under mild restrictions on the conditional moments, the asymmetric Portmanteau statistic is asymptotically normal under the null, and achieves asymptotic power comparable to the benchmark under fixed alternatives.
An empirical application revisits the exogeneity of a widely used shock series, \cite{baker2016measuring}'s Economic Policy Uncertainty shocks, and provides evidence against its (weak) exogeneity. By revisiting \cite{diercks2024rains}, enlarging the information set in light of the previous findings leads to a shift from negative to positive inflation (superadditive) response, together with contractionary effects elsewhere, thus suggesting a supply-side interpretation of the shock series.

\singlespacing
\bibliographystyle{ecta}
\bibliography{mybibfile}

@article{hong2001test,
  title={A test for volatility spillover with application to exchange rates},
  author={Hong, Yongmiao},
  journal={Journal of Econometrics},
  volume={103},
  number={1-2},
  pages={183--224},
  year={2001},
  publisher={Elsevier}
}

@article{brunnermeier2021feedbacks,
  title={Feedbacks: financial markets and economic activity},
  author={Brunnermeier, Markus and Palia, Darius and Sastry, Karthik A and Sims, Christopher A},
  journal={American Economic Review},
  volume={111},
  number={6},
  pages={1845--1879},
  year={2021},
  publisher={American Economic Association 2014 Broadway, Suite 305, Nashville, TN 37203}
}

@article{francq2007multivariate,
  title={Multivariate portmanteau test for autoregressive models with uncorrelated but nonindependent errors},
  author={Francq, Christian and Ra{\"\i}ssi, Hamdi},
  journal={Journal of Time Series Analysis},
  volume={28},
  number={3},
  pages={454--470},
  year={2007},
  publisher={Wiley Online Library}
}

@article{bouhaddioui2006generalized,
	title={A generalized portmanteau test for independence of two infinite-order vector autoregressive series},
	author={Bouhaddioui, Chafik and Roy, Roch},
	journal={Journal of Time Series Analysis},
	volume={27},
	number={4},
	pages={505--544},
	year={2006},
	publisher={Wiley Online Library}
}

@article{hong2009granger,
	title={Granger causality in risk and detection of extreme risk spillover between financial markets},
	author={Hong, Yongmiao and Liu, Yanhui and Wang, Shouyang},
	journal={Journal of Econometrics},
	volume={150},
	number={2},
	pages={271--287},
	year={2009},
	publisher={Elsevier}
}

@article{brown1971martingale,
	title={Martingale central limit theorems},
	author={Brown, Bruce M},
	journal={The Annals of Mathematical Statistics},
	pages={59--66},
	year={1971},
	publisher={JSTOR}
}

@book{hannan1970multiple,
	title={Multiple time series},
	author={Hannan, Edward James},
	volume={38},
	year={1970},
	publisher={John Wiley \& Sons}
}

@article{candelon2016nonparametric,
	title={A nonparametric test for granger causality in distribution with application to financial contagion},
	author={Candelon, Bertrand and Tokpavi, Sessi},
	journal={Journal of Business \& Economic Statistics},
	volume={34},
	number={2},
	pages={240--253},
	year={2016},
	publisher={Taylor \& Francis}
}

@book{priestley1981spectral,
	title={Spectral analysis and time series: Univariate series},
	author={Priestley, Maurice Bertram},
	volume={1},
	year={1981},
	publisher={Academic press}
}

@article{hong1996testing,
	title={Testing for independence between two covariance stationary time series},
	author={Hong, Yongmiao},
	journal={Biometrika},
	volume={83},
	number={3},
	pages={615--625},
	year={1996},
	publisher={Oxford University Press}
}

@article{haugh1976checking,
	title={Checking the independence of two covariance-stationary time series: a univariate residual cross-correlation approach},
	author={Haugh, Larry D},
	journal={Journal of the American Statistical Association},
	volume={71},
	number={354},
	pages={378--385},
	year={1976},
	publisher={Taylor \& Francis}
}

@article{hong2005generalized,
	title={Generalized spectral tests for conditional mean models in time series with conditional heteroscedasticity of unknown form},
	author={Hong, Yongmiao and Lee, Yoon-Jin},
	journal={The Review of Economic Studies},
	volume={72},
	number={2},
	pages={499--541},
	year={2005},
	publisher={Wiley-Blackwell}
}

@article{escanciano2006generalized,
	title={Generalized spectral tests for the martingale difference hypothesis},
	author={Escanciano, J Carlos and Velasco, Carlos},
	journal={Journal of Econometrics},
	volume={134},
	number={1},
	pages={151--185},
	year={2006},
	publisher={Elsevier}
}

@article{wang2022testing,
	title={Testing for the martingale difference hypothesis in multivariate time series models},
	author={Wang, Guochang and Zhu, Ke and Shao, Xiaofeng},
	journal={Journal of Business \& Economic Statistics},
	volume={40},
	number={3},
	pages={980--994},
	year={2022},
	publisher={Taylor \& Francis}
}

@article{hong1996consistent,
	title={Consistent testing for serial correlation of unknown form},
	author={Hong, Yongmiao},
	journal={Econometrica},
    volume={64},
	number={4},
	pages={837--864},
	year={1996},
	publisher={JSTOR}
}

@article{leong2023practical,
  title={A practical multivariate approach to testing volatility spillover},
  author={Leong, Soon Heng and Urga, Giovanni},
  journal={Journal of Economic Dynamics and Control},
  pages={104694},
  year={2023},
  publisher={Elsevier}
}

@book{dedecker2007weak,
  title={Weak dependence},
  author={Dedecker, J{\'e}r{\^o}me and Doukhan, Paul and Lang, Gabriel and Jos{\'e} Rafael, Le{\'o}n R and Louhichi, Sana and Prieur, Cl{\'e}mentine and Dedecker, J{\'e}r{\^o}me and Doukhan, Paul and Lang, Gabriel and Jos{\'e} Rafael, Le{\'o}n R and others},
  year={2007},
  publisher={Springer}
}

@article{li1981distribution,
  title={Distribution of the residual autocorrelations in multivariate ARMA time series models},
  author={Li, Wai Keung and McLeod, A Ian},
  journal={Journal of the Royal Statistical Society Series B: Statistical Methodology},
  volume={43},
  number={2},
  pages={231--239},
  year={1981},
  publisher={Oxford University Press}
}

@article{hosking1980multivariate,
  title={The multivariate portmanteau statistic},
  author={Hosking, Jonathan RM},
  journal={Journal of the American Statistical Association},
  volume={75},
  number={371},
  pages={602--608},
  year={1980},
  publisher={Taylor \& Francis}
}

@article{ramey2016macroeconomic,
  title={Macroeconomic shocks and their propagation},
  author={Ramey, Valerie A},
  journal={Handbook of macroeconomics},
  volume={2},
  pages={71--162},
  year={2016},
  publisher={Elsevier}
}

@article{baker2016measuring,
  title={Measuring economic policy uncertainty},
  author={Baker, Scott R and Bloom, Nicholas and Davis, Steven J},
  journal={The Quarterly Journal of Economics},
  volume={131},
  number={4},
  pages={1593--1636},
  year={2016},
  publisher={Oxford University Press}
}

@article{box1970distribution,
  title={Distribution of residual autocorrelations in autoregressive-integrated moving average time series models},
  author={Box, George EP and Pierce, David A},
  journal={Journal of the American Statistical Association},
  volume={65},
  number={332},
  pages={1509--1526},
  year={1970},
  publisher={Taylor \& Francis}
}

@article{durlauf1991spectral,
  title={Spectral based testing of the martingale hypothesis},
  author={Durlauf, Steven N},
  journal={Journal of Econometrics},
  volume={50},
  number={3},
  pages={355--376},
  year={1991},
  publisher={Elsevier}
}

@article{sims1980macroeconomics,
  title={Macroeconomics and reality},
  author={Sims, Christopher A},
  journal={Econometrica},
  volume={48},
  number={1},
  pages={1--48},
  year={1980},
  publisher={JSTOR}
}

@article{giannone2006does,
  title={Does information help recovering structural shocks from past observations?},
  author={Giannone, Domenico and Reichlin, Lucrezia},
  journal={Journal of the European Economic Association},
  volume={4},
  number={2-3},
  pages={455--465},
  year={2006},
  publisher={Oxford University Press}
}

@article{miranda2023identification,
  title={Identification with external instruments in structural VARs},
  author={Miranda-Agrippino, Silvia and Ricco, Giovanni},
  journal={Journal of Monetary Economics},
  volume={135},
  pages={1--19},
  year={2023},
  publisher={Elsevier}
}

@article{forni2014sufficient,
  title={Sufficient information in structural VARs},
  author={Forni, Mario and Gambetti, Luca},
  journal={Journal of Monetary Economics},
  volume={66},
  pages={124--136},
  year={2014},
  publisher={Elsevier}
}

@article{fernandez2007abcs,
  title={ABCs (and Ds) of understanding VARs},
  author={Fern{\'a}ndez-Villaverde, Jes{\'u}s and Rubio-Ram{\'\i}rez, Juan F and Sargent, Thomas J and Watson, Mark W},
  journal={American Economic Review},
  volume={97},
  number={3},
  pages={1021--1026},
  year={2007},
  publisher={American Economic Association}
}

@article{brockwell1987time,
  title={Time Series: Theory and Methods},
  author={Brockwell, Peter J and Davis, Richard A},
  journal={Springer Series in Statistics},
  year={1987},
  publisher={Springer New York}
}

@article{mccracken2016fred,
  title={FRED-MD: A monthly database for macroeconomic research},
  author={McCracken, Michael W and Ng, Serena},
  journal={Journal of Business \& Economic Statistics},
  volume={34},
  number={4},
  pages={574--589},
  year={2016},
  publisher={Taylor \& Francis}
}

@article{stock2002forecasting,
  title={Forecasting using principal components from a large number of predictors},
  author={Stock, James H and Watson, Mark W},
  journal={Journal of the American Statistical Association},
  volume={97},
  number={460},
  pages={1167--1179},
  year={2002},
  publisher={Taylor \& Francis}
}

@incollection{hansen2019two,
  title={Two difficulties in interpreting vector autoregressions},
  author={Hansen, Lars Peter and Sargent, Thomas J},
  booktitle={Rational expectations econometrics},
  pages={77--119},
  year={2019},
  publisher={CRC Press}
}

@article{lippi1994var,
  title={VAR analysis, nonfundamental representations, Blaschke matrices},
  author={Lippi, Marco and Reichlin, Lucrezia},
  journal={Journal of Econometrics},
  volume={63},
  number={1},
  pages={307--325},
  year={1994},
  publisher={Elsevier}
}

@article{chen2017testing,
  title={Testing for fundamental vector moving average representations},
  author={Chen, Bin and Choi, Jinho and Escanciano, Juan Carlos},
  journal={Quantitative Economics},
  volume={8},
  number={1},
  pages={149--180},
  year={2017},
  publisher={Wiley Online Library}
}

@article{diercks2024rains,
  title={When it rains it pours: Cascading uncertainty shocks},
  author={Diercks, Anthony M and Hsu, Alex and Tamoni, Andrea},
  journal={Journal of Political Economy},
  volume={132},
  number={2},
  pages={694--720},
  year={2024},
  publisher={The University of Chicago Press Chicago, IL}
}

@article{ludvigson2021uncertainty,
  title={Uncertainty and business cycles: exogenous impulse or endogenous response?},
  author={Ludvigson, Sydney C and Ma, Sai and Ng, Serena},
  journal={American Economic Journal: Macroeconomics},
  volume={13},
  number={4},
  pages={369--410},
  year={2021},
  publisher={American Economic Association 2014 Broadway, Suite 305, Nashville, TN 37203-2425}
}

@article{berger2020uncertainty,
  title={Uncertainty shocks as second-moment news shocks},
  author={Berger, David and Dew-Becker, Ian and Giglio, Stefano},
  journal={The Review of Economic Studies},
  volume={87},
  number={1},
  pages={40--76},
  year={2020},
  publisher={Oxford University Press}
}

@article{dufour1998short,
  title={Short run and long run causality in time series: theory},
  author={Dufour, Jean-Marie and Renault, Eric},
  journal={Econometrica},
  pages={1099--1125},
  year={1998},
  publisher={JSTOR}
}

@article{ljung1978measure,
  title={On a measure of lack of fit in time series models},
  author={Ljung, Greta M and Box, George EP},
  journal={Biometrika},
  volume={65},
  number={2},
  pages={297--303},
  year={1978},
  publisher={Oxford University Press}
}

@article{fernandez2015fiscal,
  title={Fiscal volatility shocks and economic activity},
  author={Fern{\'a}ndez-Villaverde, Jes{\'u}s and Guerr{\'o}n-Quintana, Pablo and Kuester, Keith and Rubio-Ram{\'\i}rez, Juan},
  journal={American Economic Review},
  volume={105},
  number={11},
  pages={3352--3384},
  year={2015},
  publisher={American Economic Association 2014 Broadway, Suite 305, Nashville, TN 37203}
}

@article{ascari2023endogenous,
  title={Endogenous uncertainty and the macroeconomic impact of shocks to inflation expectations},
  author={Ascari, Guido and Fasani, Stefano and Grazzini, Jakob and Rossi, Lorenza},
  journal={Journal of Monetary Economics},
  volume={140},
  pages={S48--S63},
  year={2023},
  publisher={Elsevier}
}

@book{lutkepohl1997handbook,
  title={Handbook of matrices},
  author={L{\"u}tkepohl, Helmut},
  year={1997},
  publisher={John Wiley \& Sons}
}

@article{jorda2023local,
  title={Local projections for applied economics},
  author={Jord{\`a}, {\`O}scar},
  journal={Annual Review of Economics},
  volume={15},
  number={1},
  pages={607--631},
  year={2023},
  publisher={Annual Reviews}
}

@article{plagborg2021local,
  title={Local projections and VARs estimate the same impulse responses},
  author={Plagborg-M{\o}ller, Mikkel and Wolf, Christian K},
  journal={Econometrica},
  volume={89},
  number={2},
  pages={955--980},
  year={2021},
  publisher={Wiley Online Library}
}

@book{kilian2017structural,
  title={Structural vector autoregressive analysis},
  author={Kilian, Lutz and L{\"u}tkepohl, Helmut},
  year={2017},
  publisher={Cambridge University Press}
}

@article{nakamura2018identification,
  title={Identification in macroeconomics},
  author={Nakamura, Emi and Steinsson, J{\'o}n},
  journal={Journal of Economic Perspectives},
  volume={32},
  number={3},
  pages={59--86},
  year={2018},
  publisher={American Economic Association 2014 Broadway, Suite 305, Nashville, TN 37203-2418}
}

@article{plagborg2022instrumental,
  title={Instrumental variable identification of dynamic variance decompositions},
  author={Plagborg-M{\o}ller, Mikkel and Wolf, Christian K},
  journal={Journal of Political Economy},
  volume={130},
  number={8},
  pages={2164--2202},
  year={2022},
  publisher={The University of Chicago Press Chicago, IL}
}

@article{mikusheva2025linear,
  title={Linear regression with weak exogeneity},
  author={Mikusheva, Anna and S{\o}lvsten, Mikkel},
  journal={Quantitative Economics},
  volume={16},
  number={2},
  pages={367--403},
  year={2025},
  publisher={Wiley Online Library}
}

@article{lobato2002testing,
  title={Testing for zero autocorrelation in the presence of statistical dependence},
  author={Lobato, Ignacio N and Nankervis, John C and Savin, N Eugene},
  journal={Econometric Theory},
  volume={18},
  number={3},
  pages={730--743},
  year={2002},
  publisher={Cambridge University Press}
}

@article{hafner2022identification,
  title={Identification of structural multivariate GARCH models},
  author={Hafner, Christian M and Herwartz, Helmut and Maxand, Simone},
  journal={Journal of Econometrics},
  volume={227},
  number={1},
  pages={212--227},
  year={2022},
  publisher={Elsevier}
}

@article{bruns2024testing,
  title={Testing for strong exogeneity in Proxy-VARs},
  author={Bruns, Martin and Keweloh, Sascha A},
  journal={Journal of Econometrics},
  volume={245},
  number={1-2},
  pages={105876},
  year={2024},
  publisher={Elsevier}
}

@article{dominguez2020specification,
  title={Specification testing with estimated variables},
  author={Dom{\'\i}nguez, Manuel A and Lobato, Ignacio N},
  journal={Econometric Reviews},
  volume={39},
  number={5},
  pages={476--494},
  year={2020},
  publisher={Taylor \& Francis}
}

@article{kuan2004new,
  title={A New Test of the Martingale Difference Hypothesis.},
  author={Kuan, Chung-Ming and Lee, Wei-Ming},
  journal={Studies in Nonlinear Dynamics \& Econometrics},
  volume={8},
  number={4},
  year={2004}
}

@article{dominguez2004consistent,
  title={Consistent estimation of models defined by conditional moment restrictions},
  author={Dom{\'\i}nguez, Manuel A and Lobato, Ignacio N},
  journal={Econometrica},
  volume={72},
  number={5},
  pages={1601--1615},
  year={2004},
  publisher={Wiley Online Library}
}

@book{lutkepohl2013introduction,
  title={Introduction to multiple time series analysis},
  author={L{\"u}tkepohl, Helmut},
  year={2013},
  publisher={Springer Science \& Business Media}
}

@book{hall2014martingale,
  title={Martingale limit theory and its application},
  author={Hall, Peter and Heyde, Christopher C},
  year={2014},
  publisher={Academic press}
}

@article{lazarus2021size,
  title={The Size-Power Tradeoff in HAR Inference},
  author={Lazarus, Eben and Lewis, Daniel J and Stock, James H},
  journal={Econometrica},
  volume={89},
  number={5},
  pages={2497--2516},
  year={2021},
  publisher={Wiley Online Library}
}

@article{lazarus2018har,
  title={HAR inference: Recommendations for practice},
  author={Lazarus, Eben and Lewis, Daniel J and Stock, James H and Watson, Mark W},
  journal={Journal of Business \& Economic Statistics},
  volume={36},
  number={4},
  pages={541--559},
  year={2018},
  publisher={Taylor \& Francis}
}

@article{andrews1991heteroskedasticity,
  title={Heteroskedasticity and autocorrelation consistent covariance matrix estimation},
  author={Andrews, Donald WK},
  journal={Econometrica},
  volume={59},
  number={3},
  pages={817--858},
  year={1991},
  publisher={JSTOR}
}

\clearpage
\onehalfspacing
\appendix
\appendixnumbering
\setcounter{page}{1}

\section{Appendix}\label{appendixA}

\subsection{Proof of Proposition \ref{prop_toyexample1}} \label{proof_prop_toyexample1}

To show that the variance of test statistic, $\mathcal{T}_{\omega}$, depends on the inverse causality, it is sufficient to study the first two moments of $Q(j)$. By the conditional homoskedasticity of the process $X$ and the LIE, we have: $ \mathbb{E}[||X_t ||^2||Z_{t-j} ||^2]=d_1 d_2,$ and $\mathbb{E}[(||X_t ||^2||Z_{t-j} ||^2)^2]=\kappa_1\kappa_2 $.
This means that, for a fixed $j\geq 1$, the squared products that characterizes the first term of the test statistic, $\mathcal{T}_{1\omega}$, are not influenced by the dependence running from past $X$ to present $Z$ (inverse causality).\\
To study the moments of the cross-products, WLOG, presum $t> s$. We have:
\begin{align*}
&\mathbb{E}[(<X_t,X_s><Z_{t-j},Z_{s-j}>)]=\mathbb{E}[X_t^\prime]\mathbb{E}[X_s<Z_{t-j},Z_{s-j}>)]=0\\
&\mathbb{E}[(<X_t,X_s><Z_{t-j},Z_{s-j}>)^2]= \mathbb{E}[ ||X_s ||^2<Z_{t-j},Z_{s-j}>^2]
\end{align*}
where the first equality is due the null in Eq.(\ref{h0grangernoncausality}), and the third equality is because of the conditional homoskedasticity. Without additional assumptions, only when $s> t-j$, we have by LIE: 
\begin{align*}
\mathbb{E}[ ||X_s ||^2<Z_{t-j},Z_{s-j}>^2]=\mathbb{E}[||X_s ||^2]\mathbb{E}[<Z_{t-j},Z_{s-j}>^2]=d_1 d_2, \quad s> t-j
\end{align*}
which implies that, for a fixed $j\geq 1$, the cross-products that characterizes the second term of the test statistic, $\mathcal{T}_{2\omega}$, are independent with respect to the inverse causality. Under mutual independence of the two processes, the same identity holds for all relevant time orderings: $\mathbb{E}[||X_s ||^2<Z_{t-j},Z_{s-j}>^2]=d_1d_2$

\subsection{Inverse Causality in Conditional Heteroskedastic DGPs}\label{lemma_inversecausality}
\begin{lemma} \label{coroll_toyexample1}
Consider a measurable real functions, $h(\cdot): \mathbb{R}\rightarrow \mathbb{R}$, and some $|\alpha|\in (0,1)$. Let $\{X_t,Z_t\}$ be a bivariate mean-zero process characterized by:
\begin{align*}
Z_t^2=\alpha Z_{t-1}^2+h(X_{t-1})+u_t, \; \; \; X_{t}\sim i.i.d.(0,1), \; \; \; u_{t}\sim i.i.d.(1,\sigma_u^2), \; \; \; X_t \perp\!\!\!\perp u_s, \, \forall s,t
\end{align*}
Denote: $\varrho_h=\mathbb{E}[h(X_{s})X_s^2]$, $\mu_{e}=\mathbb{E}[h(X_{t-1})+u_t]$, $\mu_{z}=\mathbb{E}[Z_t^2]$, and $\sigma^2_{z}=\text{Var}[Z_t^2]$.\\
We have: $\text{Var}[\mathcal{T}_{2\omega}]=\frac{1}{T^4}\sum_{j=1}^{T-2}\omega^2(j) \big(\Sigma_2(j)+ \Delta_2(j)\big)$, with:
\begin{align*}
    &\Sigma_2(j)=
    \sum_{\upsilon_1=1}^{\tau(j)}
    \left(\alpha^{|j-\upsilon_1|}\sigma_z^2+\mu_z^2\right)+ \sum_{\upsilon_2=0}^{\Upsilon(j)} \left(\alpha^{j+\upsilon_2}\sigma_z^2+\alpha^{\upsilon_2+1}\mu_z^2+\sum_{l=0}^{\upsilon_2}\alpha^l \mu_z\right)\\
    & \Delta_2(j)=\sum_{\upsilon_2=0}^{\Upsilon(j)} \left(\mu_z\mu_e(1+\alpha^{\upsilon_2}-\mathds{1}[\upsilon_2=1])-\sum_{l=0}^{\upsilon_2}\alpha^l \mu_z\right)+ \sum_{\upsilon_2=0}^{\Upsilon(j)} \alpha^{\upsilon_2-1}\mu_z(1+\varrho_h)
\end{align*}
 where: i) $\tau(j)$ is the number of times that $s>t-j$ up to $T$, at a given $j\geq1$, such that $s\not=t$, and $s,t=j+1$; ii) $\Upsilon(j)$ is the number of times that $s\leq t-j$, with respect to the same conditions. In particular:  $\Upsilon(j)=(T-j)(T-j-1)-\tau(j)$, and $\tau(j)=\mathds{1}[j<T/2](T^2-2j^2-3T+4j)/2+\mathds{1}[j\geq T/2] (T-j)(T-j-1)$.
\end{lemma}
\begin{proof}
Denote: $\mu_h=\mathbb{E}[h(X_{s})]$ and $\varrho_h=\mathbb{E}[h(X_{s})X_s^2]$,  $\mu_{e}=\mathbb{E}[h(X_{s-1})+u_s]=\mu_h+1$.
Denote $\sigma_{e}^2=\mathbb{E}[(h(X_{t-1})+u_t)^2]$, and $\mu_{z}=\mu_e/(1-\alpha), \; \sigma^2_{z}=(\sigma^2_e-\mu_e^2)/(1-\alpha^2)$. Notice that we have: $\mathbb{E}[X_t|\mathcal{I}(t-1)]=0, \;\mathbb{E}[X_t^2|\mathcal{I}(t-1)]=1$. By the same logic of Proposition \ref{prop_toyexample1}, we have:
\begin{align*}
    &\mathbb{E}[X_tX_sZ_{t-j}Z_{s-j}]=0\\
    &\mathbb{E}[X_t^2X_s^2Z_{t-j}^2Z_{s-j}^2]=\mathbb{E}[X_s^2Z_{t-j}^2Z_{s-j}^2] \\
    &\mathbb{E}[(X_t,X_sZ_{t-j},Z_{s-j})(X_{t+h},X_{s+h}Z_{t-j+h},Z_{s-j+h})]= 0, \; \; \; \forall h \not=0
\end{align*}
Note that the process $\{Z_{t}^2\}$ is a causal AR(1) process with i.i.d. noise, since we have: $Z_{t-1} \perp\!\!\!\perp X_{t-1}, \; u_t \perp\!\!\!\perp X_{t-1}$.
Notice that, for $h\geq0$:
\begin{align*}
    \mathbb{E}[Z_{t}^2Z_{t-h}^2]= \alpha^h \dfrac{\sigma^2_e-\mu_e^2}{1-\alpha^2}  +\left(\dfrac{\mu_e}{1-\alpha}\right)^2=\alpha^h \dfrac{\sigma^2_e}{1-\alpha^2}+\mu_e^2\left(\dfrac{1}{(1-\alpha)^2}-\alpha^h\right)
\end{align*}
Consider the two scenarios: $s> t-j$ and $s\leq t-j$.\\
First, when $s> t-j$, we have: $\mathbb{E}[X_s^2Z_{t-j}^2Z_{s-j}^2]=\mathbb{E}[X_s^2]\mathbb{E}[Z_{t-j}^2Z_{s-j}^2]=\alpha^{|j-\upsilon_1|} \sigma_z^2+\mu_z^2$, with: $\upsilon_1=s-t+j>0$. Indeed, by looking at $s=t-j+1$ (i.e., $\upsilon_1=1$), the following holds: $\mathbb{E}[Z_{t-j}^2Z_{s-j}^2]=\mathbb{E}[Z_{s-1}^2Z_{s-j}^2]= \alpha^{|j-1|}\sigma_z^2+\mu_z^2$, since:
\begin{align*}
    \mathbb{E}[Z_{s-1}^2Z_{s-j}^2]&=\sigma_z^2+\mu_z^2,  \qquad  \text{for} \, j=1 \\
    \cdots \\
    \mathbb{E}[Z_{s-1}^2Z_{s-j}^2]&=\mathbb{E}[Z_{s-1}^2Z_{m+1}^2]=\alpha^{m}\sigma_z^2+\mu_z^2, \quad  \text{for} \, j=m+1 
\end{align*}    
Summing up the terms, we have: $\sum_{s > t-j, s\not=t}^T  \mathbb{E}[X_s^2Z_{t-j}^2Z_{s-j}^2]= \sum_{\upsilon_1=1}^{\tau(j)}\left( \alpha^{|j-\upsilon_1|} \sigma_z^2+\mu_Z^2\right)$, 
where $\tau(j)$ is the number of times that $s> t-j$ up to $T$, at a given $j>0$, such that: $s\not=t$, $s,t=j+1$. In particular, we have:
\begin{align*}
 \tau(j)&=\dfrac{(T-j)(T-j-1)}{2} +\mathds{1}(2j-T<0)\cdot\big((T-2j)(j-1)+j(j-1)/2\big)\\
 &+\mathds{1}(2j-T\geq0)\cdot(T-j)(T-j-1)/2 
\end{align*}
Second, when $s\leq t-j$, define $\upsilon_2=t-s-j\geq 0$. In this case, we have: 
\begin{align*}
    \mathbb{E}[X_s^2 Z_{t-j}^2 Z_{s-j}^2]
    &=\alpha^{\upsilon_2+1}\left(\alpha^{j-1}\sigma_z^2+\mu_z^2\right)
    +\mu_z\mu_e(1+\alpha^{\upsilon_2})
    -\mu_z\mu_e\,\mathds{1}[\upsilon_2=1]  \\
    &\qquad
    +\alpha^{\upsilon_2-1}\mu_z(1+\varrho_h)\,\mathds{1}[\upsilon_2>0]
\end{align*}
Indeed by looking at $s=t-j$ (i.e., $\upsilon_2=0$), the following holds:
\begin{align*}
    \mathbb{E}[Z_s^2 X_s^2 Z_{s-j}^2]
    &=\alpha \mathbb{E}[Z_{s-1}^2 X_s^2 Z_{s-j}^2]
    +\mathbb{E}[u_s X_s^2 Z_{s-j}^2]
    +\mathbb{E}[h(X_{s-1}) X_s^2 Z_{s-j}^2] \\
    &=\alpha \mathbb{E}[Z_{s-1}^2 Z_{s-j}^2]
    +\mathbb{E}[Z_{s-j}^2]
    +\mathbb{E}[h(X_{s-1})]\mathbb{E}[X_s^2]\mathbb{E}[Z_{s-j}^2] \\
    &=\alpha\left(\alpha^{j-1}\sigma_z^2+\mu_z^2\right)
    +\mu_z(1+\mu_h).
\end{align*}
and by looking at $s=t-j-1$ (i.e., $\upsilon_2=1$), the following holds: $ \mathbb{E}[Z_{s+1}^2 X_s^2 Z_{s-j}^2]=\alpha^2\left(\alpha^{j-1}\sigma_z^2+\mu_z^2\right)+\alpha\mu_z(1+\mu_h)+\mu_z(1+\varrho_h)$, and the conclusion is reached by recursiveness.\\
Reconciling both scenarios, we have that $T^4\text{Var}(\mathcal{T}_{2\omega})$ is equal to:
\begin{align*}
    &\sum_{j=1}^{T-2}\omega^2(j)
    \left(\sum_{\upsilon_1=1}^{\tau(j)}
    \left(\alpha^{|j-\upsilon_1|}\sigma_z^2+\mu_z^2\right)+ \sum_{\upsilon_2=0}^{(T-j)(T-j-1)-\tau(j)}\left(\alpha^{\upsilon_2+j}\sigma_z^2+\alpha^{\upsilon_2+1}\mu_z^2\right)\right) \\
    &\quad
    +\sum_{j=1}^{T-2}\omega^2(j)
    \sum_{\upsilon_2=0}^{(T-j)(T-j-1)-\tau(j)}
    \Big(\mu_z\mu_e(1+\alpha^{\upsilon_2})
    -\mu_z\mu_e\mathds{1}[\upsilon_2=1]
    \Big) \\
    &\quad
    +\sum_{j=1}^{T-2}\omega^2(j)
    \sum_{\upsilon_2=1}^{(T-j)(T-j-1)-\tau(j)}
    \alpha^{\upsilon_2-1}\mu_z(1+\varrho_h).
\end{align*}
When $h(\cdot)=0$, meaning $Z_{t}^2=\alpha Z_{t-1}^2+u_t$, we have rather that $T^4\text{Var}(\mathcal{T}_{2\omega})$ is:  
\begin{align*}
    &\sum_{j=1}^{T-2}\omega^2(j)
    \left(\sum_{\upsilon_1=1}^{\tau(j)}
    \left(\alpha^{|j-\upsilon_1|}\sigma_z^2+\mu_z^2\right)+ \sum_{\upsilon_2=0}^{(T-j)(T-j-1)-\tau(j)}\left( \alpha^{\upsilon_2+j}\sigma_z^2+\alpha^{\upsilon_2+1}\mu_z^2\right)\right) \\
    &\qquad
    +\sum_{j=1}^{T-2}\omega^2(j)
    \sum_{\upsilon_2=0}^{(T-j)(T-j-1)-\tau(j)}
    \Big(\sum_{l=0}^{\upsilon_2}\alpha^{l}\mu_z
    \Big) :=\sum_{j=1}^{T-2}\omega^2(j)\Sigma_2(j)
\end{align*} 
so that we finally write: $T^4\text{Var}(\mathcal{T}_{2\omega})-\sum_{j=1}^{T-2}\omega^2(j)\Sigma_2(j)=\sum_{j=1}^{T-2}\omega^2(j)\Delta_2(j)$.
\end{proof}

\section{Appendix B: Proof of Proposition \ref{prop_moments}, Theorem \ref{theorem1} to \ref{theorem_power}} \label{appendixB}

\subsection{Proof of Proposition \ref{prop_moments}} \label{proof_prop_moments}
The proof consists of four parts.

First, from its definition in eq.(\ref{decompositionsintosums}): $\mathbb{E}[T\cdot \mathcal{T}_{1\omega}]=\frac{1}{T} \sum_{j=1}^{T-1} \omega(j) \sum_{t=j+1}^T \mathbb{E}[||X_t ||^2||Z_{t-j} ||^2]$.
Thus, by LIE: $\mathbb{E}[||X_t ||^2||Z_{t-j} ||^2]=\mathbb{E}[\mathbb{E}[||X_t ||^2|\mathcal{I}(t-1)]|| Z_{t-j} ||^2]=d_1d_2$, because of the conditional homoskedasticity of $X$ and the processes being standardized. Substituting, we have: $ \mathbb{E}[T \cdot \mathcal{T}_{1\omega}]=\frac{1}{T} \sum_{j=1}^{T-1} \omega(j) (T-j) d_1d_2=\mu_{\omega,T}$.

Second, by definition: $\text{Var}[T \cdot \mathcal{T}_{1\omega}]=\mathbb{E}[(T \cdot \mathcal{T}_{1\omega}-\mu_{\omega,T})^2]$. For simplicity, denote:
\begin{align*}
    T \cdot\mathcal{T}_{1\omega}-\mu_{\omega,T}=\frac{1}{T} \sum_{j=1}^{T-1} \omega(j) \sum_{t=j+1}^T \left(|| X_t ||^2|| Z_{t-j} ||^2-d_1d_2\right)=\frac{1}{T} \sum_{j=1}^{T-1} \omega(j) \sum_{t=j+1}^T \Upsilon_{t,j}^{(1)}
\end{align*}
Define: $Y_{1,t}=|| X_t ||^2-d_1$, and $Y_{2,t}=|| Z_t ||^2-d_2$, so that: $\Upsilon_{t,j}^{(1)}=Y_{1,t}|| Z_{t-j} ||^2+d_{1} Y_{2,t-j}$.
Hence, we have: $\Big\|\sum_{t=j+1}^T \Upsilon_{t,j}^{(1)}\Big\|_2^2\leq 2\Big\|\sum_{t=j+1}^T Y_{1,t}|| Z_{t-j} ||^2\Big\|_2^2+2d_1^2\Big\|\sum_{t=j+1}^T Y_{2,t-j}\Big\|_2^2$.
Let us consider the two summations.
\begin{itemize}
    \item Regarding the first term, under the assumption of conditional homoskedasticity, we have that: $\mathbb{E}[Y_{1,t}|| Z_{t-j}||^2|\mathcal{I}(t-1)]=0$, so that $\{Y_{1,t}|| Z_{t-j}||^2, \mathcal{I}(t-1)\}$ constitutes a martingale difference sequence. This means we can write:
    \begin{align*}
        \Big\|\sum_{t=j+1}^T Y_{1,t}|| Z_{t-j} ||^2\Big\|_2^2=\sum_{t=j+1}^T \mathbb{E}[||Y_{1,t}||^2|| Z_{t-j} ||^4]
    \end{align*}
    By the assumptions on the conditional moments and LIE: \\ $\mathbb{E}[|| X_t ||^4|| Z_{t-j} ||^4]=\mathbb{E}[||X_t ||^4]\mathbb{E}[||Z_{t-j} ||^4]$, which means that, together with the finiteness of fourth moments, we have: \\ $\sup_{t}\mathbb{E}[||Y_{1,t}||^2|| Z_{t-j} ||^4]<\infty$. \\
    Hence, we conclude: $\Big\|\sum_{t=j+1}^T Y_{1,t}|| Z_{t-j} ||^2\Big\|_2^2=O(T)$. 
    \item Regarding the second term, by stationarity, $\text{Var}\left(\sum_{t=j+1}^T Y_{2,t-j}\right)
    =\text{Var}\left(\sum_{t=1}^{T-j} Y_{2,t}\right)$.\\
    Let us denote $\gamma_{Y_2}(h):=\text{Cov}(Y_2,Y_{2+h})$, then:
    \begin{align*}
    \text{Var}\left(\sum_{t=1}^{T-j} Y_{2,t}\right)
    &=
    (T-j)\gamma_{Y_2}(0)+2\sum_{h=1}^{T-j-1}(T-j-h)\gamma_{Y_2}(h) \nonumber\\
    &\le
    (T-j)\gamma_{Y_2}(0)+2(T-j)\sum_{h=1}^{\infty}|\gamma_{Y_2}(h)|
    \end{align*}
    Under the assumption $|\text{Cov}[|| Z_{1}||^2,|| Z_{1+h}||^2]|=O(h^{-1-\epsilon})$, we have a standard argument for the absolute summability of the covariances, which in turn implies: $ \Big\|\sum_{t=j+1}^T Y_{2,t-j}\Big\|_2^2 = O(T)$.
\end{itemize}
Thus, we have the following: $\Big\|\sum_{t=j+1}^T \Upsilon_{t,j}^{(1)}\Big\|_2=O(T^{1/2})$, so that the variance can be bounded by virtue of Minkowski inequality:
\begin{align*}
\text{Var}[T\mathcal{T}_{1\omega}]&=\frac{1}{T^2}\left\vert\left\vert \sum_{j=1}^{T-1} \omega(j) \sum_{t=j+1}^T \Upsilon_{t,j}^{(1)}\right\vert\right\vert^2_{2}\leq \frac{1}{T^2} \left(\sum_{j=1}^{T-1} \omega(j) \left\vert\left\vert \sum_{t=j+1}^T \Upsilon_{t,j}^{(1)}\right\vert\right\vert_{2} \right)^2 
\end{align*}
By the previous results: $\text{Var}[T\mathcal{T}_{1\omega}]\leq \dfrac{\Delta}{T}\left(\sum_{j=1}^{T-1} \omega(j)  \right)^2 = \dfrac{\Delta}{T}M^2\left(M^{-1}\sum_{j=1}^{T-1} \omega(j)  \right)^2 = O(M^2/T)$
for finite $\Delta>0$, where the last equality comes by realizing that, under Assumption \ref{assumption_kernel1}, the following holds asymptotically: $M^{-1}\sum_{j=1}^{T-1} \omega(j) \rightarrow \int_{0}^{\infty}k^{2}(z)dz < \infty$.
By the same reason as above, under Assumption \ref{assumption_kernel1}, the following convergence holds asymptotically: $M^{-1} D_{\omega,T} \rightarrow \frac{1}{2}\int_{0}^{\infty}k^{4}(z)dz < \infty$, so that: $D_{\omega,T} = O(M)$.
Given the assumptions on the higher order moments of $Z$ (finiteness and covariance decay), and the result in Lemma \ref{coroll_toyexample1}, we also have that: $D_{\omega,T}^{(Hete)}= O(M)$. Together with the first and second findings, we have: 
\begin{align*}
\mathbb{E}\left[\dfrac{T \cdot\mathcal{T}_{1\omega}-\mu_{\omega,T}}{\sqrt{D_{\omega,T}}}\right]=0, \; \; \; \text{Var}\left[\dfrac{T \cdot\mathcal{T}_{1\omega}-\mu_{\omega,T}}{\sqrt{D_{\omega,T}}}\right]=O(M/T)
\end{align*}
so that, under the asymptotics of Assumption \ref{assumption_kernel1}, when $M/T\rightarrow0$ as $T\rightarrow\infty$, we have $\ell_2$ (or MSE) convergence to zero.

Third, from its definition in eq.(\ref{corrected_statistic}): 
\begin{align*}
\mathbb{E}[T \cdot\mathcal{T}_{2\omega}^c]&=\frac{2}{T} \sum_{t=2}^{T} \sum_{s=2}^{t-1} \sum_{j=t-s+1}^{s-1} \omega(j)  \mathbb{E}[<X_t,X_s><Z_{t-j},Z_{s-j}>] 
\end{align*}
Under the null hypothesis in eq.(\ref{h0grangernoncausality}) and by LIE: $\mathbb{E}[\langle X_t,X_s\rangle\langle Z_{t-j},Z_{s-j}\rangle]=\mathbb{E}[\mathbb{E}[X_t^\prime |\mathcal{I}(t-1)]X_s<Z_{t-j},Z_{s-j}>]=0$, as the indexes are such that $t>s$. Thus: $\mathbb{E}[T \cdot \mathcal{T}_{2\omega}^c]=0$.

Fourth, define appropriately: $ J_t:=\sum_{s=2}^{t-1}\ \sum_{j\in\mathfrak{J}_{ts}} \omega(j)\,\langle X_t,X_s\rangle\langle Z_{t-j},Z_{s-j}\rangle, $ so that: $T\cdot\mathcal T_{2\omega}^c=\frac{2}{T}\sum_{t=2}^T J_t$.
The process $\{(J_t,\mathcal{I}(t-1)), t\in \mathbb{Z}^+\}$ constitutes a martingale difference sequence, since: 1) we have $\mathbb{E}[J_t|\mathcal{I}(t-1)]=0$ under the null hypothesis of eq.(\ref{h0grangernoncausality}); 2) $\mathbb{E}[|J_t|]<\infty$ by the finiteness of the moments. Hence: $ \text{Var}[T\cdot\mathcal T_{2\omega}^c]=\frac{4}{T^2}\sum_{t=2}^{T}\mathbb E[J_t^2]$, \\$J_t^2=\sum_{s=2}^{t-1}\sum_{r=2}^{t-1}
\sum_{j\in\mathfrak{J}_{ts}}\sum_{\ell\in\mathfrak{J}_{tr}}
\omega(j)\omega(\ell)
\langle X_t,X_s\rangle \langle X_t,X_r\rangle
\langle Z_{t-j},Z_{s-j}\rangle\langle Z_{t-\ell},Z_{r-\ell}\rangle$.
Notice that, under the null, only the cross-terms such that: $r=s$, will be non-zero. By the LIE and conditional homoskedasticity of $X$, we write:
\begin{align*}
    &\mathbb{E}\left[\langle X_t,X_s\rangle^2
    \langle Z_{t-j},Z_{s-j}\rangle\langle Z_{t-\ell},Z_{s-\ell}\rangle\right]\\
    &\quad =\mathbb{E}[X_s^\prime\mathbb{E}[X_tX_t^\prime|\mathcal{I}(t-1)]X_s\langle Z_{t-j},Z_{s-j}\rangle\langle Z_{t-\ell},Z_{s-\ell}\rangle]\\
    & \quad =\mathbb{E}[||X_s||^2\langle Z_{t-j},Z_{s-j}\rangle\langle Z_{t-\ell},Z_{s-\ell}\rangle]=d_1 \gamma_{t,s}(j,\ell)
\end{align*}
Meaning that: 
\begin{align*}
    \text{Var}[T\cdot\mathcal T_{2\omega}^c]=\frac{4}{T^2}\sum_{t=2}^{T}\sum_{s=2}^{t-1}\sum_{j\in\mathfrak{J}_{ts}}\sum_{\ell\in\mathfrak{J}_{ts}}
    \omega(j)\omega(\ell)d_1 \mathbb{E}\left[\langle Z_{t-j},Z_{s-j}\rangle\langle Z_{t-\ell},Z_{s-\ell}\rangle\right]=D_{\omega,T}^{(Hete)}
\end{align*}
with $\mathfrak J_{ts}:=\{j\in\mathbb Z:\ t-s+1\le j\le s-1\}$, and $\gamma_{t,s}(j,l)=\mathbb E\!\left[\langle Z_{t-j},Z_{s-j}\rangle\ \langle Z_{t-\ell},Z_{s-\ell}\rangle\right]$, when these last moments exist. Notice that when we rewrite the summation in $j-\ell$ terms: $D_{\omega,T}^{(Hete)}=\frac{4d_1}{T^2}
\sum_{j=1}^{T-2}\sum_{\ell=1}^{T-2}\omega(j)\omega(\ell)
\sum_{s=\max\{j,\ell\}+1}^{T-1}
\sum_{t=s+1}^{\min\{T,\ s+\min(j,\ell)-1\}}
\gamma_{t,s}(j,\ell) $.

\subsection{Proof of Theorem \ref{theorem1}}\label{proof_of_theorem1}
The proof of asymptotic normality of the proposed statistic, $\mathcal{T}^{(Hete)}$, translates into proving the asymptotic normality of the dominating term, $T \cdot \mathcal{T}_{2\omega}^c$, as by Proposition \ref{prop_moments} (together with Lemma \ref{lemma_mixing_implications}) we have asymptotically: $(D_{\omega,T})^{-1/2}(T \mathcal{T}_{1\omega}-\mu_{\omega,T})=o_p(1)$, when $M/T\rightarrow 0$ as $M,T\rightarrow \infty$. If we have $M^2/T\rightarrow 0$, the asymptotical negligibility holds as well. For simplicity, define: $W_{ts}:=\sum_{j=t-s+1}^{s-1}\omega(j)\langle Z_{t-j},Z_{s-j}\rangle$, such that: $J_t = \sum_{s=2}^{t-1} \sum_{j=t-s+1}^{s-1} \omega(j)\langle X_t,X_s\rangle \langle Z_{t-j},Z_{s-j}\rangle = \sum_{s=2}^{t-1}X_{t}^\prime X_{s} W_{ts}$. Note that, under Assumption \ref{assumption_kernel1}, $\sum_{j=t-s+1}^{s-1}\omega(j)= O(M)$, as discussed in Proposition \ref{prop_moments}. As well, as previously discussed, the process $\{(J_t,\mathcal{I}(t-1)), t\in \mathbb{Z}^+\}$ constitutes a martingale difference sequence. Finally, to invoke \cite{brown1971martingale}'s theorem, two conditions need to be verified.
\begin{enumerate}[i)]
    \item The Lindeberg condition needs to hold: 
    \begin{align*}
        T^{-2}(D_{\omega,T}^{(Hete)})^{-1} \sum_{t=2}^T \mathbb{E}\Big[(J_t)^2 \cdot \mathbf{1}\big\{|J_t|> \epsilon (D_\omega^{(Hete)})^{1/2} \big\}\Big]\rightarrow 0
    \end{align*}
    It suffices then the following Lyapunov condition to hold: \\ $T^{-4}(D_{\omega,T}^{(Hete)})^{-2}\sum_{t=2}^{T} \mathbb{E}[(J_t)^4]\rightarrow 0$. Let us define: $S_t=\sum_{s=2}^{t-1} X_sW_{ts}=\sum_{s=2}^{t-1}U_{s,t}$. By the conditional homokurtosis of $X$, we can write:
    \begin{align*}
        \mathbb E[J_t^4]=\mathbb E\big[\mathbb E[(X_t'S_t)^4\mid\mathcal I(t-1)]\big]\leq \Delta \mathbb{E}\big[\|S_t\|^4\big].
    \end{align*}
    for some $\Delta>0$, due to the fact that $S_t \in \mathcal{I}(t-1)$. Fix $t$ and define the filtration $\mathcal G_{s}:=\sigma\{(X_u,Z_u):u\le s\}$. Note that, due to the correction term: $W_{ts} \in G_{s-1}$. Thus, under $\mathcal H_0$, $\{X_s,\mathcal G_s\}$ is a martingale difference sequence, so $\mathbb E[X_s\mid \mathcal G_{s-1}]=0$. Since $W_{ts}\in \mathcal G_{s-1}$, we obtain: $\mathbb E[U_{s,t}\mid \mathcal G_{s-1}]=\mathbb E[X_s\mid \mathcal G_{s-1}]\,W_{ts}=0$
    Hence, for each fixed $t$, the sequence $\{(U_{s,t},\mathcal G_s),\ s=2,\dots,t-1\}$ is a martingale difference sequence in $s$, implying $S_t$ is a martingale sum. Next, by Theorem 2.11 of \cite{hall2014martingale}, we apply the Burkholder-type inequality for martingales (for $p=4$ and some finite $\Delta>0$): 
    \begin{align*}
        \mathbb{E}\|S_t\|^4 \leq
        \Delta \left\{
        \mathbb{E}\left(\sum_{s=2}^{t-1}\mathbb E\big[||U_{s,t}||^2| \mathcal{G}_{s-1}\big]\right)^2+
        \sum_{s=2}^{t-1}\mathbb{E}\big[||U_{s,t}||^4\big]
        \right\}
    \end{align*}
    Regarding the second term, by LIE and the conditional homokurtosis, we can write: $ \mathbb{E}\big[||U_{s,t}||^4] = \mathbb{E}\big[||X_{s}||^4 ||W_{ts}||^4] \leq \Delta \mathbb{E}\big[||W_{ts}||^4]$, for some finite $\Delta>0$. By Cauchy-Schwarz, we have: 
    \begin{align*}
        ||W_{ts}||\leq \sum_{j=t-s+1}^{s-1}|\omega(j)|\,\|Z_{t-j}\|\,\|Z_{s-j}\| \leq \Big(\sum_{j=t-s+1}^{s-1}|\omega(j)|\Big)\max_{j}\|Z_{t-j}\|\,\|Z_{s-j}\|
    \end{align*}
    where, by Assumption \ref{assumption_kernel1} and the finiteness of moments, it follows that: \\$\sup_{t>s}\mathbb{E}||W_{ts}||^4 =O(M^4)$. \\ Hence, for some finite $\Delta>0$: $\sum_{s=2}^{t-1}\mathbb{E}||U_{s,t}||^4\leq \Delta \sum_{s=2}^{t-1}\mathbb{E}||W_{ts}||^4
    = O(tM^4)$.\\
    Regarding the first term, under the conditional homoskedasticity of $X$, we have: $ \mathbb{E}[||U_{s,t}||^2 |G_{s-1}]=\mathbb{E}[||X_{s}||^2 |G_{s-1}]W_{ts}^2=d_1W_{ts}^2$.
    Thus:
    \begin{equation}\label{eq:quadratic_variation_term}
    \mathbb{E}\left(\sum_{s=2}^{t-1}\mathbb E\big[||U_{s,t}||^2| \mathcal{G}_{s-1}\big]\right)^2=d_1^2\mathbb{E}\left(\sum_{s=2}^{t-1} W_{ts}^2\right)^2
    \end{equation}
    By Cauchy-Schwarz and the previous bounds: $\mathbb{E}\left(\sum_{s=2}^{t-1} W_{ts}^2\right)^2 =O(t^2M^4)$.
    Combining the bounds, it gives: $ \mathbb{E}[||S_t||^4]= O(t^2M^4)$. Plugging it back, we write: $  \mathbb{E}[J_t^4]=O(t^2M^4), \; \sum_{t=2}^T \mathbb E[J_t^4]=O(T^3M^4)$.
    Since $D_{\omega,T}^{(\mathrm{Hete})}\asymp M$ (Assumption \ref{assumption_kernel1} and Proposition \ref{prop_moments}),
    we obtain: 
    \begin{align*}
       T^{-4}\big(D_{\omega,T}^{(\mathrm{Hete})}\big)^{-2}\sum_{t=2}^T \mathbb{E}[J_t^4]=O\left(T^{-4}M^{-2}\cdot T^3M^4\right)=O\left(\frac{M^2}{T}\right) 
    \end{align*}
    This in turn verifies the Lyapunov (and Lindeberg) condition when $\frac{M^2}{T}\rightarrow 0$.
    To reach a sharper conclusion in terms of the ratio between $M$ and $T$, the alternative line is to appeal to Marcinkiewicz-Zygmund inequalities. Under the assumptions of the time series $\{\Lambda_{i,t}\}$: $C_{r,4}^{\Lambda}=O(r^{-2})$, for $r\rightarrow\infty$, by application of Theorem 4.1 of \cite{dedecker2007weak}, we have that: $ \mathbb{E}[(J_t)^4]=O(t^2M^2)$. Similar result can be obtained via Rosenthal type inequalities. For further details, please refer to Theorem 4.2 and Corollary 5.4 of \cite{dedecker2007weak}. For this latter scenario, we have the sharper conclusion:\\ $ T^{-4}(D_\omega^{(Hete)})^{-2}\left(\sum_{t=2}^{T} \mathbb{E}[(J_t)^4]\right)\leq  \Delta T^{-4} M^{-2} \left(T^{3} M^2\right)$.
    \item The following condition needs to hold:
    \begin{align*}
    T^{-2}\big(D_{\omega,T}^{(\mathrm{Hete})}\big)^{-1}\sum_{t=2}^T \mathbb{E}[J_t^2\mid \mathcal I(t-1)]\xrightarrow{p} 1.
    \end{align*}
    By conditional homoskedasticity of $X$, we write: $ \mathbb{E}[J_t^2|\mathcal I(t-1)]=\|S_t\|^2$. Let us denote: $V_T^2:=\sum_{t=2}^T \mathbb E[J_t^2\mid \mathcal I(t-1)]=\sum_{t=2}^T \|S_t\|^2.$
    Then we wish to prove: $(T^2 D_{\omega,T}^{(Hete)})^{-1}4V_T^2\xrightarrow{p} 1$, since we have: $T\mathcal T_{2\omega}^c=(2/T)\sum_{t=2}^T J_t$. We write:
    \begin{align*}
       ||S_t||^2=\Big\|\sum_{s=2}^{t-1}X_sW_{ts}\Big\|^2=\sum_{s=2}^{t-1}\sum_{r=2}^{t-1}W_{ts}W_{tr}\,\langle X_s,X_r\rangle. 
    \end{align*}
    Under the null, we have: $\mathbb E[W_{ts}W_{tr}\langle X_s,X_r\rangle]=0, \forall r\not= s$.\\ Thus: $  \mathbb{E}||S_t||^2=\sum_{s=2}^{t-1}\mathbb{E}\big[||X_s||^2W_{ts}^2\big]=d_1\sum_{s=2}^{t-1}\mathbb E[W_{ts}^2]$, where the last equality is by conditional homoskedasticity. Hence: $ \mathbb{E}[V_T^2]
    =d_1\sum_{t=2}^{T}\sum_{s=2}^{t-1}\mathbb{E}[W_{ts}^2]
    =\frac{T^2}{4}D_{\omega,T}^{(Hete)}$
    Note that: $\mathbb E[V_T^2]=O(T^2M)$ due to $D_{\omega,T}^{(Hete)}=O(M)$.\\
    Under Assumption \ref{assumptionZmixing}, the products of the form
    $\langle Z_{t-j},Z_{s-j}\rangle\langle Z_{t-\ell},Z_{s-\ell}\rangle$
    have absolutely summable covariances across $t$ and uniformly over $(j,\ell)$ (see Lemma \ref{lemma_mixing_implications}). 
    In particular, we have the following bound for the covariance across $\{\|S_t\|^2\}$:\\ $\sum_{h=1}^{\infty}\sup_{t}\left|\mathrm{Cov}\left(||S_t||^2,||S_{t+h}||^2\right)\right| \leq \Delta M^2$, for some finite $\Delta>0$. Hence:
    \begin{align*}
        \text{Var}(V_T^2)
        &\leq \Delta\left(T\sup_t\mathrm{Var}(\|S_t\|^2)+2T\sum_{h=1}^{\infty}\sup_t\left|\mathrm{Cov}\left(\|S_t\|^2,\|S_{t+h}\|^2\right)\right|\right) \\
        &= O(T^3M^4)
    \end{align*}
    for some finite $\Delta>0$ (see Lemma \ref{lemma_mixing_implications}). Since $\mathbb E[V_T^2]\asymp T^2M$, we conclude:
    \begin{align*}
       \frac{\text{Var}(V_T^2)}{\mathbb{E}[V_T^2]^2}=O(M^2/T) 
    \end{align*}
    which means: $\text{Var}(V_T^2)/\mathbb{E}[V_T^2]^2\rightarrow 0$, as $M,T\rightarrow\infty$, implying the desired convergence in probability. Note that, by a similar argument using Marcinkiewicz-Zygmund inequalities, under the additional conditions on $\{\Lambda_{i,t}\}$, we have: $\text{Var}(V_T^2)= O(T^3M^2)$, which implies: $\text{Var}(V_T^2)/\mathbb{E}[V_T^2]^2=O(1/T)$.
\end{enumerate}

\begin{lemma}\label{lemma_mixing_implications} First, for a finite $\Delta>0$, we have:\\ $\sum_{h=1}^{\infty}\big|\text{Cov}(||Z_0||^2,||Z_h||^2)\big|
\leq \Delta \sum_{h=1}^{\infty}\alpha(h)^{\delta/(8+\delta)}
<\infty$. \\
Second, for a finite $\Delta>0$, we have for all $h\ge1$,
\begin{align*}\label{eq:mixing_cov_A}
&\sup_{t>s}\big|\mathrm{Cov}\big(\langle Z_{t-j},Z_{s-j}\rangle\langle Z_{t-\ell},Z_{s-\ell}\rangle,\langle Z_{t+h-j},Z_{s+h-j}\rangle\langle Z_{t+h-\ell},Z_{s+h-\ell}\rangle\big)\big|\leq \Delta\alpha(h)^{\delta/(8+\delta)}\\
&\sum_{h=1}^{\infty}\sup_{t>s}\big|\text{Cov}\big(\langle Z_{t-j},Z_{s-j}\rangle\langle Z_{t-\ell},Z_{s-\ell}\rangle,\langle Z_{t+h-j},Z_{s+h-j}\rangle\langle Z_{t+h-\ell},Z_{s+h-\ell}\rangle\big)\big|
<\infty
\end{align*}
Third, under the Assumptions of Theorem \ref{theorem1}, for a finite $\Delta>0$, we have: \\$\sum_{h=1}^{\infty}\sup_{t}\big|\mathrm{Cov}(\|S_t\|^2,\|S_{t+h}\|^2)\big|
\le \Delta\,M^2$.
\end{lemma}

\begin{proof}
Let be $p=(8+\delta)/4$. By Assumption \ref{assumptionZmixing}, we have: $\mathbb{E}||Z_0||^{8+\delta}<\infty$, and so the finiteness of the other lower order moments.
By standard inequalities for $\alpha$-mixing sequences \citep{dedecker2007weak}, for some finite $\Delta>0$: $|\mathrm{Cov}(U,V)|\le \Delta\,\|U\|_{p}\|V\|_{p}\,\alpha(h)^{1-2/p}$, where $U$ and $V$ are processes measurable with respect to the sigma of process $Z$.\\
First, by direct application of the inequality with $U=||Z_0||^2$ and $V=||Z_h||^2$, we have: $  |\text{Cov}(||Z_0||^2,||Z_h||^2)|\leq \Delta|| ||Z_0\|^2\|_{p}^2\,\alpha(h)^{\delta/(8+\delta)}$.\\
Second, we have: $\langle Z_{t-j},Z_{s-j}\rangle\langle Z_{t-\ell},Z_{s-\ell}\rangle \leq ||Z_{t-j}||||Z_{s-j}|||| Z_{t-\ell}||||Z_{s-\ell}||$ and so bounded. Fixing $(t,s,j,\ell)$ and looking over lags $h$ along time, for a finite $\Delta>0$:
\begin{align*}
    \sup_{t>s}\big|\text{Cov}\big(\langle Z_{t-j},Z_{s-j}\rangle\langle Z_{t-\ell},Z_{s-\ell}\rangle,\langle Z_{t+h-j},Z_{s+h-j}\rangle\langle Z_{t+h-\ell},Z_{s+h-\ell}\rangle\big)\big|\\
    \leq \Delta \sup_{t>s}||A_{t,s}(j,\ell)||_p^2\alpha(h)^{\delta/(8+\delta)} 
    \leq \Delta\alpha(h)^{\delta/(8+\delta)}
\end{align*}
using a standard argument for overlapping blocks with small $h$.\\
Third, recall: $ ||S_t||^2=\sum_{s=2}^{t-1}\sum_{r=2}^{t-1}W_{ts}W_{tr}\,\langle X_s,X_r\rangle$.
Under the assumptions of Theorem \ref{theorem1}, we have: $\sup_t\mathbb E||X_t||^{4}<\infty$. Let us consider the following: 
\begin{align*}
    W_{ts}^2=\sum_{j\in\mathfrak J_{ts}}\sum_{\ell\in\mathfrak J_{ts}}\omega(j)\omega(\ell)\langle Z_{t-j},Z_{s-j}\rangle\langle Z_{t-\ell},Z_{s-\ell}\rangle
\end{align*}
Note that, due to Assumption \ref{assumption_kernel1}, there are at most $O(M)$ indices and so the total sum has at most $O(M^2)$ terms. Thus, by virtue of Cauchy-Schwarz and the covariance summability, the covariances across $\{\|S_t\|^2\}$ can be bounded as follows: $\sum_{h=1}^{\infty}\sup_{t}\left|\mathrm{Cov}\left(||S_t||^2,||S_{t+h}||^2\right)\right|\leq \Delta\Big(\sum_{j\in \mathfrak{J}}|\omega(j)|\Big)^2\leq \Delta M^2$, for some finite $\Delta>0$. Note that, a similar argument can applied to $\sup_t\text{Var}(||S_t||^2)$, which leads to: $\sup_t\text{Var}(||S_t||^2)=O(T^2M^4)$.
\end{proof}

\subsection{Proof of Theorem \ref{theorem_power}} \label{proof_of_theorem_power}

The proof is parallel to the one of \cite{hong2001test}'s Theorem 2, \cite{bouhaddioui2006generalized}'s Theorem 2. Let us write:
\begin{align*}
   \dfrac{M^{1/2}}{T} \mathcal{T}^{(Hete)} = \dfrac{ M^{1/2}\mathcal{T}_{1\omega}-M^{1/2}T^{-1}\mu_{\omega,T}}{\sqrt{D_{\omega,T}}}+\dfrac{M^{1/2}\mathcal{T}_{2\omega}^c}{\sqrt{D_{\omega,T}^{(Hete)}}}
\end{align*}
First, notice that, under Assumption \ref{assumption_kernel1} and the condition on $Z$ and the asymptotic rates, we have: $\mu_{\omega,T} = O(M)$, $D_{\omega,T}= O(M)$,  $D_{\omega,T}^{(Hete)}= O(M)$ (Proposition \ref{prop_moments}).\\ Define the limit: $M/D_{\omega,T} \rightarrow \Delta^2$, with $\Delta>0$. Since $M/T\rightarrow 0$, we have that:
\begin{align*}
\dfrac{M^{1/2}}{T} \mathcal{T}^{(Hete)} =\Delta\mathcal{T}_{1\omega}+\dfrac{M^{1/2}\mathcal{T}_{2\omega}^c}{\sqrt{D_{\omega,T}^{(Hete)}}}+o(1).
\end{align*}
The proof is concluded once showing: $\mathcal{T}_{1\omega}=\sum_{j=1}^{\infty}   \left\vert\left\vert \text{vec}\left[\Gamma_{XZ}(j)\right] \right\vert\right\vert^2 +o_p(1)$.\\ Notice that:
\begin{align*}
   & \mathcal{T}_{1\omega}-\sum_{j=1}^{\infty} || \Gamma_{XZ}(j)||^2 =\sum_{j=1}^{T-1} \omega(j) \left(\left\vert\left\vert\text{vec}\left[\widehat{\Gamma}_{XZ}(j)\right] \right\vert\right\vert^2-\left\vert\left\vert \text{vec}\left[\Gamma_{XZ}(j)\right] \right\vert\right\vert^2 \right)\\ 
   & \qquad + \sum_{j=1}^{T-1} \left(\omega(j)-1\right)\left\vert\left\vert \text{vec}\left[\Gamma_{XZ}(j)\right] \right\vert\right\vert^2+\sum_{j=T}^{\infty}\left\vert\left\vert \text{vec}\left[\Gamma_{XZ}(j)\right] \right\vert\right\vert^2 
\end{align*}
For the second and last term, by the dominated convergence theorem, the boundedness of the covariances and by the asymptotic rates: \\
$\sum_{j=1}^{T-1} \left(\omega(j)-1\right)\left\vert\left\vert \text{vec}\left[\Gamma_{XZ}(j)\right] \right\vert\right\vert^2 =o_p(1)$, $\sum_{T}^{\infty}\left\vert\left\vert \text{vec}\left[\Gamma_{XZ}(j)\right] \right\vert\right\vert^2=o_p(1)$.\\
For the first term, we have:
\begin{align*}
     &\sum_{j=1}^{T-1} \omega(j) \left(\left\vert\left\vert\text{vec}\left[\widehat{\Gamma}_{XZ}(j)\right] \right\vert\right\vert^2-\left\vert\left\vert \text{vec}\left[\Gamma_{XZ}(j)\right] \right\vert\right\vert^2 \right) \\
     & = \left[\sum_{j=1}^{T-1} \omega(j) \left\vert\left\vert\text{vec}\left[\widehat{\Gamma}_{XZ}(j)\right]-\text{vec}\left[\Gamma_{XZ}(j) \right]\right\vert\right\vert^2\right]+ 2 \sum_{j=1}^{T-1} \omega(j) \lambda(j) 
\end{align*}
with $\lambda(j)=\left\langle \text{vec}\left[\Gamma_{XZ}(j)\right], \text{vec}\left[\widehat{\Gamma}_{XZ}(j)-\Gamma_{XZ}(j)\right] \right\rangle$. Denote $\rho_{kl}(j)$, the $(k,l)-$ entry of the matrix $\Gamma_{XZ}(j)$. 
If $\{X_t,Z_t\}$ is a fourth-order stationary process, by  \cite{hannan1970multiple} pg.209-210 or by \cite{priestley1981spectral} pg. 325–26, for a fixed $h$ lag, up to a negligible component: 
\begin{align*}
&\sum_{k,l} \mathbb{E}\left[\left(\widehat{\Gamma}_{kl,UV}(h)-\Gamma_{kl,UV}(h)\right)^2\right] =\text{Var}[\widehat{\Gamma}_{kl,UV}(h)] \\
&=\frac{1}{T}\left(\sum_{i=-T+1}^{T-1}\left(1-\dfrac{i}{T}\right)\rho_{kl}(i+h)\rho_{kl}(i-h)+\left(1-\dfrac{|i|}{T}\right)\kappa_{klkl,XZ}(h,i,h+i)\right)
\end{align*}
Following a similar argument of \cite{hong2001test}'s Lemma A.6, together with the absolute summability of the fourth-order cumulants, for a finite constant $C$:
\begin{align*}
\sup_{1\le j\le T-1} \mathbb{E}\left[\left\vert\left\vert\text{vec}\left[\widehat{\Gamma}_{XZ}(j)\right]-\text{vec}\left[\Gamma_{XZ}(j) \right]\right\vert\right\vert^2 \right]\leq \frac{C}{T}.
\end{align*}
which means that, by Markov's inequality:
\begin{align*}
   \left[\sum_{j=1}^{T-1} \omega(j) \left\vert\left\vert\text{vec}\left[\widehat{\Gamma}_{XZ}(j)\right]-\text{vec}\left[\Gamma_{XZ}(j) \right]\right\vert\right\vert^2\right]=O_p\left(\frac{M}{T}\right)=o_p(1), 
\end{align*}
where the last equality is because of the assumption on the asymptotic rates. Hence: $\sum_{j=1}^{T-1}\omega(j) \lambda(j) = o_p(1)$, by Cauchy-Schwarz and the boundedness of the covariances. By realizing that $\mathcal{T}_{2\omega}^c$ is a sum of fourth-order cumulants, then we can conclude as well: $\mathcal{T}_{2\omega}^c=O_p(M/T)$. As $D_{\omega,T}^{(Hete)}=O(M)$, this concludes the proof.

\section{Monte Carlo Experiments}\label{appendixC}

\subsection{Simulation Study}\label{sec3_subsec_simulation}

This section summarizes the evidence from the simulation study in the Online Appendix that investigates the finite-sample properties of the testing procedures. The study compares two finite-sample Portmanteau statistics: the finite-sample proposed statistic (\textit{Hete}) and the finite-sample benchmark statistic (\textit{Hong}):
\begin{align*}
       \dfrac{\mathcal{T}_{1\omega}^{(f)}- \mu_{\omega,T}^{(f)}}{\sqrt{ D_{\omega,T}^{(f)}}}+\dfrac{\mathcal{T}_{2\omega}^{(f), c}}{\sqrt{D_{\omega,T}^{(f),(Hete)}}},  \quad  \dfrac{\mathcal{T}_{\omega}^{(f)}- \mu_{\omega,T}^{(f)}}{\sqrt{ D_{\omega,T}^{(f)}}}
\end{align*}
The difference between the previous statistics (together with their center and scale) and these latter finite-sample versions is that are scaled by the effective sample, $T-j$ (or $T-\ell$). Three additional statistics are considered: the finite-sample proposed statistic applied to $Z$ and to residuals from an AR(1) fitted to $X$  (\textit{HeteF}), as well as  Wald statistics from a VAR(1) fitted to the joint process that test either a single zero-restriction (past $Z$ on present $X$) or a double zero-restriction (past $Z$ \textit{and} past $X$ on present $X$). All the statistics are formally defined in Appendix \ref{appendixC_statistics}.

Given the emphasis on the role of the correction term under weak exogeneity, the simulations focus on empirical rejection rates for DGPs where the null holds but inverse causality is present. The main Monte Carlo experiments use bivariate designs where the structural shock $X$ is a standardized strong white noise, while the omitted variable $Z$ depends on past $X$ either through its conditional mean or through its conditional variance. Four DGP families capture different forms of inverse causality: linear and nonlinear dependence in the conditional mean, and ARCH and GARCH-type dependence in the conditional variance. Innovations are drawn from a Student-$t$ distribution (refer to Online Appendix).

Under the null, the proposed statistics display substantially better size properties than the benchmark statistic, and are often comparable to the Wald statistics. In the Linear-in-Mean design, the finite-sample proposed statistic remains close to the nominal level across almost all parameter configurations. By contrast, the benchmark becomes increasingly oversized when the omitted variable is persistent, with rejection rates well above the nominal level. This pattern is consistent with the theoretical mechanism: when inverse causality is combined with persistence in $Z$, the benchmark statistic incorporates dependence from past $X$ to current $Z$ into its variance (Appendix \ref{lemma_inversecausality}). Similar conclusion are more pronounced in nonlinear inverse causality. The benchmark overrejects when persistence is high, whereas the rejection rates of the proposed statistic remain much closer to the nominal level, reducing the size distortion. When inverse causality operates through conditional variances, the evidence is more nuanced. In the ARCH-type DGPs, the benchmark still tends to overreject. By contrast, all procedures are broadly well sized in the GARCH-type DGPs. This suggests that the benchmark distortion depends not only on the presence of inverse causality, but also on how this channel interacts with the serial dependence of the omitted variable (i.e., mean vs. variance effects).

Wald procedures are generally correctly sized under correctly specified (linear) dynamics. However, they become conservative in the nonlinear designs, while the proposed statistics remain better aligned with the nominal size. Across the DGP families, \textit{HeteF} often underrejects, especially when the autoregressive structure of $X$ is irrelevant: the additional AR filtering on the process $X$ is redundant and, in fact, the (unfiltered) proposed statistic already improves over the benchmark.

The second part of the simulation study considers three DGP families in which weak exogeneity fails (i.e., past $Z$ affects present $X$). The power curves show that the proposed statistics retain meaningful power under linear alternatives, especially when the causal coefficient is large and when the bandwidth includes more lags. As expected, power is lower for small violations of weak exogeneity and for nonlinear alternatives. 

Overall, the Monte Carlo evidence supports the asymptotic results. The correction term improves size control when inverse causality is present, while preserving the ability to detect violations of weak exogeneity under fixed alternatives. As previously announced, from the simulations, it emerges that the corrected statistic is preferable when the researcher has limited prior information about how omitted variables respond to past shocks, particularly when omitted variables are persistent and the shock series is non-Gaussian.

Regarding the smoothing parameter $M$, the size correction prioritizes larger lag lengths, suggesting $\ln{T}\ll M<\sqrt{T}$. An informal rule of thumb for the lower bound is: $\underline{M}=0.75T^{1/3}$ \citep{lazarus2018har}, and for the upper bound is: $\overline{M}=\sqrt{T}-1$. More formally, one should appeal to the bandwidth rule in \cite{hong2005generalized} using their plug-in method with a preliminary
bandwidth $M^{pre}=c(10T)^{1/5}$, with $c=2,4,6$ \citep{wang2022testing}. 

\subsection{Finite-sample Statistics}\label{appendixC_statistics}
The Monte Carlo experiments compare five (2+3) test statistics. The empirical application uses the first two test statistics, i.e. the finite-sample Portmanteau asymmetric and benchmark statistics:
\begin{align*}
    \textit{Hong}: \dfrac{\mathcal{T}_{\omega}^{(f)}- \mu_{\omega,T}^{(f)}}{\sqrt{ D_{\omega,T}^{(f)}}}, \quad \textit{Hete}: \dfrac{\mathcal{T}_{1\omega}^{(f)}- \mu_{\omega,T}^{(f)}}{\sqrt{ D_{\omega,T}^{(f)}}}+\dfrac{\mathcal{T}_{2\omega}^{(f), c}}{\sqrt{D_{\omega,T}^{(f),(Hete)}}}
\end{align*}
where:
\begin{align*}
& \mathcal{T}_{\omega}^{(f)}=\mathcal{T}_{1\omega}^{(f)}+\mathcal{T}_{2\omega}^{(f)}, \quad \mathcal{T}_\omega^{(f), c}=\mathcal{T}_{1\omega}^{(f)}+\mathcal{T}_{2\omega}^{(f), c}\\
&\mathcal{T}_{1\omega}^{(f)} = \sum_{j=1}^{T-1} \omega(j)\frac{1}{T-j} \sum_{t=j+1}^T || X_t ||^2|| Z_{t-j} ||^2 \\
&\mathcal{T}_{2\omega}^{(f)}=\sum_{j=1}^{T-2} \omega(j)\frac{1}{T-j} \sum_{s,t=j+1, s\not=t}^T \langle X_t,X_s\rangle\langle Z_{t-j},Z_{s-j}\rangle\\
&\mathcal{T}_{2\omega}^{(f), c}=\sum_{j=1}^{T-2} \omega(j)\frac{1}{T-j} \sum_{s,t=j+1,  s\not=t, j>|t-s|} \langle X_t,X_s\rangle\langle Z_{t-j},Z_{s-j}\rangle\\
\mu_{\omega,T}^{(f)}&=d_1d_2\sum_{j=1}^{T-1}\left(1-\dfrac{j}{T-j}\right) \omega(j) \\
     D_{\omega,T}^{(f)}&=2d_1d_2\sum_{j=1}^{T-2}\left(1-\dfrac{j}{T-j}\right)\left(1-\dfrac{j+1}{T-j}\right) \omega^2(j)\\
     D_{\omega,T}^{(f),(Hete)} & = 4d_1 \sum_{j,\ell=1}^{T-2}\omega(j)\omega(\ell) \widehat{\Xi}_{j\ell}\\
     \widehat{\Xi}_{j\ell} & =\begin{cases}
     \frac{1}{(T-j)^2}\sum_{s=1}^{j-1} \frac{T-s-j}{T-s}\sum_{t=j+s+1}^{T}(\langle Z_{t-j}, Z_{t-j-s}\rangle)^2 \qquad j=\ell\\
     \frac{1}{(T-m)^2}\sum_{s=1}^{\min\{j,\ell\}-1} \frac{T-m-s-\delta}{T-m-s} \\
     \qquad \qquad \sum_{t=m+\delta+1}^{T}\langle Z_{t-j}, Z_{t-j-s}\rangle \langle Z_{t-\ell}, Z_{t-\ell-s}\rangle  \quad j\not=\ell
     \end{cases}
\end{align*}
with $m=\max\{j,\ell\}$, $\delta=|j-\ell|$. The additional three statistics are:
\begin{itemize}
    \item \textit{HeteF}: The proposed statistic (\textit{Hete}) applied to cross-correlations between residuals $\hat{u}_{x,t}$ and $Z_t$, where $\hat{u}_{x,t}=X_t-\hat{a}X_{t-1}$ with $\hat{a}$ the least squares estimator from an AR(1) regression of $X$ on its own lag.
    \item \textit{Wald-single} and \textit{Wald-double}. \\
    Heteroskedastic-consistent Wald statistics from a VAR(1) fitted to the joint process:
\begin{align*}
    \begin{pmatrix}
        X_t \\
        Z_t
    \end{pmatrix}=
    \begin{pmatrix}
        b_{11} & b_{12} \\
        b_{21} & b_{22} 
    \end{pmatrix}
    \begin{pmatrix}
        X_{t-1} \\
        Z_{t-1}
    \end{pmatrix}+\begin{pmatrix}
        u_{x,t} \\
        u_{z,t}
    \end{pmatrix}
\end{align*}
constructed following \cite{lutkepohl2013introduction} (Section 3.6, Eq.~3.6.4).
\textit{Wald-single} tests $\mathcal{H}_0^{single}: b_{12}=0$, \textit{Wald-double} tests $\mathcal{H}_0^{double}: b_{11}=0,b_{12}=0$.
\end{itemize}
For both simulations and empirical applications, the weighting function for the Portmanteau-type statistics is the quadratic spectral kernel, see Eq.(2.7) in \citep{andrews1991heteroskedasticity}: $\omega(j)=k^2(j/M)$, with $k(x)=
\frac{25}{16\pi^2x^2}\left(\frac{\sin(6\pi x/5)}{6\pi x/5} -\cos (6\pi x/5)\right)$.

\clearpage

\section{Online Appendix}\label{Online_Appendix}

\subsection{Asymptotic Theory for Estimated Processes}\label{chp1_sec3_estimated_processes}

Structural shocks and omitted variables are seldom observed directly. Instead, practitioners fit models to observed data and conduct inference using estimated residuals. This practice therefore naturally aligns with the \cite{haugh1976checking}'s two-step approach, where causality is studied after fitting separate models. In this section, we extend the asymptotic theory to this setting. Let $\{W_{1,t},W_{2,t}; t=1,...,T\}$ denote observed processes of dimensions $d_1$ and $d_2$, respectively. We assume both admit causal conditional mean representations: 
\begin{equation}
\begin{aligned}\label{conditional_representation}
	W_{1,t}=\mu_X(\theta^0_{1},\{W_{1,s};s<t\})+X_{t}, \quad 
    W_{2,t}=\mu_Z(\theta^0_{2},\{W_{2,s};s<t\})+Z_{t} 
	\end{aligned} 
\end{equation}
where, $\mu_{X}(\theta^0_{1},\cdot)\in\mathbb{R}^{d_1}$, $\mu_{Z}(\theta^0_{2},\cdot)\in\mathbb{R}^{d_2}$ are known measurable functions parameterized by finite-dimensional time-invariant parameters $\theta_1^0$ and $\theta_2^0$, and such that: $\mathbb{E}[W_{1,t}\vert \{W_{1,s};s<t\}] = \mathbb{E}[W_{1,t}\vert \{X_{s};s<t\}]$, $\mathbb{E}[W_{2,t}\vert \{W_{2,s};s<t\}] = \mathbb{E}[W_{2,t}\vert \{Z_{s};s<t\}]$. The innovation $\{X_t\}$ is defined by the condition: $\mathbb{E}[X_t|\{X_s;s<t\}]=0$, meaning that the shock $X$ forms a martingale difference sequence with respect to its own history. The process $W_2$ is specified as a function of its own past rather than the joint past $\{W_{1,s},W_{2,s}\}_{s<t}$ to avoid implicitly modeling inverse causality. The representation in Eq.(\ref{conditional_representation}) encompasses a general class of multivariate time series models for conditional means (and variance), requiring only that innovations are correctly captured via finite-dimensional time-invariant parameters.

Suppose the practitioner has $\sqrt{T}-$consistent estimators, $\{\widehat{\theta}_i\}_{i=1,2}$, of the true parameters (e.g., least squares or maximum likelihood). Given sample data, estimated structural shocks and omitted variables, $\widehat{X}_{t}(\widehat{\theta}_1)$ and $ \widehat{Z}_{t}(\widehat{\theta}_2)$, depend on the limited observed past, $\widehat{\mathcal{I}}(t-1)$, and on the estimators.\footnote{These are the estimated pseudo-version of the innovations with arbitrary starting values, since we do not observe the infinite past of the time series, i.e., $\mathcal{I}(t-1)$.} Define the standardized zero-mean innovations and their estimated counterparts:
\begin{align}
\label{standardizedprocesses}
&U_t=\left(\Gamma_X\right)^{-1/2}X_t, \quad
V_t=\left(\Gamma_Z\right)^{-1/2}Z_t, \quad \text{s.t.} \quad \mathbb{E}[||U_t||^2]=d_{1}, \;  \mathbb{E}[||V_t||^2]=d_{2} \\
\label{standardizedresiduals}
&\widehat{U}_t=\left(\widehat{\Gamma}_X\right)^{-1/2}\widehat{X}_t, \quad
\widehat{V}_t=\left(\widehat{\Gamma}_Z\right)^{-1/2}\widehat{Z}_t, \quad \qquad t=1,...,T
\end{align}
where $\Gamma_X=\mathbb{E}[X_t X_t^\prime]$ and $\Gamma_Z=\mathbb{E}[Z_t Z_t^\prime]$ have empirical counterparts $\widehat{\Gamma}_X$ and $\widehat{\Gamma}_Z$. 

Parallel to Eq.(\ref{corrected_statistic}), the proposed statistic based on estimated standardized processes is:
\begin{align}
    \widehat{\mathcal{T}}_\omega^c & = \widehat{\mathcal{T}}_\omega - \frac{1}{T^2} \sum_{j=1}^{T-2} \omega(j) \sum_{s,t=j+1,j\leq |t-s|}^T \langle\widehat{U}_t,\widehat{U}_s\rangle\langle\widehat{V}_{t-j},\widehat{V}_{s-j}\rangle= \widehat{\mathcal{T}}_{1\omega}+\widehat{\mathcal{T}}_{2\omega}^c \label{statistic12}
\end{align}
where $\widehat{\mathcal{T}}_\omega=\sum_{j=1}^{T-1}\omega(j)|| \widehat{\Gamma}_{UV}(j)||_F^2$ is the empirical counterpart of the statistic in Eq.(\ref{decompositionsintosums}). Similar to Eq.(\ref{kernel1}), define the variance estimator:
\begin{align*}
    \widehat{D}_{\omega,T}^{(Hete)} &= \frac{2d_1}{T^2}
        \sum_{j=1}^{T-2}\sum_{\ell=1}^{T-2}\omega(j)\omega(\ell)
        \sum_{s=\max\{j,\ell\}+1}^{T-1}
        \sum_{t=s+1}^{\min\{T,\ s+\min(j,\ell)-1\}}
         \widehat{\gamma}_{t,s}(j,\ell)
\end{align*}
where: $\widehat{\gamma}_{t,s}(j,\ell)=
\langle \widehat{V}_{t-j},\widehat{V}_{s-j}\rangle \langle \widehat{V}_{t-\ell},\widehat{V}_{s-\ell}\rangle$. Finally, define:
\begin{align*}
    \widehat{\mathcal{T}}^{(Hete)}:=T\left(\dfrac{ \widehat{\mathcal{T}}_{1\omega}-\mu_{\omega,T}}{\sqrt{D_{\omega,T}}}+\dfrac{\widehat{\mathcal{T}}_{2\omega}^c}{\sqrt{\widehat{D}_{\omega,T}^{(Hete)}}}\right)
\end{align*}
\begin{theo}\label{theorem_conditional}
Suppose the processes $\{W_{i,t}\}^{i=1,2}_{t=1,...,T}$ admit the causal representation of Eq.(\ref{conditional_representation}). 
Suppose the assumptions of Theorem \ref{theorem1}, with $M/T\rightarrow 0$, as $T,M\rightarrow\infty$, and additionally Assumptions \ref{martingale_estimated}-\ref{regularityassumption}-\ref{score_orthogonality}.
Let $\{\widehat{\theta}_i\}_{i=1,2}$ be $\sqrt{T}-$consistent estimators of the true parameters $\{\theta_i^0\}_{i=1,2}$. Under the null $\mathcal{H}_0$ in Eq.(\ref{h0grangernoncausality}), we have: $\widehat{\mathcal{T}}^{(Hete)}\xrightarrow{d} \mathcal{N}(0,1)$.
\end{theo}

Theorem \ref{theorem_conditional} shows that estimation error does not affect the limiting distribution under three sets of additional conditions. First, correct model specification delivers parametric-rate estimation of the innovations, so the approximation error is asymptotically negligible. Second and third, appropriate smoothness ensures uniform $\ell_2$-convergence and bounded derivatives (i.e., Eq.(\ref{regularitycondition1})–(\ref{regularitycondition2})). Consequently, the plug-in statistic constructed from standardized estimated residuals retains the same asymptotic distribution as in the infeasible case. \\

To prove Theorem \ref{theorem_conditional}, some additional notations and objects need to be introduced.
Given a parameter belonging to a parameter space, $\theta \in \Theta$, I define the gradient and Hessian operators with respect to $\theta$ to be $\nabla_{\theta}$ and $\nabla_{\theta}^2$, whenever have proper meaning. Define $\{\Theta_i\}_{i=1,2}$ the parameter spaces of the real parameters $\{\theta_i^0\}_{i=1,2}$, respectively. Following Eq.(\ref{conditional_representation}), define the processes: $\widetilde{X}_{t}(\theta_1)=W_{1,t}-\mu_X(\theta_{1},\mathcal{I}_X(t-1)), \widetilde{Z}_{t}(\theta_2) =W_{2,t}-\mu_Z(\theta_{2},\mathcal{I}_Z(t-1))$, 
with corresponding standardized innovations of the unobservable infinite past: $\widetilde{U}_t(\theta_1)=\left(\Gamma_X\right)^{-1/2}\widetilde{X}_{t}(\theta_1), \;
\widetilde{V}_t(\theta_2)=\left(\Gamma_Z\right)^{-1/2}\widetilde{Z}_{t}(\theta_2)$. Under the parametrization of eq.(\ref{conditional_representation}), we have: $ U_t=\widetilde{U}_t(\theta_1^0), \; V_t=\widetilde{V}_t(\theta_2^0)$. Define further the population-standardized estimated innovations as follow: $\breve{U}_t=\left(\Gamma_X\right)^{-1/2}\widehat{X}_{t}, \;
\breve{V}_t=\left(\Gamma_Z\right)^{-1/2}\widehat{Z}_{t}$, where, using the previous notation, we have: $	\widehat{X}_{t}=W_{1,t}-\mu_X(\widehat{\theta}_{1},\widehat{\mathcal{I}}_X(t-1))$, $\widehat{Z}_{t}=W_{2,t}-\mu_Z(\widehat{\theta}_{2},\widehat{\mathcal{I}}_Z(t-1))$
where $\widehat{\mathcal{I}}_X(t-1)$ and $\widehat{\mathcal{I}}_Z(t-1)$ are the feasible information sets, i.e., the information sets constrained to the observable finite past of the time series, $\{W_{i,t}\}_{t=1,...,T}^{i=1,2}$. Note that, generally:
$\widehat{X}_{t}\not=\widetilde{X}_t(\widehat{\theta}_1), \;    \widehat{Z}_{t}\not=\widetilde{Z}_t(\widehat{\theta}_2)$. Recall that: 
$\widehat{U}_t=\left(\widehat{\Gamma}_X\right)^{-1/2}\widehat{X}_{t}, \; \widehat{V}_t=\left(\widehat{\Gamma}_Z\right)^{-1/2}\widehat{Z}_{t}$.
Denote the $k^{th}$ entry-wise element of $\widehat{U}_{t},\breve{U}_{t}, \widetilde{U}_{t}, U_{t}$ with $\widehat{U}_{k,t}, \breve{U}_{k,t}, \widetilde{U}_{k,t}, U_{k,t}$, respectively. In a similar fashion, denote $\widehat{V}_{l,t-j},\breve{V}_{l,t-j}, \widetilde{V}_{l,t-j}, V_{l,t-j}$ the $l^{th}$ element of, respectively, $\widehat{V}_{t-j},\breve{V}_{t-j}, \widetilde{V}_{t-j}, V_{t-j}$. Denote: 
\begin{align*}
& C_{UV}(j)=\frac{1}{T} \sum_{t=j+1}^T U_t (V_{t-j})^\prime, \; \; C_{X}=\frac{1}{T} \sum_{t=1}^T X_t (X_{t})^\prime, \; \; C_{Z}=\frac{1}{T} \sum_{t=1}^T Z_t (Z_{t})^\prime\\
& \widehat{\Gamma}_{UV}(j)=\frac{1}{T} \sum_{t=j+1}^T \widehat{U}_t (\widehat{V}_{t-j})^\prime, \; \; \widehat{C}_{UV}(j)=\frac{1}{T}\sum_{t=j+1}^T \breve{U}_t (\breve{V}_{t-j})^\prime, \\
& \widehat{\Gamma}_{\hat{X}\hat{Z}}(j)=\frac{1}{T}\sum_{t=j+1}^T \widehat{X}_t \widehat{Z}_{t-j}^\prime
\end{align*}
which are respectively: i) the sample covariance between standardized innovations and the sample variances of the innovations; ii) the sample covariance between feasible standardized residuals, iii) and the population-standardized sample covariance between estimated residuals, iv) the sample covariance between estimated residuals. Additional to the assumptions listed in Theorem \ref{theorem_conditional}, this paper presume the following assumptions:
\begin{assum}\label{martingale_estimated}
For each $j\ge 1$ and each $(k,l)$, the processes
$\{(\breve{U}_{k,t}-\widetilde{U}_{k,t})V_{l,t-j}\}_{t\ge j+1}$ and $\{U_{k,t}(\breve{V}_{l,t-j}-\widetilde{V}_{l,t-j})\}_{t\ge j+1}$ are martingale difference sequences with respect to the feasible filtration, $\widehat{\mathcal{I}}(t-1)$, with: $\sup_{t\geq j+1}\mathbb{E}\left[(\breve{U}_{k,t}-\widetilde{U}_{k,t})^2V_{l,t-j}^2 \vert \widehat{\mathcal{I}}(t-1)\right]=O(T^{-1})$, \\$\sup_{t\geq j+1}\mathbb{E}\left[U_{k,t}^2(\breve{V}_{l,t-j}-\widetilde{V}_{l,t-j})^2 \vert \widehat{\mathcal{I}}(t-1)\right]=O(T^{-1})$.
\end{assum}
\begin{assum}\label{regularityassumption}
\begin{equation}\label{regularitycondition1}
\begin{aligned}
\sup_{\theta_1\in\Theta_1}\frac{1}{T} \sum_{t=1}^T \mathbb{E} ||U_t-\widetilde{U}_t||^2=O(T^{-1}), \quad  \sup_{\theta_1\in\Theta_1}\frac{1}{T} \sum_{t=1}^T  \mathbb{E} ||\breve{U}_t-\widetilde{U}_t||^2=O(T^{-1})\\
\sup_{\theta_2\in\Theta_2}\frac{1}{T} \sum_{t=1}^T   \mathbb{E} ||V_t-\widetilde{V}_t||^2=O(T^{-1}), \quad \sup_{\theta_2\in\Theta_2}\frac{1}{T} \sum_{t=1}^T  \mathbb{E} ||\breve{V}_t-\widetilde{V}_t||^2=O(T^{-1}) 
\end{aligned}
\end{equation}
\begin{equation}\label{regularitycondition2}
\begin{aligned}
\sup_{\theta_1\in\Theta_1}\frac{1}{T}\sum_{t=1}^T \mathbb{E} ||\nabla_{\theta_1}\widetilde{U}_t(\theta_1)||^4=O(1), & &	\sup_{\theta_2\in\Theta_2}\frac{1}{T}\sum_{t=1}^T \mathbb{E} ||\nabla_{\theta_2}\widetilde{V}_t(\theta_2)||^4=O(1) \\
\sup_{\theta_1\in\Theta_1}\frac{1}{T}\sum_{t=1}^T \mathbb{E} ||\nabla_{\theta_1}^2\widetilde{U}_t(\theta_1)||^4=O(1), & &	\sup_{\theta_2\in\Theta_2}\frac{1}{T}\sum_{t=1}^T \mathbb{E} ||\nabla_{\theta_2}^2\widetilde{V}_t(\theta_2)||^4=O(1)
\end{aligned}
\end{equation}    
\end{assum}
\begin{assum}\label{score_orthogonality}
For each $k=1,...,d_1$, $l=1,...,d_2$
\begin{align*}
    \mathbb{E}\left[\left\|\sum_{t=j+1}^T \nabla_{\theta_1}\widetilde U_{k,t}(\theta_1^0) V_{l,t-j} \right\|^2 \right]=O(T), \quad  \mathbb{E}\left[\left\|\sum_{t=j+1}^TU_{k,t}\nabla_{\theta_2}\widetilde V_{l,t-j}(\theta_2^0)\right\|^2\right]=O(T)
\end{align*}
and for any value $\bar\theta_i \in [\widehat\theta_i,\theta_i^0]$:
\begin{align*}
    \mathbb{E}\left[\left\|\sum_{t=j+1}^T\nabla_{\theta_1}^2\widetilde U_{k,t}(\bar\theta_1) V_{l,t-j}\right\|^2\right]=O(T^2), \quad 
    \mathbb{E}\left[\left\|\sum_{t=j+1}^T U_{k,t} \nabla_{\theta_2}^2\widetilde V_{l,t-j}(\bar\theta_2)\right\|^2\right]
=O(T^2).
\end{align*}
\end{assum}
which are regularity conditions on the uniform $\ell_2-$convergence, the second-order differentiability and boundedness of the derivatives. These last conditions are comparably standard in the literature. For instance, see Assumption A3 in \cite{hong2005generalized}, Assumptions 3.2-3.4 in \cite{wang2022testing}, and Assumptions 2.3-2.4 in \cite{leong2023practical}. A remark is needed. Assumption \ref{martingale_estimated} reads as a high-level orthogonality condition on the estimation error of the structural shock $X$, that ensures that the estimation error is (sufficiently) asymptotically uncorrelated, given the consistent $\sqrt{T}-$estimator. When admitting the following expansion: $\breve U_t-\widetilde U_t = \nabla_{\theta_1}\widetilde U_t(\theta_1^0)^{\prime}(\widehat\theta_1-\theta_1^0)+r_t$, with: $\sup_t \mathbb{E} r_t^2=o(T^{-1})$, Assumption \ref{martingale_estimated} essentially requires the estimation error term to be orthogonal to $V_{t-j}$, conditionally on the feasible past. Other conditions may be imposed to reach similar conclusions \citep{dominguez2020specification}.\\

We write the following: 
\begin{align*}
  \widehat{\mathcal{T}}^{(Hete)} &=  \dfrac{T \left(\widehat{\mathcal{T}}_{\omega}-\widehat{\mathcal{T}}_{\omega}^\star \right)}{\sqrt{D_{\omega,T}}}+
    \dfrac{T \left(\widehat{\mathcal{T}}_{\omega}^\star -\mathcal{T}_{\omega}^\star\right)}{\sqrt{D_{\omega,T}}}+
    \dfrac{T \left(\mathcal{T}_{\omega}^\star\right)}{\sqrt{D_{\omega,T}}} -\dfrac{T\left(\mathcal{T}_{2\omega}^{\star}\right)}{\sqrt{D_{\omega,T}}}\\
    &+ 
    \dfrac{T\left(\widehat{\mathcal{T}}_{2\omega}^c-\mathcal{T}_{2\omega}^{c\star}\right)}{\sqrt{\widehat{D}_{\omega,T}^{(Hete)}}} +\dfrac{\sqrt{D_{\omega,T}^{(Hete)}}}{\sqrt{\widehat{D}_{\omega,T}^{(Hete)}}} \left(\dfrac{T (\mathcal{T}_{2\omega}^{c\star})}{\sqrt{D_{\omega,T}^{(Hete)}}}\right)
\end{align*}
\begin{align*}
&\widehat{\mathcal{T}}_{\omega}=\sum_{j=1}^{T-1} \omega(j) ||\widehat{\Gamma}_{UV}(j)||_F^2, \; \widehat{\mathcal{T}}_{\omega}^\star =\sum_{j=1}^{T-1} \omega(j)||\widehat{C}_{UV}(j)||_F^2, \;  \mathcal{T}_{\omega}^\star =\sum_{j=1}^{T-1} \omega(j) ||C_{UV}(j)||_F^2 \\
&\widehat{\mathcal{T}}_{2\omega}^c= \frac{1}{T^2} \sum_{j=1}^{T-2} \omega(j) \sum_{s,t=j+1,s\not=t,j>|t-s|}^{T}  \langle\widehat{U}_t,\widehat{U}_s\rangle\langle\widehat{V}_{t-j},\widehat{V}_{s-j}\rangle \\
&\widehat{\mathcal{T}}_{2\omega}^{c\star}= \frac{1}{T^2} \sum_{j=1}^{T-2} \omega(j) \sum_{s,t=j+1,s\not=t,j>|t-s|}^{T}\langle U_t,U_s\rangle\langle V_{t-j},V_{s-j}\rangle 
\end{align*}
where, $\mathcal{T}_{2\omega}^{\star}$ and $\mathcal{T}_{2\omega}^{c\star}$, is defined accordingly to Eq.(\ref{corrected_statistic}), and the scale, $\widehat{D}_{\omega,T}^{(Hete)}$, is defined accordingly to Eq.(\ref{kernel1}). The proof of Theorem \ref{theorem_conditional} follows from Propositions (\ref{prop_conditional1})-(\ref{prop_conditional4}), and by Slutsky's together with Theorem \ref{theorem1} which shows that:\\ $(D_{\omega,T})^{-1/2}(T \cdot \mathcal{T}_{1\omega}^{c\star}-\mu_{\omega,T})+(D_{\omega,T}^{(Hete)})^{-1/2}(T \cdot \mathcal{T}_{2\omega}^{c\star})\xrightarrow{d}\mathcal{N}(0,1)$.
\begin{prop}\label{prop_conditional1}
	Under the assumptions of Th.\ref{theorem_conditional}: $(\widehat{D}_{\omega,T}^{(Hete)})^{-1/2}T \left(\widehat{\mathcal{T}}_{\omega}^\star -\mathcal{T}_{\omega}^\star\right)\xrightarrow{p} 0$.
\end{prop}
\begin{prop}\label{prop_conditional2}
	Under the assumptions of Th.\ref{theorem_conditional}: 
    $(\widehat{D}_{\omega,T}^{(Hete)})^{-1/2}T \left(\widehat{\mathcal{T}}_{\omega}-\widehat{\mathcal{T}}_{\omega}^\star \right)\xrightarrow{p} 0$.
\end{prop}
\begin{prop}\label{prop_conditional3}
	Under the assumptions of Th.\ref{theorem_conditional}: $(\widehat{D}_{\omega,T}^{(Hete)})^{-1/2}T\left(\widehat{\mathcal{T}}_{2\omega}^c-\mathcal{T}_{2\omega}^{c\star}\right)\xrightarrow{p} 0$.
\end{prop}
\begin{prop}\label{prop_conditional4}
	Under the assumptions of Th.\ref{theorem_conditional}: $\sqrt{\widehat{D}_{\omega,T}^{(Hete)}}/\sqrt{D_{\omega,T}^{(Hete)}}\xrightarrow{p}1$
	\begin{proof}
		Given the consistency of the estimators, the boundedness of the moments and eq.(\ref{regularitycondition1})-(\ref{regularitycondition2}), showing that $\left(\widehat{D}_{\omega,T}^{(Hete)}-D_{\omega,T}^{(Hete)}\right)=o_p(1)$ is parallel to the proof of Proposition \ref{prop_conditional2}. The proof concludes by application of the continuous mapping theorem.
	\end{proof}
\end{prop}

\subsubsection{Proof of Proposition \ref{prop_conditional1}}\label{proof_prop_conditional1}

The proof is parallel to the one of \cite{hong2001test}'s Lemma A.1-2, \cite{bouhaddioui2006generalized}'s Lemma 2, and \cite{leong2023practical}'s Appendix B. The aim is to show that: $ T \left(\widehat{\mathcal{T}}_{\omega}^\star -\mathcal{T}_{\omega}^\star\right)=o_p(M^{1/2}) $, since $\widehat{D}_{\omega,T}^{(Hete)}=O(M)$ by Proposition \ref{prop_moments} and Proposition \ref{prop_conditional4}. Consider the difference:
\begin{align*}
   T\left(\widehat{\mathcal{T}}_{\omega}^\star -\mathcal{T}_{\omega}^\star\right)&
   = T \sum_{j=1}^{T-1} \omega(j) \left(\left\vert\left\vert\text{vec}\left[\widehat{C}_{UV}(j)\right] \right\vert\right\vert^2-\left\vert\left\vert\text{vec}\left[C_{UV}(j)\right] \right\vert\right\vert^2\right)\\
   &=T\sum_{j=1}^{T-1}  \omega(j) \Big(
||\text{vec}(\widehat{C}_{UV}(j))-\text{vec}(C_{UV}(j))||^2  \\
& \; \; \; \; \; \; \; \; \; \; \; \; + 2 \langle \text{vec}(C_{UV}(j)), \text{vec}(\widehat{C}_{UV}(j))-\text{vec}(C_{UV}(j)) \rangle \Big) 
\end{align*}
Let us consider the first term. We have: 
\begin{align*}
    & ||\text{vec}(\widehat{C}_{UV}(j))-\text{vec}(C_{UV}(j))||^2 = 
    \sum_{k=1}^{d_1}\sum_{l=1}^{d_2} \left(\sum_{t=j+1}^{T} \frac{\breve{U}_{k,t} \breve{V}_{l,t-j}}{T}-\frac{U_{k,t} V_{l,t-j}}{T}\right)^2 \\
    &=\sum_{k=1}^{d_1}\sum_{l=1}^{d_2} \Bigg(\sum_{t=j+1}^{T} \frac{(\breve{U}_{k,t}-U_{k,t})V_{l,t-j}}{T}+\frac{U_{k,t}(\breve{V}_{l,t-j}- V_{l,t-j}) }{T}+\frac{(\breve{U}_{k,t}-U_{k,t})(\breve{V}_{l,t-j}- V_{l,t-j})}{T}\Bigg)^2\\
    &= \sum_{k=1}^{d_1}\sum_{l=1}^{d_2} \left(F_{Tj}^{(1)}+F_{Tj}^{(2)}+F_{Tj}^{(3)}\right)^2
\end{align*} 
By Cauchy-Schwarz inequality, for some finite $\Delta>0$: \\
$||\text{vec}(\widehat{C}_{UV}(j))-\text{vec}(C_{UV}(j))||^2\leq \Delta \sum_{k=1}^{d_1}\sum_{l=1}^{d_2} \left((F_{Tj}^{(1)})^2+(F_{Tj}^{(2)})^2+(F_{Tj}^{(3)})^2\right)$. Let us study the three terms:
    \begin{enumerate}
    \item By Cauchy--Schwarz,
    $\sup_j (F_{Tj}^{(3)})^2 \leq \left(\frac{1}{T}\sum_{t=1}^T(\breve U_{k,t}-U_{k,t})^2\right)\left(\frac{1}{T}\sum_{t=1}^T(\breve V_{l,t}-V_{l,t})^2\right)$.
    Next, by decomposing the difference as:
    $\breve U-U=(\breve U-\widetilde U)+(\widetilde U-U)$, we write:
    \begin{align*}
    \frac{1}{T}\sum_{t=1}^T(\breve U_{k,t}-U_{k,t})^2
    &\le 2\left(\frac{1}{T}\sum_{t=1}^T(\breve U_{k,t}-\widetilde U_{k,t})^2\right)
       +2\left(\frac{1}{T}\sum_{t=1}^T(\widetilde U_{k,t}-U_{k,t})^2\right)
    =O_p(T^{-1}),
    \end{align*}
    where the last equality follows from \eqref{regularitycondition1} and by Markov inequality.
    The same holds for: $\frac{1}{T}\sum_{t=1}^T(\breve V_{l,t}-V_{l,t})^2=O_p(T^{-1})$. Hence: $\sup_j (F_{Tj}^{(3)})^2 = O_p(T^{-2})$.
    \item By similar decomposition:\\ $F_{Tj}^{(1)}=\frac{1}{T}\sum_{t=j+1}^T(\breve U_{k,t}-\widetilde U_{k,t})V_{l,t-j}
    +\frac{1}{T}\sum_{t=j+1}^T(\widetilde U_{k,t}-U_{k,t})V_{l,t-j}=F_{Tj}^{(11)}+F_{Tj}^{(12)}$.
    \begin{itemize}
        \item Under Assumption \ref{martingale_estimated}, we can write:\\
        $\mathbb{E}\big[(F_{Tj}^{(11)})^2\big]
        =\frac{1}{T^2}\sum_{t=j+1}^T \mathbb{E}\big[(\breve U_{k,t}-\widetilde U_{k,t})^2V_{l,t-j}^2\big] = O(T^{-2})$.
        By Markov inequality, we have: $F_{Tj}^{(11)}=O_p(T^{-1})$ and $ (F_{Tj}^{(11)})^2=O_p(T^{-2})$.
        \item By the conditions of eq.(\ref{regularitycondition2}), the term $F^{(12)}_{Tj}$ can be expressed into two terms using its Taylor expansion (up to the second order):
		\begin{align*}
		F^{(12)}_{Tj} 
		& = \frac{1}{T}(\hat{\theta}_1-\theta_1^0)^\prime \sum_{t=j+1}^T (\nabla_{\theta_1}\widetilde{U}_{k,t}(\theta_1^0)V_{l,t-j} ) \\
		& \quad \quad+ \frac{1}{2T} (\hat{\theta}_1-\theta_1^0)^\prime \sum_{t=j+1}^T (\nabla_{\theta_1}^2\widetilde{U}_{k,t}(\breve{\theta}_1)V_{l,t-j}) (\hat{\theta}_1-\theta_1^0)
		\end{align*}
		where $\breve{\theta}_1\in[\hat{\theta}_1,\theta_1^0]$. By virtue of Cauchy-Schwarz inequality: 
		\begin{align*}
		& \frac{1}{T}(\hat{\theta}_1-\theta_1^0)^\prime \sum_{t=j+1}^T (\nabla_{\theta_1}\widetilde{U}_{k,t}(\theta_1^0)V_{l,t-j}) = \frac{1}{T}(\hat{\theta}_1-\theta_1^0)^\prime \sum_{t=j+1}^T ((\nabla_{\theta_1}U_{k,t})(V_{l,t-j}))\\
				& \; \; \; \; \; \; \; \; \; \; \leq \frac{1}{T^2} \mathbb{E}\left[ ||\hat{\theta}_1-\theta_1^0||^2\right]\mathbb{E} \left[\left\{\sum_{t=j+1}^T || (\nabla_{\theta_1}U_{k,t})(V_{l,t-j}) || \right\}^2\right] \\
				& \; \; \; \; \; \; \; \; \; \; \leq  \frac{1}{T^2} \mathbb{E}\left[ ||\hat{\theta}_1-\theta_1^0||^2\right] \sum_{t=j+1}^T \mathbb{E} \left[|| \nabla_{\theta_1}U_{k,t}||^4\right]\mathbb{E}\left[||V_{l,t-j}||^2\right] = O_p(T^{-2})
		\end{align*}
		where the last equality follows from the boundedness of the moments, consistency of the estimator, and Assumptions \ref{regularityassumption}-\ref{score_orthogonality}.
		By the same logic:
		\begin{align*} &\frac{1}{2T} (\hat{\theta}_1-\theta_1^0)^\prime \sum_{t=j+1}^T (\nabla_{\theta_1}^2\widetilde{U}_{k,t}(\breve{\theta}_1)V_{l,t-j} ) (\hat{\theta}_1-\theta_1^0)\\
		& \qquad \qquad \leq \frac{1}{4T^2} \mathbb{E}\left[ ||\hat{\theta}_1-\theta_1^0||^4\right]\mathbb{E} \left[\left\vert \left\vert \sum_{t=j+1}^T  \nabla_{\theta_1}^2\widetilde{U}_{k,t}(\breve{\theta}_1)V_{l,t-j} \right\vert \right\vert^2\right]= O_p(T^{-2})
		\end{align*}
        which means: $F^{(12)}_{Tj}=O_p(T^{-2})$. In conclusions: $(F_{Tj}^{(1)})^2=O_p(T^{-2})$.
    \end{itemize}
    \item By reasoning analogue to the one above, we have: $(F_{Tj}^{(2)})^2=O_p(T^{-2})$.
    \end{enumerate}
    Since: $\sum_{j=1}^{T-1}\omega(j)||\text{vec}(\widehat C_{UV}(j))-\text{vec}(C_{UV}(j))||^2=O_p(MT^{-2})$, we have:\\ $T\sum_{j=1}^{T-1}\omega(j)||\text{vec}(\widehat C_{UV}(j))-\text{vec}(C_{UV}(j))||^2=O_p(M/T)$.
    Regarding the second term, by Cauchy-Schwarz inequality, we have:
    \begin{align*}
     &\sum_{j=1}^{T-1}\omega(j) \langle \text{vec}(C_{UV}(j)), \text{vec}(\widehat{C}_{UV}(j))-\text{vec}(C_{UV}(j)) \rangle \\
    & \; \; \; \; \; \; \; \; \; \; \; \;  \; \; \;=  \sum_{k=1}^{d_1d_2}\sum_{l=1}^{d_1d_2}  \sum_{j=1}^{T-1}\omega(j)  C_{k,UV}(j)\left( \widehat{C}_{l,UV}(j)-C_{l,UV}(j)\right) \\
    & \; \; \; \; \; \; \; \; \; \; \; \;  \; \; \;\leq \sum_{k=1}^{d_1d_2}\sum_{l=1}^{d_1d_2} \left(\sum_{j=1}^{T-1}\omega(j)  (C_{k,UV}(j))^2\right)^{1/2}\left(\sum_{j=1}^{T-1}\omega(j) \left(\widehat{C}_{l,UV}(j)-C_{l,UV}(j)\right)^2\right)^{1/2}
\end{align*}
By previous results: $\left(\sum_{j=1}^{T -1}\omega(j)\left(\widehat{C}_{l,UV}(j)-C_{l,UV}(j)\right)^2\right)^{1/2}=O_p(M^{1/2}T^{-1})$ and \\
$\left(\sum_{j=1}^{T-1}\omega(j) (C_{k,UV}(j))^2\right)^{1/2}=O_p(M^{1/2}T^{-1/2})$, \\
where the last equality is by: \\
$\sum_{j=1}^{T-1}\omega(j) (C_{k,UV}(j))^2 \leq \sum_{j=1}^{T-1}\omega(j)||C_{k,UV}(j)||^2 =O(MT^{-3/2})$.\\
In conclusions: $ T\sum_{j=1}^{T-1}\omega(j) \langle \text{vec}(C_{UV}(j)), \text{vec}(\widehat{C}_{UV}(j))-\text{vec}(C_{UV}(j)) \rangle= O_p(MT^{-1/2})$.
As announced, since $M/T\rightarrow 0$, as $T,M\rightarrow \infty$:$ \widehat{\mathcal{T}}_{\omega}^\star -\mathcal{T}_{\omega}^\star=o_p(M^{1/2})$, which means that the proof is concluded by Proposition \ref{prop_moments} and Proposition \ref{prop_conditional4}.

\subsubsection{Proof of Proposition \ref{prop_conditional2}}\label{proof_prop_conditional2}
The proof is parallel to the one of \cite{hong2001test}'s Lemma A.1, \cite{bouhaddioui2006generalized}'s Proposition 3, and \cite{leong2023practical}'s Appendix B. Similar to the proof of Proposition \ref{prop_conditional1}, the aim is to show that: $ T \left(\widehat{\mathcal{T}}_{\omega}-\widehat{\mathcal{T}}_{\omega}^\star \right)=O_p(MT^{-1/2})$, since $\widehat{D}_{\omega,T}^{(Hete)}=O(M)$ by Proposition \ref{prop_conditional4} and Proposition \ref{prop_moments}. We write:
\begin{align*}
   T \left(\widehat{\mathcal{T}}_{\omega}-\widehat{\mathcal{T}}_{\omega}^\star \right)
   &= T\sum_{j=1}^{T-1} k^2\left(\frac{j}{M}\right) \Big(\text{vec}(\widehat{\Gamma}_{\hat{X}\hat{Z}}(j))^\prime \left( \widehat{\Gamma}_{Z}^{-1} \otimes \widehat{\Gamma}_{X}^{-1}\right)\text{vec}(\widehat{\Gamma}_{\hat{X}\hat{Z}}(j))\\
   & \qquad \qquad \qquad \qquad \qquad -\text{vec}(\widehat{\Gamma}_{\hat{X}\hat{Z}}(j))^\prime \left( \Gamma_{Z}^{-1} \otimes \Gamma_{X}^{-1}\right)\text{vec}(\widehat{\Gamma}_{\hat{X}\hat{Z}}(j))\Big)\\
&= T\sum_{j=1}^{T-1} k^2\left(\frac{j}{M}\right) \text{vec}(\widehat{\Gamma}_{\hat{X}\hat{Z}}(j))^\prime \left( \widehat{\Gamma}_{Z}^{-1} \otimes \widehat{\Gamma}_{X}^{-1}- \Gamma_{Z}^{-1} \otimes \Gamma_{X}^{-1}\right)\text{vec}(\widehat{\Gamma}_{\hat{X}\hat{Z}}(j)) 
\end{align*}
Recall that $\widehat{X}_{k,t}, X_{k,t}$ are the $k^{th}$ element of $\widehat{X}_{t}(\hat{\theta}_1)$, $X_{t}$, respectively. Consider the term $(\widehat{\Gamma}_{X} - \Gamma_{X})$: $||\widehat{\Gamma}_{X} - \Gamma_{X}||_F \leq \sum_{k=1}^{d_1} \sum_{l=1}^{d_1}||\widehat{\Gamma}_{kl,X}-\Gamma_{kl,X}|| $. By triangular equality: $||\widehat{\Gamma}_{kl,X}-\Gamma_{kl,X}||=||\widehat{\Gamma}_{kl,X}-C_{kl,X}+C_{kl,X}-\Gamma_{kl,X}||\leq ||\widehat{\Gamma}_{kl,X}-C_{kl,X}||+||C_{kl,X}-\Gamma_{kl,X}||$. For the first term, by Cauchy-Schwarz inequality, for the diagonal terms ($k=l$): 
\begin{align*}
&\widehat{\Gamma}_{kl,X}-C_{kl,X}=\frac{1}{T} \sum_{t=1}^{T} 2(\widehat{X}_{k,t}-X_{k,t})X_{k,t}+(\widehat{X}_{k,t}-X_{k,t})^2 \\
& \leq \frac{1}{T}  \left(\sum_{t=1}^T (X_{k,t})^2   \right)^{1/2}\left( \sum_{t=1}^T (\widehat{X}_{k,t}-X_{k,t})^2  \right)^{1/2}+\frac{1}{T} \sum_{t=1}^{T}(\widehat{X}_{k,t}-X_{k,t})^2=O_p\left(\frac{1}{\sqrt{T}}\right) 
\end{align*}
since we have that the term $\frac{1}{T}  \sum_{t=1}^T (X_{k,t})^2 =O_p(1)$ because of Chebyshev inequality and boundedness of moments, while the term $\frac{1}{T} \sum_{t=1}^{T}(\widehat{X}_{k,t}-X_{k,t})^2=O_p(T^{-1})$ as consequence of what proved in Proposition \ref{prop_conditional1}. Similar logic can be applied for $k\not=l$. By application of Chebyshev inequality: $||C_{kl,X}-\Gamma_{kl,X}||=O_p(T^{-1/2}), \; \forall k,l$. 
This in turn implies: $\widehat{\Gamma}_{kl,X}-\Gamma_{kl,X}=O_p(T^{-1/2})$, and $\widehat{\Gamma}_{kl,Z}-\Gamma_{kl,Z}=O_p(T^{-1/2})$, with $\frac{1}{T}  \sum_{t=1}^T (Z_{k,t})^2 =O_p(1)$, as well as $\Gamma_{kl,X}=O(1)$ and $\Gamma_{kl,Z}=O(1)$. By virtue of the continuous mapping theorem: $\widehat{\Gamma}_{Z}^{-1}\otimes\widehat{\Gamma}_{X}^{-1} - \Gamma_{Z}^{-1} \otimes \Gamma_{X}^{-1}=O_p(T^{-1/2})$. Since $\Gamma_X$ and $\Gamma_Z$ are bounded (i.e., $\Gamma_X=O(1),\Gamma_Z=O(1)$), studying the boundedness of the term, $\widehat{\Gamma}_{\hat{X}\hat{Z}}(j)$, is equivalent to studying the boundedness of the term, $\widehat{C}_{UV}(j)$. Thus, the focus shifts to the term:
\begin{align*}
& \sum_{j=1}^{T-1} \omega(j) \big(\text{vec}(\widehat{\Gamma}_{\hat{X}\hat{Z}}(j))^\prime \text{vec}(\widehat{\Gamma}_{\hat{X}\hat{Z}}(j))\big)\asymp \sum_{j=1}^{T-1} \omega(j) \big(\text{vec}(\widehat{C}_{UV}(j))^\prime \text{vec}(\widehat{C}_{UV}(j))\big) \\
&\; \; \; \;  \; \; = \sum_{j=1}^{T-1} \omega(j)\Big(\text{vec}(\widehat{C}_{UV}(j))^\prime \text{vec}(\widehat{C}_{UV}(j))-\text{vec}(C_{UV}(j))^\prime \text{vec}(C_{UV}(j)) \\
& \; \; \; \; \; \; \; \;  \; \; \; \;  \; \; \; \;  \; \; \; \;  \; \; \; \;  \; \; \; \;  \; \; \; \;  +\text{vec}(C_{UV}(j))^\prime \text{vec}(C_{UV}(j)) \Big)
\end{align*}
Both parts of the sum are $O_p(MT^{-1})$ by Proposition \ref{prop_conditional1}, since:
\begin{align*}
    & \sum_{j=1}^{T-1} \omega(j)\left(\text{vec}(\widehat{C}_{UV}(j))^\prime \text{vec}(\widehat{C}_{UV}(j))-\text{vec}(C_{UV}(j))^\prime \text{vec}(C_{UV}(j)) \right) \\
    & \; = \sum_{k=1}^{d_1} \sum_{l=1}^{d_2} \sum_{j=1}^{T-1} \omega(j)\left\{\left(\widehat{C}_{kl,UV}(j)-C_{kl,UV}(j)\right)^2+2 \widehat{C}_{kl,UV}(j)\left(\widehat{C}_{kl,UV}(j)-C_{kl,UV}(j)\right)\right\},\\
    &\sum_{j=1}^{T-1} \omega(j)\big(\text{vec}(C_{UV}(j))^\prime \text{vec}(C_{UV}(j)) \big)= \sum_{r=1}^{d_1d_2} \sum_{j=1}^{T-1} \omega(j) (C_{r,UV}(j))^2
\end{align*}
This means that: $\sum_{j=1}^{T-1} \omega(j) \text{vec}(\widehat{C}_{UV}(j))^\prime \text{vec}(\widehat{C}_{UV}(j))=O_p(MT^{-1})$. \\Hence: $T \left(\widehat{\mathcal{T}}_{\omega}-\widehat{\mathcal{T}}_{\omega}^\star \right)= T \cdot O_p(MT^{-1}) \cdot O_p(T^{-1/2})=O_p(MT^{-1/2})$, 
which concludes the proof as $\widehat{D}_{\omega,T}^{(Hete)}=O(M)$ by  Proposition \ref{prop_conditional4} and Proposition \ref{prop_moments}.

\subsubsection{Proof of Proposition \ref{prop_conditional3}}\label{proof_prop_conditional3}

We consider the decomposition: $\widehat{\mathcal{T}}_{2\omega}^{c}-\mathcal{T}_{2\omega}^{c\star}=\left(\widehat{\mathcal{\mathcal{T}}}_{2\omega}^{c}-\breve{\mathcal{T}}_{2\omega}^{c}\right)+\left(\breve{\mathcal{T}}_{2\omega}^{c}-\mathcal{T}_{2\omega}^{c\star}\right)$, with:\\ $ \breve{\mathcal{T}}_{2\omega}= \frac{1}{T^2} \sum_{j=1}^{T-2} \omega(j) \sum_{s,t=j+1,s\not=t,j>|t-s|}^{T} \langle\breve{U}_t,\breve{U}_s \rangle  \langle \breve{V}_{t-j},\breve{V}_{s-j} \rangle $. By a similar argument of Proposition \ref{prop_conditional2}: $ (\widehat{D}_{\omega,T}^{(Hete)})^{-1/2}T\left(\widehat{\mathcal{T}}_{2\omega}^c-\breve{\mathcal{T}}_{2\omega}^c\right)\xrightarrow{p} 0$,
since the convergence is uniquely driven by the distance between $\widehat{\Gamma}_{X}$ and $\Gamma_{X}$, and between $\widehat{\Gamma}_{Z}$ and $\Gamma_{Z}$. It remains to prove: $T\left(\breve{\mathcal{T}}_{2\omega}^c-\mathcal{T}_{2\omega}^{c\star}\right)=o_p(M^{1/2})$. We have:\\
\begin{align*}
&T\left(\breve{\mathcal{T}}_{2\omega}^c-\mathcal{T}_{2\omega}^{c\star}\right)= \frac{1}{T} \sum_{j=1}^{T-2} \omega(j) \sum_{s,t=j+1,s\not=t,j>|t-s|}^{T}  \left( \breve{U}_t^\prime\breve{U}_s \breve{V}_{t-j}^\prime\breve{V}_{s-j}-(U_t)^\prime U_s (V_{t-j})^\prime V_{s-j}\right)\\
&=\sum_{k=1}^{d_1}\sum_{l=1}^{d_2}\frac{1}{T} \sum_{j=1}^{T-2} \omega(j) \sum_{s,t=j+1,s\not=t,j>|t-s|}^{T} \left( \breve{U}_{k,t}\breve{U}_{k,s} \breve{V}_{l,t-j}\breve{V}_{l,s-j}- U_{k,t} U_{k,s} V_{l,t-j} V_{l,s-j}\right)\\
&=\sum_{k=1}^{d_1}\sum_{l=1}^{d_2}\frac{1}{T} \sum_{j=1}^{T-2} \omega(j) \sum_{s,t=j+1,s\not=t,j>|t-s|}^{T} \Bigg( \left[\breve{U}_{k,t}\breve{U}_{k,s} - U_{k,t}U_{k,s} \right]V_{l,t-j} V_{l,s-j} +U_{k,t}U_{k,s}\\
& \cdot\left[\breve{V}_{l,t-j}\breve{V}_{l,s-j}-
V_{l,t-j} V_{l,s-j}\right]+ \left[\breve{U}_{k,t}\breve{U}_{k,s} - U_{k,t}U_{k,s}\right]\left[\breve{V}_{l,t-j}\breve{V}_{l,s-j}-
V_{l,t-j} V_{l,s-j}\right]
\Bigg)\\
&=\sum_{k=1}^{d_1}\sum_{l=1}^{d_2} \sum_{j=1}^{T-2} \omega(j) \sum_{s,t=j+1,s\not=t,j>|t-s|}^{T}T^{-1}\Bigg( \Big[
(\breve{U}_{k,t}- U_{k,t})U_{k,s} \\
& + (\breve{U}_{k,s} - U_{k,s})U_{k,t}+
(\breve{U}_{k,t}- U_{k,t})(\breve{U}_{k,s}-U_{k,s})
\Big]V_{l,t-j} V_{l,s-j}+U_{k,t}U_{k,s}\Big[\\
&
(\breve{V}_{l,t-j}-V_{l,t-j}) V_{l,s-j}+(\breve{V}_{l,s-j}-V_{l,s-j})V_{l,t-j}+(\breve{V}_{l,t-j}-V_{l,t-j}) (\breve{V}_{l,s-j}-V_{l,s-j})
\Big]
\\
&+\left[\breve{U}_{k,t}\breve{U}_{k,s} - U_{k,t}U_{k,s}\right]\left[\breve{V}_{l,t-j}\breve{V}_{l,s-j}-
V_{l,t-j} V_{l,s-j}\right]\Bigg)
\end{align*}
Proposition \ref{prop_conditional1} concludes the proof as all terms inside the brackets are $O_p(T^{-3/2})$.

\subsection{Monte Carlo study}

\subsubsection{DGPs under the Null}
We consider four DGP families where $X$ is a standardized Strong White Noise (SWN): $X_{t} = \epsilon_x, \quad \epsilon_x\sim \text{i.i.d.}(0,1)$, ensuring weak exogeneity holds by construction. Following Lemma \ref{coroll_toyexample1}, the DGPs vary along two dimensions: (i) the inverse causality channel from past $X$ to present $Z$, and (ii) the conditional mean and variance properties of $Z$. The four specifications for the univariate process $Z$ are:
\begin{enumerate}[a)]
\item \textsc{DGP 1a (Linear-in-Mean):}
\begin{align*}
Z_t=\alpha Z_{t-1}+\beta_1 X_{t-1}+\epsilon_z, \quad \epsilon_z \sim \text{i.i.d.}(0,1)
\end{align*}
\item \textsc{DGP 2a (Nonlinear-in-Mean):}
\begin{align*}
Z_t=\alpha Z_{t-1}+\beta_1 X_{t-1}+\sum_{h=1}^2\frac{\beta_1}{2h}X_{t-1}^{2h}+\epsilon_z, \quad \epsilon_z \sim \text{i.i.d.}(0,1)
\end{align*}
\item \textsc{DGP 3a (ARCH-type Inverse Causality):}
\begin{align*}
Z_t&=\alpha Z_{t-1}+\sigma_{z,t}\epsilon_z, \quad \epsilon_z \sim \text{i.i.d.}(0,1) \\
\sigma_{z,t}^2&=1+\sum_{h=1}^3\frac{\beta_1}{h!} X_{t-h}^2
\end{align*}
\item \textsc{DGP 4a (GARCH-type Inverse Causality):}
\begin{align*}
Z_t&=\sigma_{z,t}\epsilon_z, \quad \epsilon_z \sim \text{i.i.d.}(0,1) \\
\sigma_{z,t}^2&=1+\alpha \sigma_{z,t-1}^2+\beta_2 X_{t-1}^2
\end{align*}
\end{enumerate}
Parameter values span: $\alpha\in\{0, 0.2, 0.3, 0.4, 0.5, 0.6, 0.7\}$, $\beta_1\in\{0, 0.2, 0.4, 0.6, 1.0\}$, $\beta_2\in\{0.05, 0.1, 0.2, 0.3\}$.
When $\beta_1=0$ the two processes are mutually independent. When $\beta_1\not=0$ ($\beta_2\not=0$) inverse causality is present while weak exogeneity is maintained.

To examine the magnitude of the correction term, given that higher-order moments of $X$ are relevant (Proposition \ref{prop_moments}), innovations are drawn from a multivariate $t$-distribution with 6 degrees of freedom: $(\epsilon_x, \epsilon_z) \sim t_6(0,I_2)$.\footnote{Results under Gaussian innovations, $(\epsilon_x, \epsilon_z) \sim \mathcal{N}(0,I_2)$, are qualitatively similar.} This choice aligns with the underlying assumption of the empirical analysis of this paper (Section \ref{empirics}), and with empirical evidence from macroeconomic applications. For instance, \cite{brunnermeier2021feedbacks} estimate structural shocks identified through heteroskedasticity in SVARs as scaled 
$t$-variates with 5.7 degrees of freedom, thus motivating our calibration.
The smoothing parameter $M$ takes value $12$ or $36$, which corresponds to lags up to 1 year or 3 years, when the data is observed at monthly frequency. For each design, $1,000$ Monte Carlo simulations are run, with $T=1000$. All results report rejection rates at the 5\% nominal significance level.

\begin{table}[h!]
\centering
\scalebox{0.73}{
\begin{threeparttable}
\caption[Empirical rates]{Rejection frequencies for \textsc{DGP1a}:  \footnotesize{This table presents the rejection frequencies of five testing procedures, when the time series are generated by \textsc{DGP1a}; sample size: $T=1000$; $1000$ iterations; the weighting function is the quadratic spectral kernel;  the smoothing parameter range is: $M=\{12,36\}$; nominal significance level is 5\%.}}
\footnotesize
\label{size_DGP1a_T100}
\begin{tabular}{lcccccccccccc}
  $M=12$ & $\approx 2\ln{T}$ &  &  &  &  &  &  &  &  &  &  &  \\
  \toprule
   & \multicolumn{5}{c}{Hete} &  &  & \multicolumn{5}{c}{Hong} \\
  \cmidrule{2-6}\cmidrule{9-13}
   & $\beta_1=0$ & $\beta_1=0.2$ & $\beta_1=0.4$ & $\beta_1=0.6$ & $\beta_1=1$ &  &  & $\beta_1=0$ & $\beta_1=0.2$ & $\beta_1=0.4$ & $\beta_1=0.6$ & $\beta_1=1$ \\
  \cmidrule{2-6}\cmidrule{9-13}
  $\alpha=0$ & 0.046 & 0.056 & 0.041 & 0.048 & 0.044 &  &  & 0.051 & 0.044 & 0.034 & 0.037 & 0.043 \\
  $\alpha=0.2$ & 0.049 & 0.064 & 0.059 & 0.044 & 0.048 &  &  & 0.044 & 0.051 & 0.044 & 0.050 & 0.034 \\
  $\alpha=0.3$ & 0.046 & 0.065 & 0.036 & 0.052 & 0.051 &  &  & 0.052 & 0.051 & 0.065 & 0.059 & 0.070 \\
  $\alpha=0.4$ & 0.041 & 0.043 & 0.055 & 0.041 & 0.045 &  &  & 0.052 & 0.066 & 0.060 & 0.062 & 0.054 \\
  $\alpha=0.5$ & 0.054 & 0.045 & 0.040 & 0.046 & 0.052 &  &  & 0.062 & 0.079 & 0.059 & 0.068 & 0.065 \\
  $\alpha=0.6$ & 0.049 & 0.043 & 0.055 & 0.055 & 0.050 &  &  & 0.068 & 0.079 & 0.089 & 0.085 & 0.081 \\
  $\alpha=0.7$ & 0.046 & 0.041 & 0.052 & 0.052 & 0.035 &  &  & 0.075 & 0.099 & 0.090 & 0.087 & 0.092 \\
  \cmidrule{2-13}
\end{tabular}%

\vspace{0.4cm}

\begin{tabular}{lccccc}
  $M=12$ & $\approx 2\ln{T}$ &  &  &  &  \\
  \toprule
   & \multicolumn{5}{c}{HeteF} \\
  \cmidrule{2-6}
   & $\beta_1=0$ & $\beta_1=0.2$ & $\beta_1=0.4$ & $\beta_1=0.6$ & $\beta_1=1$ \\
  \cmidrule{2-6}
  $\alpha=0$ & 0.050 & 0.054 & 0.040 & 0.046 & 0.045 \\
  $\alpha=0.2$ & 0.039 & 0.050 & 0.045 & 0.043 & 0.031 \\
  $\alpha=0.3$ & 0.033 & 0.041 & 0.030 & 0.049 & 0.040 \\
  $\alpha=0.4$ & 0.037 & 0.024 & 0.033 & 0.026 & 0.030 \\
  $\alpha=0.5$ & 0.022 & 0.028 & 0.021 & 0.030 & 0.018 \\
  $\alpha=0.6$ & 0.020 & 0.020 & 0.030 & 0.034 & 0.029 \\
  $\alpha=0.7$ & 0.021 & 0.018 & 0.026 & 0.028 & 0.022 \\
  \cmidrule{2-6}
\end{tabular}%

\vspace{0.4cm}

\begin{tabular}{lcccccccccccc}
  $M=36$ & $\approx \sqrt{T}$ &  &  &  &  &  &  &  &  &  &  &  \\
  \toprule
   & \multicolumn{5}{c}{Hete} &  &  & \multicolumn{5}{c}{Hong} \\
  \cmidrule{2-6}\cmidrule{9-13}
   & $\beta_1=0$ & $\beta_1=0.2$ & $\beta_1=0.4$ & $\beta_1=0.6$ & $\beta_1=1$ &  &  & $\beta_1=0$ & $\beta_1=0.2$ & $\beta_1=0.4$ & $\beta_1=0.6$ & $\beta_1=1$ \\
  \cmidrule{2-6}\cmidrule{9-13}
  $\alpha=0$ & 0.043 & 0.051 & 0.041 & 0.062 & 0.034 &  &  & 0.046 & 0.049 & 0.043 & 0.045 & 0.050 \\
  $\alpha=0.2$ & 0.044 & 0.058 & 0.061 & 0.055 & 0.045 &  &  & 0.052 & 0.055 & 0.051 & 0.051 & 0.050 \\
  $\alpha=0.3$ & 0.059 & 0.055 & 0.048 & 0.049 & 0.044 &  &  & 0.058 & 0.055 & 0.058 & 0.054 & 0.068 \\
  $\alpha=0.4$ & 0.037 & 0.060 & 0.060 & 0.054 & 0.048 &  &  & 0.055 & 0.075 & 0.083 & 0.072 & 0.056 \\
  $\alpha=0.5$ & 0.056 & 0.056 & 0.049 & 0.052 & 0.049 &  &  & 0.084 & 0.094 & 0.079 & 0.090 & 0.087 \\
  $\alpha=0.6$ & 0.039 & 0.033 & 0.048 & 0.056 & 0.052 &  &  & 0.090 & 0.102 & 0.129 & 0.115 & 0.114 \\
  $\alpha=0.7$ & 0.050 & 0.037 & 0.040 & 0.046 & 0.049 &  &  & 0.135 & 0.151 & 0.133 & 0.131 & 0.147 \\
  \cmidrule{2-13}
\end{tabular}%

\vspace{0.4cm}

\begin{tabular}{lccccc}
  $M=36$ & $\approx \sqrt{T}$ &  &  &  &  \\
  \toprule
   & \multicolumn{5}{c}{HeteF} \\
  \cmidrule{2-6}
   & $\beta_1=0$ & $\beta_1=0.2$ & $\beta_1=0.4$ & $\beta_1=0.6$ & $\beta_1=1$ \\
  \cmidrule{2-6}
  $\alpha=0$ & 0.046 & 0.050 & 0.037 & 0.065 & 0.033 \\
  $\alpha=0.2$ & 0.040 & 0.041 & 0.048 & 0.043 & 0.039 \\
  $\alpha=0.3$ & 0.041 & 0.039 & 0.043 & 0.035 & 0.039 \\
  $\alpha=0.4$ & 0.018 & 0.033 & 0.033 & 0.021 & 0.026 \\
  $\alpha=0.5$ & 0.022 & 0.021 & 0.024 & 0.030 & 0.028 \\
  $\alpha=0.6$ & 0.024 & 0.020 & 0.022 & 0.029 & 0.015 \\
  $\alpha=0.7$ & 0.020 & 0.016 & 0.029 & 0.025 & 0.029 \\
  \cmidrule{2-6}
\end{tabular}%

\vspace{0.4cm}

\begin{tabular}{lcccccccccccc}
  \toprule
   & \multicolumn{5}{c}{Wald Single} &  &  & \multicolumn{5}{c}{Wald Double} \\
  \cmidrule{2-6}\cmidrule{9-13}
   & $\beta_1=0$ & $\beta_1=0.2$ & $\beta_1=0.4$ & $\beta_1=0.6$ & $\beta_1=1$ &  &  & $\beta_1=0$ & $\beta_1=0.2$ & $\beta_1=0.4$ & $\beta_1=0.6$ & $\beta_1=1$ \\
  \cmidrule{2-6}\cmidrule{9-13}
  $\alpha=0$ & 0.046 & 0.060 & 0.060 & 0.051 & 0.050 &  &  & 0.043 & 0.044 & 0.054 & 0.055 & 0.049 \\
  $\alpha=0.2$ & 0.050 & 0.039 & 0.054 & 0.055 & 0.045 &  &  & 0.048 & 0.045 & 0.044 & 0.059 & 0.037 \\
  $\alpha=0.3$ & 0.051 & 0.044 & 0.051 & 0.058 & 0.059 &  &  & 0.055 & 0.055 & 0.043 & 0.061 & 0.076 \\
  $\alpha=0.4$ & 0.041 & 0.054 & 0.056 & 0.052 & 0.046 &  &  & 0.048 & 0.054 & 0.056 & 0.045 & 0.043 \\
  $\alpha=0.5$ & 0.058 & 0.045 & 0.055 & 0.049 & 0.060 &  &  & 0.036 & 0.054 & 0.051 & 0.043 & 0.049 \\
  $\alpha=0.6$ & 0.044 & 0.041 & 0.044 & 0.043 & 0.060 &  &  & 0.037 & 0.040 & 0.049 & 0.056 & 0.048 \\
  $\alpha=0.7$ & 0.058 & 0.050 & 0.049 & 0.054 & 0.052 &  &  & 0.048 & 0.058 & 0.041 & 0.045 & 0.051 \\
  \cmidrule{2-13}
\end{tabular}%
\end{threeparttable}
}
\end{table}


\begin{table}[h!]
\centering
\scalebox{0.73}{
\begin{threeparttable}
\caption[Empirical rates]{Rejection frequencies for \textsc{DGP2a}:  \footnotesize{This table presents the rejection frequencies of five testing procedures, when the time series are generated by \textsc{DGP2a}; sample size: $T=1000$; $1000$ iterations; the weighting function is the quadratic spectral kernel;  the smoothing parameter range is: $M=\{12,36\}$; nominal significance level is 5\%.}}
\footnotesize
\label{size_DGP2a_T100}
\begin{tabular}{lcccccccccc}
  $M=12$ & $\approx 2\ln{T}$ &  &  &  &  &  &  &  &  &  \\
  \toprule
   & \multicolumn{4}{c}{Hete} &  &  & \multicolumn{4}{c}{Hong} \\
  \cmidrule{2-5}\cmidrule{8-11}
   & $\beta_1=0.2$ & $\beta_1=0.4$ & $\beta_1=0.6$ & $\beta_1=1$ &  &  & $\beta_1=0.2$ & $\beta_1=0.4$ & $\beta_1=0.6$ & $\beta_1=1$ \\
  \cmidrule{2-5}\cmidrule{8-11}
  $\alpha=0.2$ & 0.056 & 0.069 & 0.070 & 0.064 &  &  & 0.061 & 0.066 & 0.052 & 0.044 \\
  $\alpha=0.3$ & 0.077 & 0.071 & 0.075 & 0.069 &  &  & 0.055 & 0.061 & 0.065 & 0.056 \\
  $\alpha=0.4$ & 0.056 & 0.075 & 0.074 & 0.052 &  &  & 0.056 & 0.069 & 0.065 & 0.061 \\
  $\alpha=0.5$ & 0.055 & 0.059 & 0.079 & 0.062 &  &  & 0.054 & 0.054 & 0.070 & 0.058 \\
  $\alpha=0.6$ & 0.074 & 0.037 & 0.080 & 0.071 &  &  & 0.084 & 0.068 & 0.099 & 0.077 \\
  $\alpha=0.7$ & 0.066 & 0.069 & 0.068 & 0.058 &  &  & 0.094 & 0.086 & 0.083 & 0.079 \\
  \cmidrule{2-11}
\end{tabular}%

\vspace{0.4cm}

\begin{tabular}{lcccc}
  $M=12$ & $\approx 2\ln{T}$ &  &  &  \\
  \toprule
   & \multicolumn{4}{c}{HeteF} \\
  \cmidrule{2-5}
   & $\beta_1=0.2$ & $\beta_1=0.4$ & $\beta_1=0.6$ & $\beta_1=1$ \\
  \cmidrule{2-5}
  $\alpha=0.2$ & 0.054 & 0.059 & 0.075 & 0.060 \\
  $\alpha=0.3$ & 0.062 & 0.071 & 0.077 & 0.065 \\
  $\alpha=0.4$ & 0.055 & 0.071 & 0.070 & 0.045 \\
  $\alpha=0.5$ & 0.061 & 0.058 & 0.070 & 0.059 \\
  $\alpha=0.6$ & 0.072 & 0.034 & 0.075 & 0.064 \\
  $\alpha=0.7$ & 0.061 & 0.072 & 0.064 & 0.056 \\
  \cmidrule{2-5}
\end{tabular}%

\vspace{0.4cm}

\begin{tabular}{lcccccccccc}
  $M=36$ & $\approx \sqrt{T}$ &  &  &  &  &  &  &  &  &  \\
  \toprule
   & \multicolumn{4}{c}{Hete} &  &  & \multicolumn{4}{c}{Hong} \\
  \cmidrule{2-5}\cmidrule{8-11}
   & $\beta_1=0.2$ & $\beta_1=0.4$ & $\beta_1=0.6$ & $\beta_1=1$ &  &  & $\beta_1=0.2$ & $\beta_1=0.4$ & $\beta_1=0.6$ & $\beta_1=1$ \\
  \cmidrule{2-5}\cmidrule{8-11}
  $\alpha=0.2$ & 0.072 & 0.098 & 0.085 & 0.075 &  &  & 0.076 & 0.070 & 0.060 & 0.060 \\
  $\alpha=0.3$ & 0.083 & 0.086 & 0.080 & 0.069 &  &  & 0.076 & 0.080 & 0.079 & 0.060 \\
  $\alpha=0.4$ & 0.056 & 0.081 & 0.083 & 0.079 &  &  & 0.087 & 0.081 & 0.076 & 0.085 \\
  $\alpha=0.5$ & 0.075 & 0.069 & 0.091 & 0.095 &  &  & 0.090 & 0.080 & 0.087 & 0.090 \\
  $\alpha=0.6$ & 0.072 & 0.062 & 0.070 & 0.085 &  &  & 0.111 & 0.116 & 0.116 & 0.120 \\
  $\alpha=0.7$ & 0.076 & 0.084 & 0.095 & 0.070 &  &  & 0.170 & 0.164 & 0.166 & 0.151 \\
  \cmidrule{2-11}
\end{tabular}%

\vspace{0.4cm}

\begin{tabular}{lcccc}
  $M=36$ & $\approx \sqrt{T}$ &  &  &  \\
  \toprule
   & \multicolumn{4}{c}{HeteF} \\
  \cmidrule{2-5}
   & $\beta_1=0.2$ & $\beta_1=0.4$ & $\beta_1=0.6$ & $\beta_1=1$ \\
  \cmidrule{2-5}
  $\alpha=0.2$ & 0.069 & 0.095 & 0.083 & 0.066 \\
  $\alpha=0.3$ & 0.074 & 0.072 & 0.076 & 0.066 \\
  $\alpha=0.4$ & 0.058 & 0.075 & 0.079 & 0.074 \\
  $\alpha=0.5$ & 0.076 & 0.055 & 0.074 & 0.081 \\
  $\alpha=0.6$ & 0.065 & 0.060 & 0.076 & 0.070 \\
  $\alpha=0.7$ & 0.069 & 0.076 & 0.087 & 0.055 \\
  \cmidrule{2-5}
\end{tabular}%

\vspace{0.4cm}

\begin{tabular}{lcccccccccc}
  \toprule
   & \multicolumn{4}{c}{Wald Single} &  &  & \multicolumn{4}{c}{Wald Double} \\
  \cmidrule{2-5}\cmidrule{8-11}
   & $\beta_1=0.2$ & $\beta_1=0.4$ & $\beta_1=0.6$ & $\beta_1=1$ &  &  & $\beta_1=0.2$ & $\beta_1=0.4$ & $\beta_1=0.6$ & $\beta_1=1$ \\
  \cmidrule{2-5}\cmidrule{8-11}
  $\alpha=0.2$ & 0.064 & 0.055 & 0.040 & 0.046 &  &  & 0.054 & 0.049 & 0.041 & 0.037 \\
  $\alpha=0.3$ & 0.046 & 0.050 & 0.033 & 0.048 &  &  & 0.043 & 0.036 & 0.034 & 0.051 \\
  $\alpha=0.4$ & 0.050 & 0.045 & 0.061 & 0.043 &  &  & 0.040 & 0.028 & 0.056 & 0.036 \\
  $\alpha=0.5$ & 0.049 & 0.045 & 0.031 & 0.039 &  &  & 0.037 & 0.039 & 0.039 & 0.046 \\
  $\alpha=0.6$ & 0.044 & 0.041 & 0.044 & 0.041 &  &  & 0.028 & 0.021 & 0.040 & 0.040 \\
  $\alpha=0.7$ & 0.040 & 0.037 & 0.031 & 0.043 &  &  & 0.033 & 0.021 & 0.025 & 0.024 \\
  \cmidrule{2-11}
\end{tabular}%
\end{threeparttable}
}
\end{table}


\begin{table}[h!]
\centering
\scalebox{0.73}{
\begin{threeparttable}
\caption[Empirical rates]{Rejection frequencies for \textsc{DGP3a}:  \footnotesize{This table presents the rejection frequencies of five testing procedures, when the time series are generated by \textsc{DGP3a}; sample size: $T=1000$; $1000$ iterations; the weighting function is the quadratic spectral kernel;  the smoothing parameter range is: $M=\{12,36\}$; nominal significance level is 5\%.}}
\footnotesize
\label{size_DGP3a_T100}
\begin{tabular}{lcccccccccc}
  $M=12$ & $\approx 2\ln{T}$ &  &  &  &  &  &  &  &  &  \\
  \toprule
   & \multicolumn{4}{c}{Hete} &  &  & \multicolumn{4}{c}{Hong} \\
  \cmidrule{2-5}\cmidrule{8-11}
   & $\beta_1=0.2$ & $\beta_1=0.4$ & $\beta_1=0.6$ & $\beta_1=1$ &  &  & $\beta_1=0.2$ & $\beta_1=0.4$ & $\beta_1=0.6$ & $\beta_1=1$ \\
  \cmidrule{2-5}\cmidrule{8-11}
  $\alpha=0.2$ & 0.071 & 0.045 & 0.048 & 0.045 &  &  & 0.060 & 0.050 & 0.056 & 0.050 \\
  $\alpha=0.3$ & 0.049 & 0.041 & 0.052 & 0.054 &  &  & 0.059 & 0.070 & 0.069 & 0.043 \\
  $\alpha=0.4$ & 0.054 & 0.062 & 0.049 & 0.055 &  &  & 0.069 & 0.068 & 0.061 & 0.058 \\
  $\alpha=0.5$ & 0.050 & 0.046 & 0.040 & 0.031 &  &  & 0.064 & 0.086 & 0.068 & 0.066 \\
  $\alpha=0.6$ & 0.059 & 0.066 & 0.046 & 0.041 &  &  & 0.072 & 0.079 & 0.068 & 0.092 \\
  $\alpha=0.7$ & 0.046 & 0.051 & 0.043 & 0.060 &  &  & 0.095 & 0.092 & 0.084 & 0.090 \\
  \cmidrule{2-11}
\end{tabular}%

\vspace{0.4cm}

\begin{tabular}{lcccc}
  $M=12$ & $\approx 2\ln{T}$ &  &  &  \\
  \toprule
   & \multicolumn{4}{c}{HeteF} \\
  \cmidrule{2-5}
   & $\beta_1=0.2$ & $\beta_1=0.4$ & $\beta_1=0.6$ & $\beta_1=1$ \\
  \cmidrule{2-5}
  $\alpha=0.2$ & 0.061 & 0.040 & 0.036 & 0.034 \\
  $\alpha=0.3$ & 0.035 & 0.024 & 0.030 & 0.048 \\
  $\alpha=0.4$ & 0.033 & 0.031 & 0.030 & 0.035 \\
  $\alpha=0.5$ & 0.030 & 0.019 & 0.019 & 0.024 \\
  $\alpha=0.6$ & 0.036 & 0.036 & 0.028 & 0.026 \\
  $\alpha=0.7$ & 0.029 & 0.018 & 0.028 & 0.034 \\
  \cmidrule{2-5}
\end{tabular}%

\vspace{0.4cm}

\begin{tabular}{lcccccccccc}
  $M=36$ & $\approx \sqrt{T}$ &  &  &  &  &  &  &  &  &  \\
  \toprule
   & \multicolumn{4}{c}{Hete} &  &  & \multicolumn{4}{c}{Hong} \\
  \cmidrule{2-5}\cmidrule{8-11}
   & $\beta_1=0.2$ & $\beta_1=0.4$ & $\beta_1=0.6$ & $\beta_1=1$ &  &  & $\beta_1=0.2$ & $\beta_1=0.4$ & $\beta_1=0.6$ & $\beta_1=1$ \\
  \cmidrule{2-5}\cmidrule{8-11}
  $\alpha=0.2$ & 0.058 & 0.054 & 0.037 & 0.043 &  &  & 0.061 & 0.043 & 0.046 & 0.058 \\
  $\alpha=0.3$ & 0.049 & 0.035 & 0.037 & 0.037 &  &  & 0.059 & 0.055 & 0.056 & 0.051 \\
  $\alpha=0.4$ & 0.048 & 0.046 & 0.052 & 0.037 &  &  & 0.080 & 0.075 & 0.061 & 0.081 \\
  $\alpha=0.5$ & 0.046 & 0.052 & 0.046 & 0.045 &  &  & 0.085 & 0.090 & 0.089 & 0.091 \\
  $\alpha=0.6$ & 0.049 & 0.061 & 0.052 & 0.054 &  &  & 0.101 & 0.115 & 0.117 & 0.119 \\
  $\alpha=0.7$ & 0.044 & 0.046 & 0.045 & 0.060 &  &  & 0.131 & 0.159 & 0.124 & 0.149 \\
  \cmidrule{2-11}
\end{tabular}%

\vspace{0.4cm}

\begin{tabular}{lcccc}
  $M=36$ & $\approx \sqrt{T}$ &  &  &  \\
  \toprule
   & \multicolumn{4}{c}{HeteF} \\
  \cmidrule{2-5}
   & $\beta_1=0.2$ & $\beta_1=0.4$ & $\beta_1=0.6$ & $\beta_1=1$ \\
  \cmidrule{2-5}
  $\alpha=0.2$ & 0.046 & 0.043 & 0.035 & 0.033 \\
  $\alpha=0.3$ & 0.033 & 0.025 & 0.026 & 0.035 \\
  $\alpha=0.4$ & 0.030 & 0.025 & 0.024 & 0.029 \\
  $\alpha=0.5$ & 0.025 & 0.020 & 0.021 & 0.025 \\
  $\alpha=0.6$ & 0.028 & 0.025 & 0.019 & 0.025 \\
  $\alpha=0.7$ & 0.028 & 0.022 & 0.025 & 0.026 \\
  \cmidrule{2-5}
\end{tabular}%

\vspace{0.4cm}

\begin{tabular}{lcccccccccc}
  \toprule
   & \multicolumn{4}{c}{Wald Single} &  &  & \multicolumn{4}{c}{Wald Double} \\
  \cmidrule{2-5}\cmidrule{8-11}
   & $\beta_1=0.2$ & $\beta_1=0.4$ & $\beta_1=0.6$ & $\beta_1=1$ &  &  & $\beta_1=0.2$ & $\beta_1=0.4$ & $\beta_1=0.6$ & $\beta_1=1$ \\
  \cmidrule{2-5}\cmidrule{8-11}
  $\alpha=0.2$ & 0.054 & 0.055 & 0.045 & 0.055 &  &  & 0.041 & 0.068 & 0.049 & 0.058 \\
  $\alpha=0.3$ & 0.046 & 0.050 & 0.060 & 0.051 &  &  & 0.049 & 0.046 & 0.051 & 0.043 \\
  $\alpha=0.4$ & 0.061 & 0.058 & 0.050 & 0.045 &  &  & 0.068 & 0.051 & 0.041 & 0.051 \\
  $\alpha=0.5$ & 0.040 & 0.054 & 0.050 & 0.051 &  &  & 0.040 & 0.049 & 0.046 & 0.051 \\
  $\alpha=0.6$ & 0.041 & 0.060 & 0.045 & 0.050 &  &  & 0.043 & 0.048 & 0.049 & 0.055 \\
  $\alpha=0.7$ & 0.054 & 0.043 & 0.039 & 0.050 &  &  & 0.048 & 0.039 & 0.056 & 0.068 \\
  \cmidrule{2-11}
\end{tabular}%
\end{threeparttable}
}
\end{table}


\begin{table}[h!]
\centering
\scalebox{0.73}{
\begin{threeparttable}
\caption[Empirical rates]{Rejection frequencies for \textsc{DGP4a}:  \footnotesize{This table presents the rejection frequencies of five testing procedures, when the time series are generated by \textsc{DGP4a}; sample size: $T=1000$; $1000$ iterations; the weighting function is the quadratic spectral kernel;  the smoothing parameter range is: $M=\{12,36\}$; nominal significance level is 5\%.}}
\footnotesize
\label{size_DGP4a_T100}
\begin{tabular}{lcccccccccc}
  $M=12$ & $\approx 2\ln{T}$ &  &  &  &  &  &  &  &  &  \\
  \toprule
   & \multicolumn{4}{c}{Hete} &  &  & \multicolumn{4}{c}{Hong} \\
  \cmidrule{2-5}\cmidrule{8-11}
   & $\beta_2=0.05$ & $\beta_2=0.1$ & $\beta_2=0.2$ & $\beta_2=0.3$ &  &  & $\beta_2=0.05$ & $\beta_2=0.1$ & $\beta_2=0.2$ & $\beta_2=0.3$ \\
  \cmidrule{2-5}\cmidrule{8-11}
  $\alpha=0.2$ & 0.049 & 0.046 & 0.040 & 0.044 &  &  & 0.045 & 0.052 & 0.037 & 0.050 \\
  $\alpha=0.3$ & 0.039 & 0.051 & 0.040 & 0.074 &  &  & 0.058 & 0.049 & 0.044 & 0.036 \\
  $\alpha=0.4$ & 0.055 & 0.059 & 0.056 & 0.061 &  &  & 0.050 & 0.048 & 0.044 & 0.059 \\
  $\alpha=0.5$ & 0.052 & 0.040 & 0.046 & 0.072 &  &  & 0.046 & 0.036 & 0.051 & 0.040 \\
  $\alpha=0.6$ & 0.043 & 0.059 & 0.054 & 0.048 &  &  & 0.051 & 0.059 & 0.052 & 0.051 \\
  $\alpha=0.7$ & 0.041 & 0.055 & 0.052 & 0.061 &  &  & 0.043 & 0.045 & 0.051 & 0.051 \\
  \cmidrule{2-11}
\end{tabular}%

\vspace{0.4cm}

\begin{tabular}{lcccc}
  $M=12$ & $\approx 2\ln{T}$ &  &  &  \\
  \toprule
   & \multicolumn{4}{c}{HeteF} \\
  \cmidrule{2-5}
   & $\beta_2=0.05$ & $\beta_2=0.1$ & $\beta_2=0.2$ & $\beta_2=0.3$ \\
  \cmidrule{2-5}
  $\alpha=0.2$ & 0.045 & 0.046 & 0.041 & 0.044 \\
  $\alpha=0.3$ & 0.043 & 0.049 & 0.040 & 0.075 \\
  $\alpha=0.4$ & 0.054 & 0.058 & 0.059 & 0.062 \\
  $\alpha=0.5$ & 0.049 & 0.037 & 0.045 & 0.070 \\
  $\alpha=0.6$ & 0.045 & 0.054 & 0.050 & 0.050 \\
  $\alpha=0.7$ & 0.046 & 0.054 & 0.046 & 0.059 \\
  \cmidrule{2-5}
\end{tabular}%

\vspace{0.4cm}

\begin{tabular}{lcccccccccc}
  $M=36$ & $\approx \sqrt{T}$ &  &  &  &  &  &  &  &  &  \\
  \toprule
   & \multicolumn{4}{c}{Hete} &  &  & \multicolumn{4}{c}{Hong} \\
  \cmidrule{2-5}\cmidrule{8-11}
   & $\beta_2=0.05$ & $\beta_2=0.1$ & $\beta_2=0.2$ & $\beta_2=0.3$ &  &  & $\beta_2=0.05$ & $\beta_2=0.1$ & $\beta_2=0.2$ & $\beta_2=0.3$ \\
  \cmidrule{2-5}\cmidrule{8-11}
  $\alpha=0.2$ & 0.068 & 0.048 & 0.041 & 0.046 &  &  & 0.050 & 0.050 & 0.040 & 0.049 \\
  $\alpha=0.3$ & 0.051 & 0.062 & 0.052 & 0.056 &  &  & 0.049 & 0.046 & 0.039 & 0.033 \\
  $\alpha=0.4$ & 0.058 & 0.054 & 0.044 & 0.060 &  &  & 0.048 & 0.051 & 0.059 & 0.050 \\
  $\alpha=0.5$ & 0.048 & 0.034 & 0.043 & 0.056 &  &  & 0.050 & 0.041 & 0.059 & 0.049 \\
  $\alpha=0.6$ & 0.045 & 0.072 & 0.060 & 0.062 &  &  & 0.035 & 0.048 & 0.046 & 0.050 \\
  $\alpha=0.7$ & 0.049 & 0.049 & 0.051 & 0.046 &  &  & 0.039 & 0.046 & 0.043 & 0.049 \\
  \cmidrule{2-11}
\end{tabular}%

\vspace{0.4cm}

\begin{tabular}{lcccc}
  $M=36$ & $\approx \sqrt{T}$ &  &  &  \\
  \toprule
   & \multicolumn{4}{c}{HeteF} \\
  \cmidrule{2-5}
   & $\beta_2=0.05$ & $\beta_2=0.1$ & $\beta_2=0.2$ & $\beta_2=0.3$ \\
  \cmidrule{2-5}
  $\alpha=0.2$ & 0.068 & 0.048 & 0.037 & 0.049 \\
  $\alpha=0.3$ & 0.050 & 0.059 & 0.051 & 0.051 \\
  $\alpha=0.4$ & 0.054 & 0.055 & 0.040 & 0.059 \\
  $\alpha=0.5$ & 0.048 & 0.035 & 0.044 & 0.054 \\
  $\alpha=0.6$ & 0.041 & 0.064 & 0.060 & 0.064 \\
  $\alpha=0.7$ & 0.050 & 0.043 & 0.051 & 0.044 \\
  \cmidrule{2-5}
\end{tabular}%

\vspace{0.4cm}

\begin{tabular}{lcccccccccc}
  \toprule
   & \multicolumn{4}{c}{Wald Single} &  &  & \multicolumn{4}{c}{Wald Double} \\
  \cmidrule{2-5}\cmidrule{8-11}
   & $\beta_2=0.05$ & $\beta_2=0.1$ & $\beta_2=0.2$ & $\beta_2=0.3$ &  &  & $\beta_2=0.05$ & $\beta_2=0.1$ & $\beta_2=0.2$ & $\beta_2=0.3$ \\
  \cmidrule{2-5}\cmidrule{8-11}
  $\alpha=0.2$ & 0.054 & 0.046 & 0.046 & 0.048 &  &  & 0.045 & 0.052 & 0.041 & 0.044 \\
  $\alpha=0.3$ & 0.049 & 0.050 & 0.050 & 0.052 &  &  & 0.058 & 0.040 & 0.048 & 0.045 \\
  $\alpha=0.4$ & 0.048 & 0.046 & 0.044 & 0.058 &  &  & 0.048 & 0.048 & 0.041 & 0.060 \\
  $\alpha=0.5$ & 0.059 & 0.046 & 0.048 & 0.045 &  &  & 0.054 & 0.048 & 0.054 & 0.048 \\
  $\alpha=0.6$ & 0.049 & 0.044 & 0.044 & 0.054 &  &  & 0.060 & 0.052 & 0.051 & 0.052 \\
  $\alpha=0.7$ & 0.074 & 0.064 & 0.046 & 0.050 &  &  & 0.058 & 0.050 & 0.059 & 0.065 \\
  \cmidrule{2-11}
\end{tabular}%
\end{threeparttable}
}
\end{table}


\newpage
\clearpage
\newpage
\clearpage

\subsubsection{DGPs under the Alternatives}
We consider three DGP families where weak exogeneity does not hold such that past $Z$ impacts present $X$, meaning that, in macroeconometric terms, the process $X$ cannot be considered as structural shocks. The three specifications for the bivariate process $\{X_{t},Z_{t}\}$ are:
\begin{enumerate}[a)]
\item \textsc{DGP 1b (Linear-in-Mean):}
\begin{align*}
X_{t}=\gamma_1 Z_{t-1}+\epsilon_x, \quad 
Z_{t}=0.4 Z_{t-1}-\beta_3 X_{t-1}+\epsilon_z, \quad (\epsilon_x,\epsilon_z) \sim \text{i.i.d.}(0,I_2)
\end{align*}
\item \textsc{DGP 2b (Nonlinear-in-Mean):}
\begin{align*}
X_{t}=\gamma_1 Z_{t-1}^2/8+\epsilon_x, \quad 
Z_{t}=0.4 Z_{t-1}-\beta_3 X_{t-1}/2+\epsilon_z, \quad (\epsilon_x,\epsilon_z) \sim \text{i.i.d.}(0,I_2)
\end{align*}
\item \textsc{DGP 4b (GARCH-type Causality):}
\begin{align*}
& X_{t}=\epsilon_x, \quad Z_t=\epsilon_z, \quad (\epsilon_x,\epsilon_z) \sim \mathcal{N}(0,\Sigma_t), \quad \Sigma_t = (\sigma_{x,t}^2, 0; 0, \sigma_{z,t}^2)  \\
&\sigma_{x,t}^2= 0.5+\gamma_2 \sigma_{x,t-1}^2+0.1 Z_{t-1}^2, \quad \sigma_{z,t}= 0.5+\beta_3 \sigma_{z,t-1}^2+0.1 Z_{t-1}^2
\end{align*}
\end{enumerate}
Parameter values span: $\gamma_1\in\{-0.4, -0.2, -0.05, 0.05, 0.2, 0.4\}$, \\
$\gamma_2\in\{0.05, 0.1, 0.2, 0.4, 0.6, 0.8\}$, 
$\beta_3\in\{0, 0.3, 0.8\}$.

The smoothing parameter $M$ takes value $\{12, 24, 36\}$, which corresponds to lags up to 1 year, 2 years and 3 years, when the data is observed at monthly frequency. For each design, $1,000$ Monte Carlo simulations are run, with $T=1000$. All results report rejection rates at the 5\% nominal significance level.

\clearpage
\begin{figure}[h!]
\centering
\begin{subfigure}{0.75\textwidth}
 \centering 
  \includegraphics[width=0.8\linewidth]{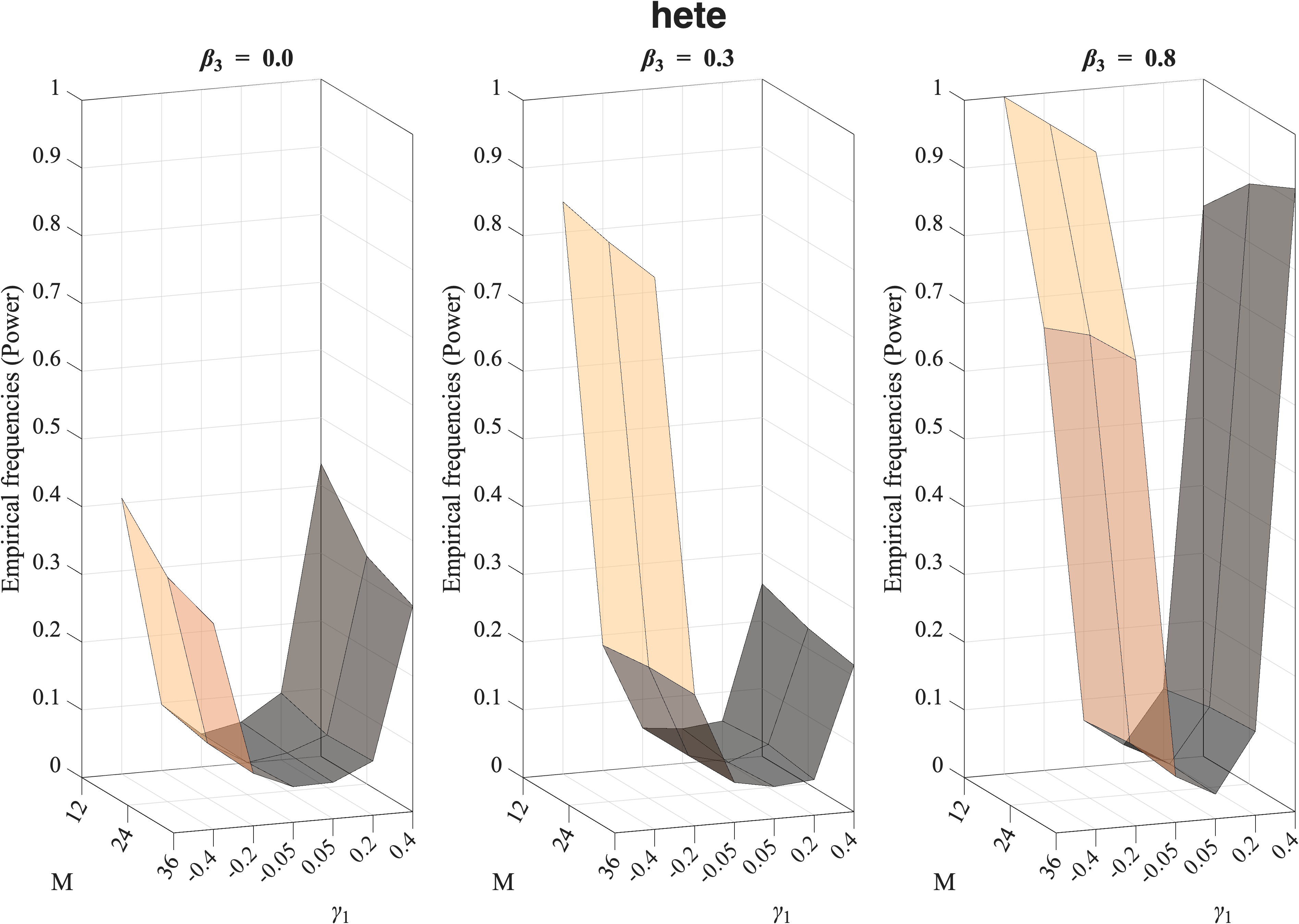}
  \caption{\footnotesize{\textsc{DGP 1b}, with $\beta=\{0,0.3,0.8\}$}}
\end{subfigure}
\par\medskip
\begin{subfigure}{0.75\textwidth}
 \centering
  \includegraphics[width=0.8\linewidth]{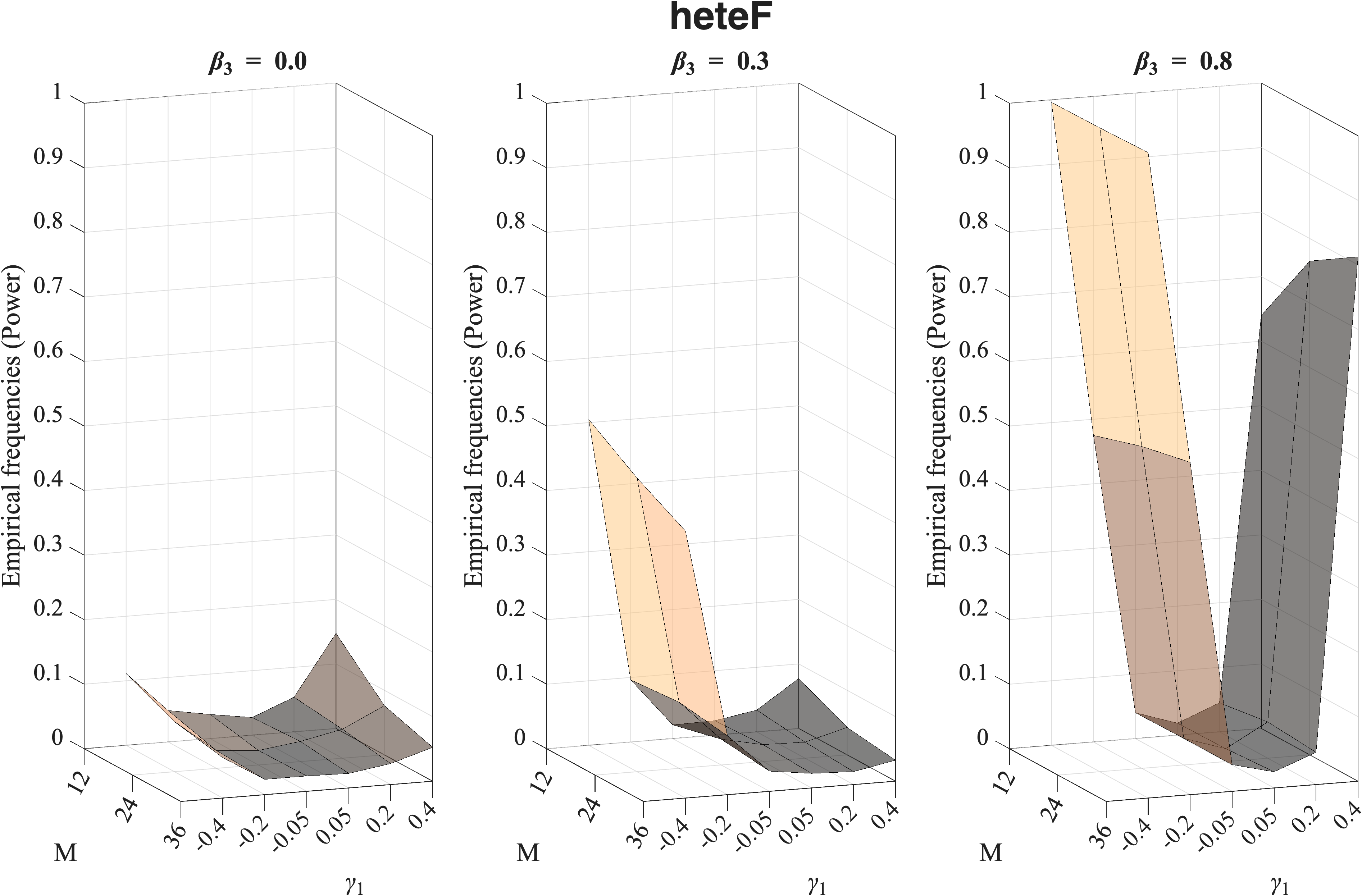}
  \caption{\footnotesize{\textsc{DGP 1b}, with $\beta=\{0,0.3,0.8\}$}}
\end{subfigure}
\par\medskip
\caption[Empirical rates]{
				Power curves of the proposed tests:  \footnotesize{These figures present the rejection rates of the testing procedure associated to three proposed statistics (\textit{Hete, HeteF}), under the alternatives (empirical power); sample size: $T=1000$; $1000$ iterations; the weighting function is the quadratic spectral kernel; the smoothing parameter: $M=\{12,24,36\}$; nominal significance level is 5\%.
			}}
\label{power_LiM_corrected}
\end{figure}
\clearpage

\begin{figure}[h!]
\centering
\begin{subfigure}{0.65\textwidth}
 \centering 
  \includegraphics[width=0.8\linewidth]{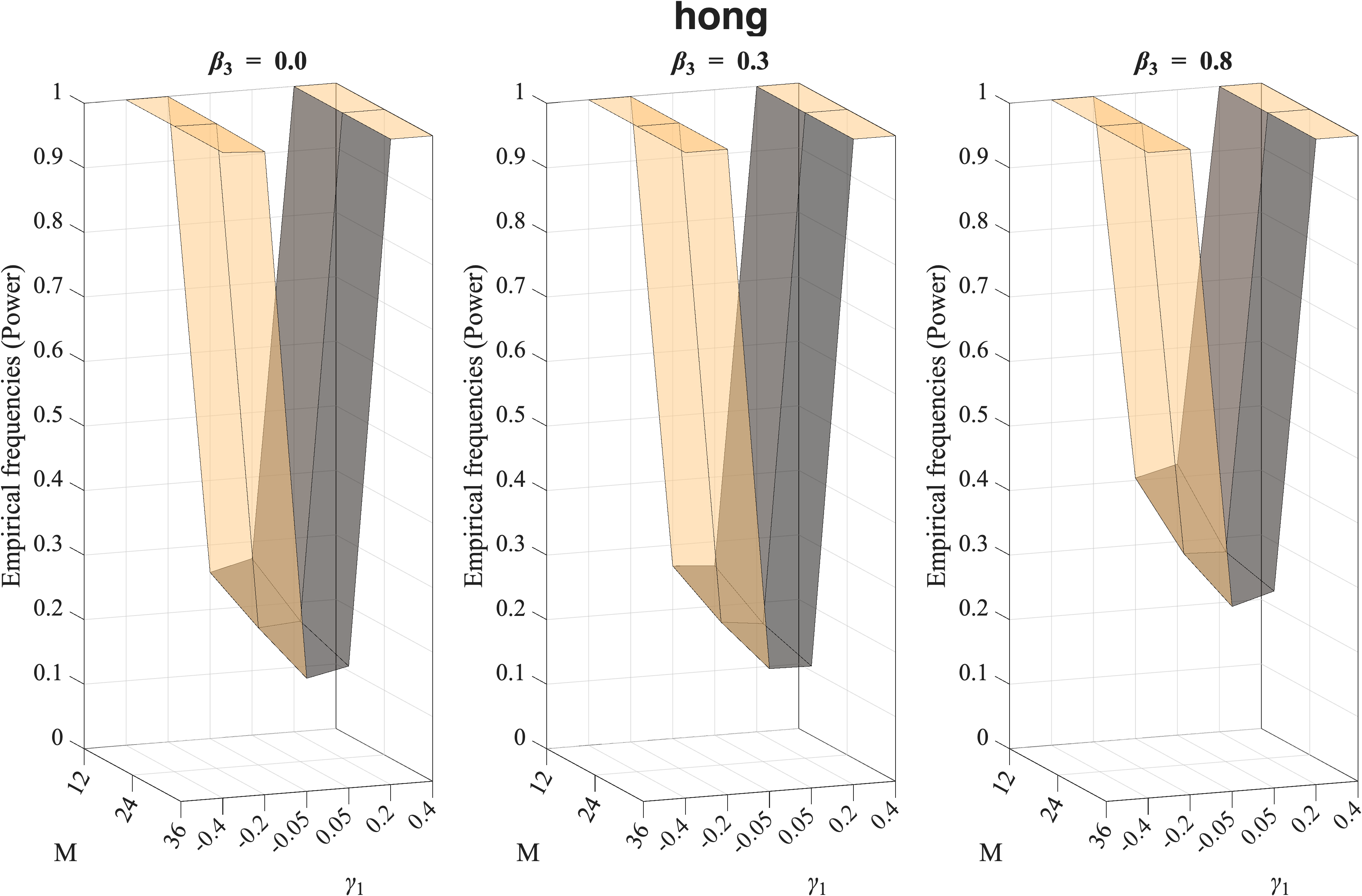}
  \caption{\footnotesize{\textsc{DGP 1b}, with $\beta=\{0,0.3,0.8\}$}}
\end{subfigure}
\par\medskip
\begin{subfigure}{0.65\textwidth}
 \centering
  \includegraphics[width=0.8\linewidth]{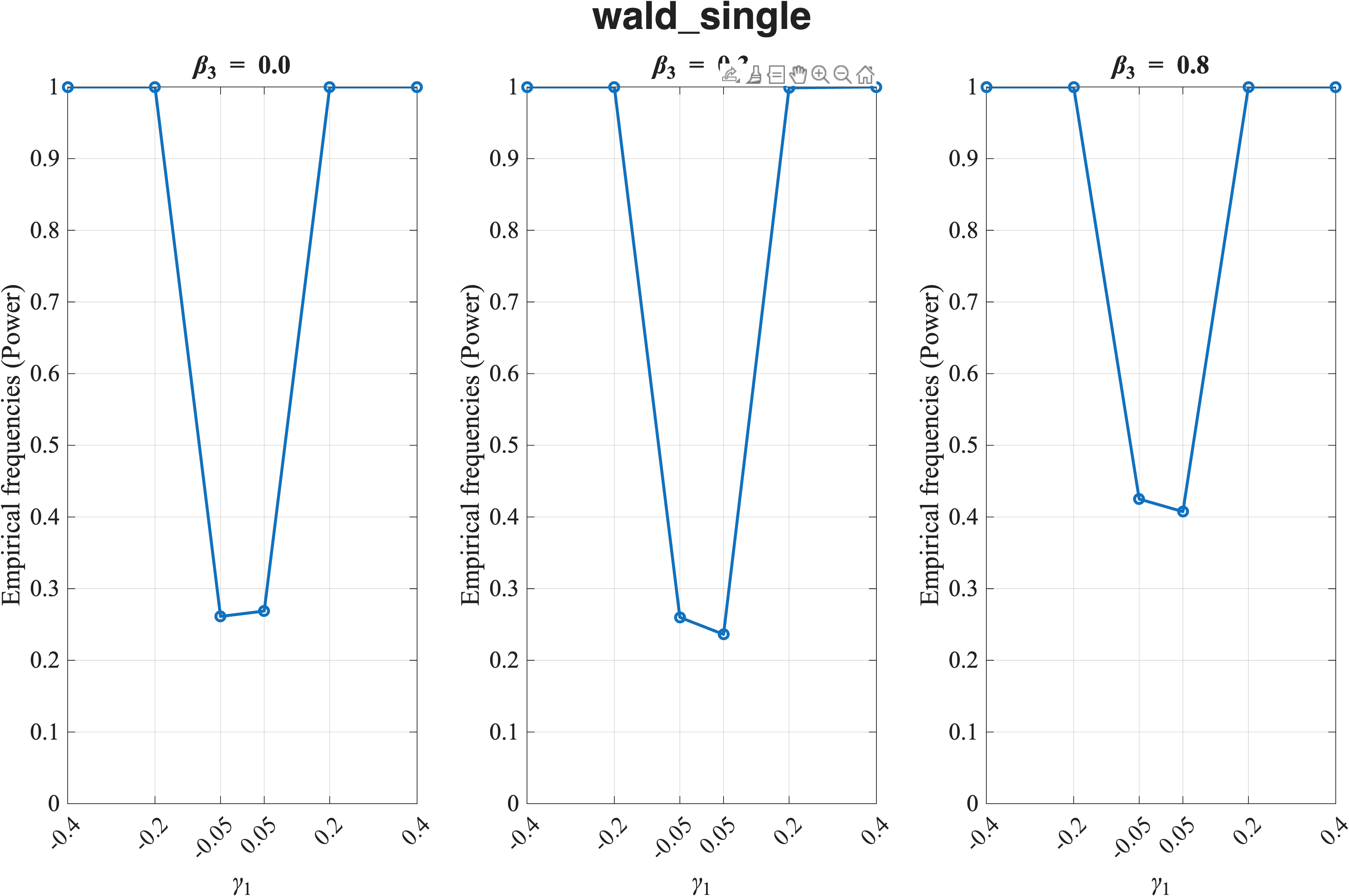}
  \caption{\footnotesize{\textsc{DGP 1b}, with $\beta=\{0,0.3,0.8\}$}}
\end{subfigure}
\par\medskip
\begin{subfigure}{0.65\textwidth}
 \centering
  \includegraphics[width=0.8\linewidth]{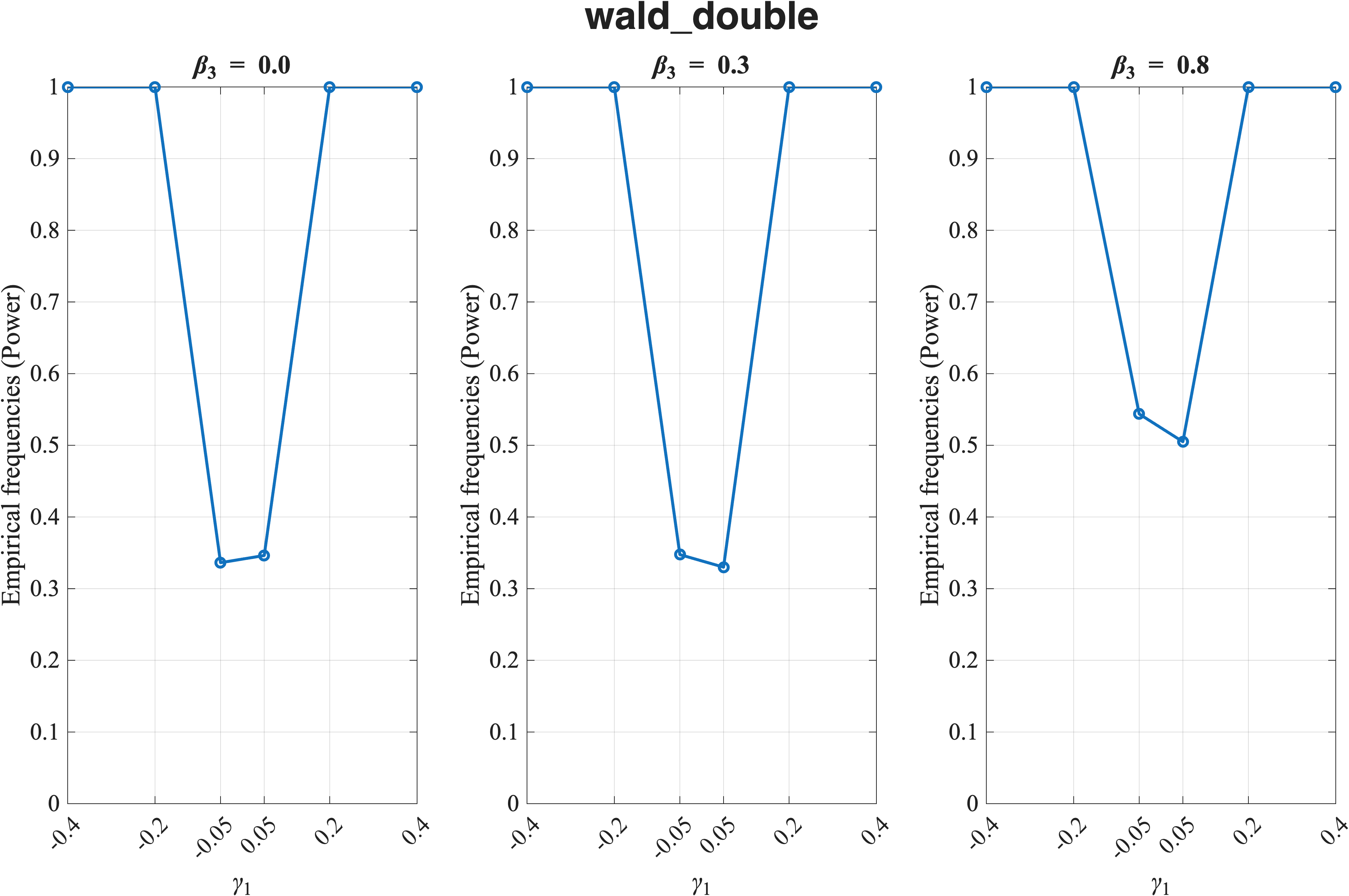}
  \caption{\footnotesize{\textsc{DGP 1b}, with $\beta=\{0,0.3,0.8\}$}}
\end{subfigure}
\caption[Empirical rates]{
				Power curves of the benchmark tests:  \footnotesize{These figures present the rejection rates of the testing procedure associated to three benchmark statistics (\textit{Hong, Wald Single, Wald Double}), under the alternatives (empirical power); sample size, $T=1000$; $1000$ iterations; the weighting function is the quadratic spectral kernel;  the smoothing parameter range is: $M=\{12,24,36\}$; nominal significance level is 5\%.
			}}
\label{power_LiM_benchmark}
\end{figure}
\clearpage
\begin{figure}[h!]
\centering
\begin{subfigure}{0.75\textwidth}
 \centering
  \includegraphics[width=0.8\linewidth]{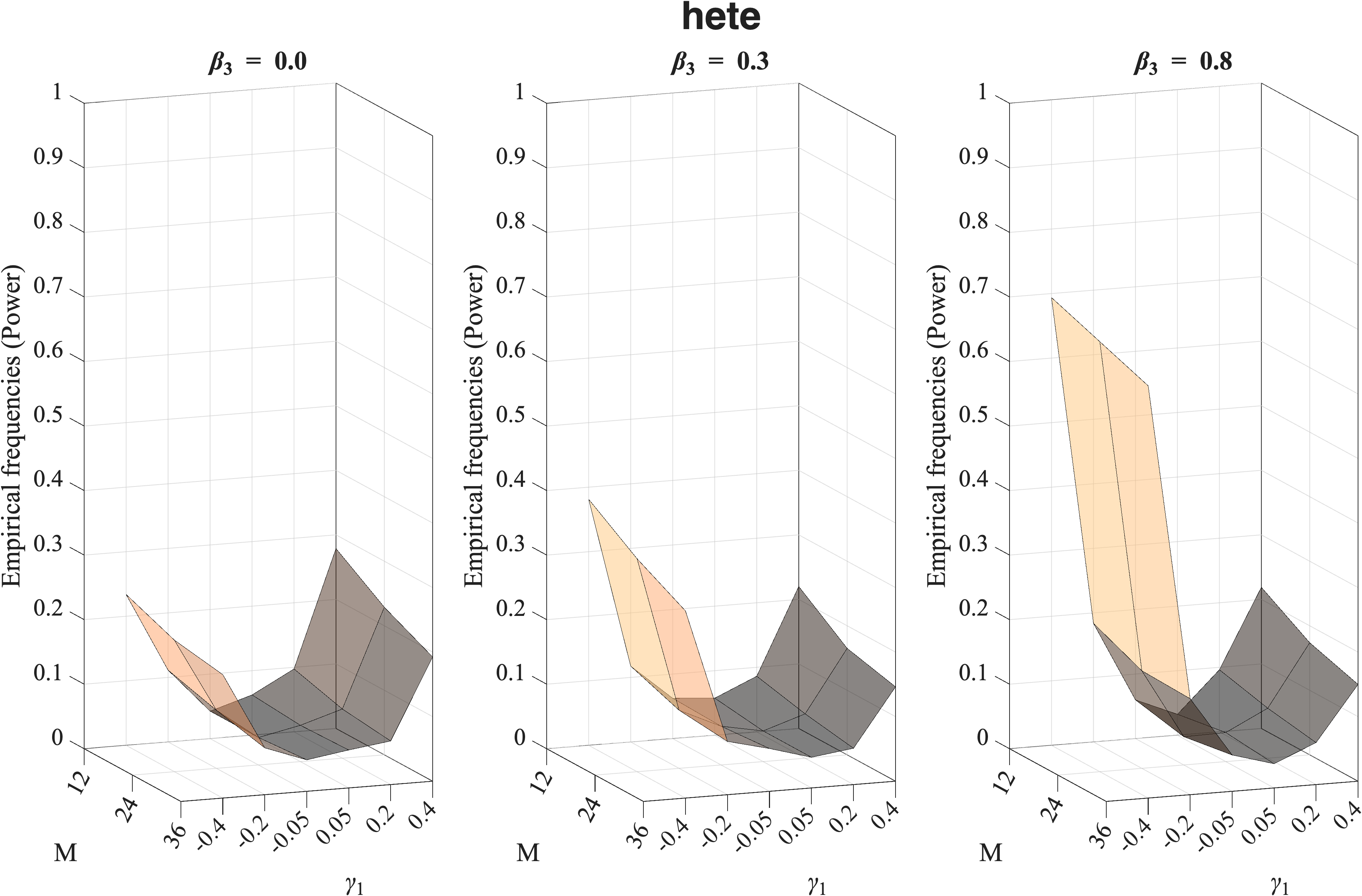}
  \caption{\footnotesize{\textsc{DGP 2b}, with $\beta=\{0,0.3,0.8\}$}}
\end{subfigure}
\par\medskip
\begin{subfigure}{0.75\textwidth}
 \centering
  \includegraphics[width=0.8\linewidth]{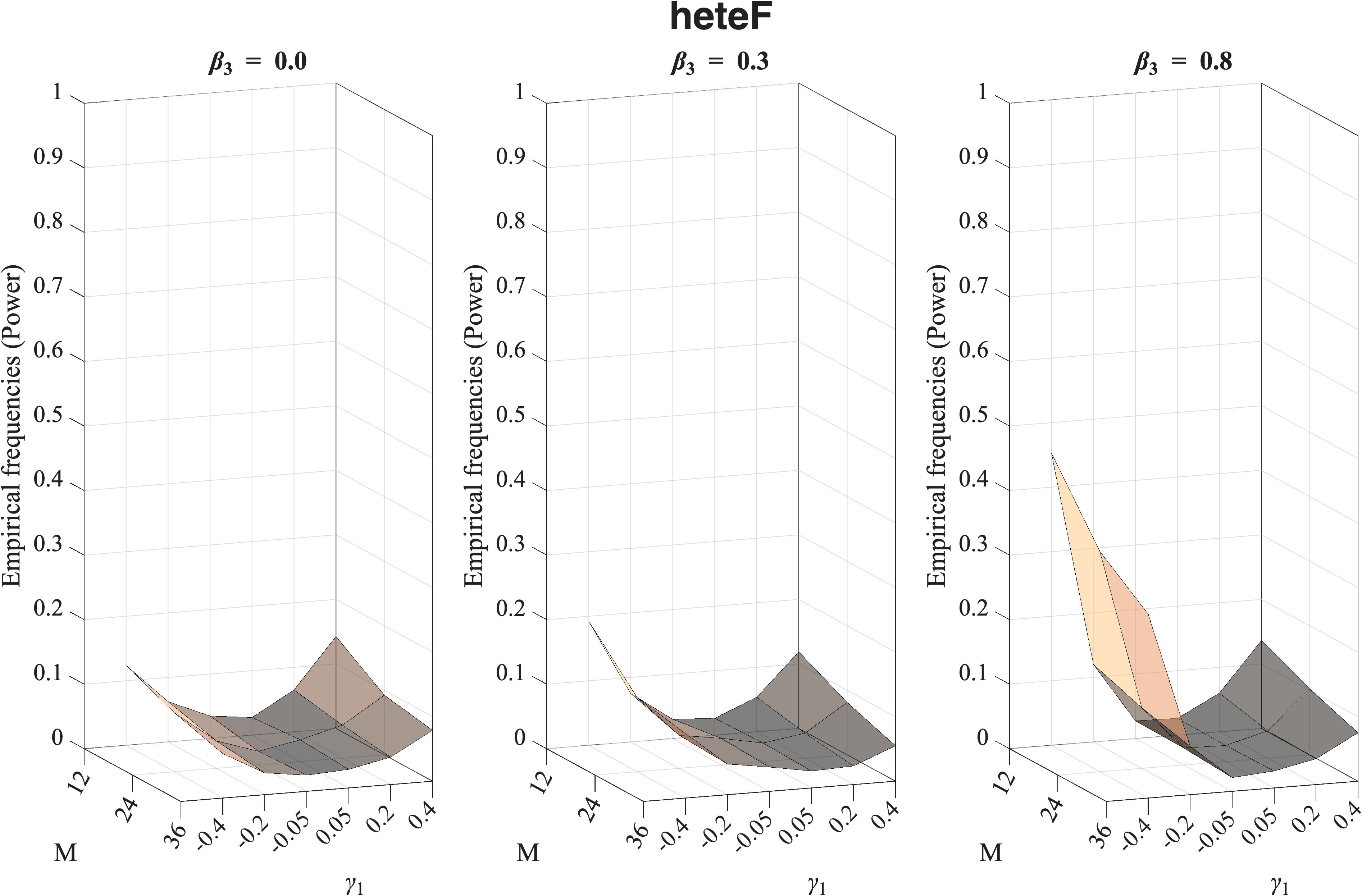}
  \caption{\footnotesize{\textsc{DGP 2b}, with $\beta=\{0,0.3,0.8\}$}}
\end{subfigure}
\caption[Empirical rates]{
				Power curves of the proposed tests:  \footnotesize{These figures present the rejection rates of the testing procedure associated to the proposed statistics (\textit{Hete, HeteF}), under the alternatives (empirical power); sample size: $T=1000$; $1000$ iterations; the weighting function is the quadratic spectral kernel; the smoothing parameter: $M=\{12,24,36\}$; nominal significance level is 5\%.
			}}
\label{power_LiS_corrected}
\end{figure}
\clearpage

\begin{figure}[h!]
\centering
\begin{subfigure}{0.65\textwidth}
 \centering 
  \includegraphics[width=0.8\linewidth]{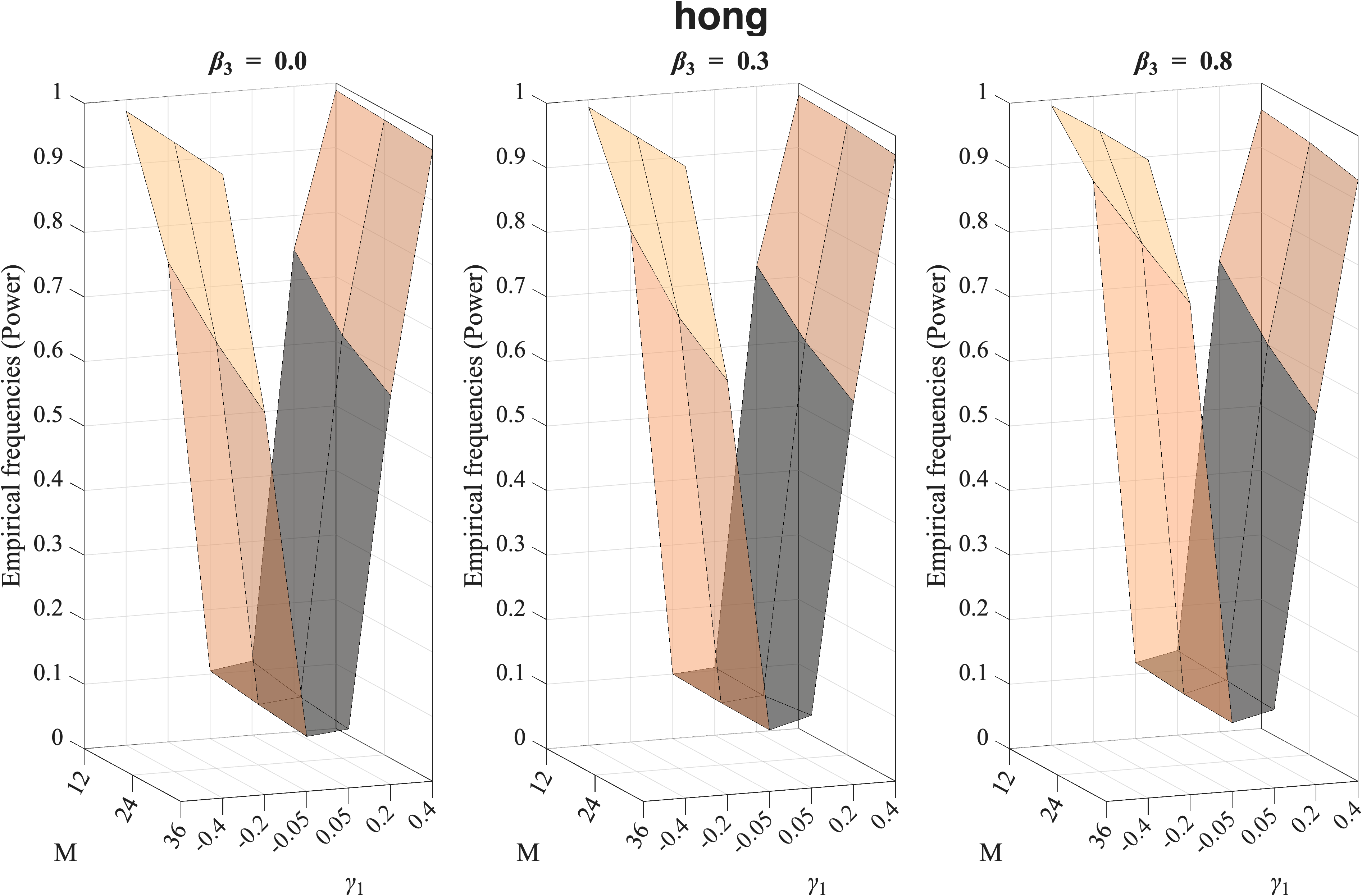}
  \caption{\footnotesize{\textsc{DGP 2b}, with $\beta=\{0,0.3,0.8\}$}}
\end{subfigure}
\par\medskip
\begin{subfigure}{0.65\textwidth}
 \centering
  \includegraphics[width=0.8\linewidth]{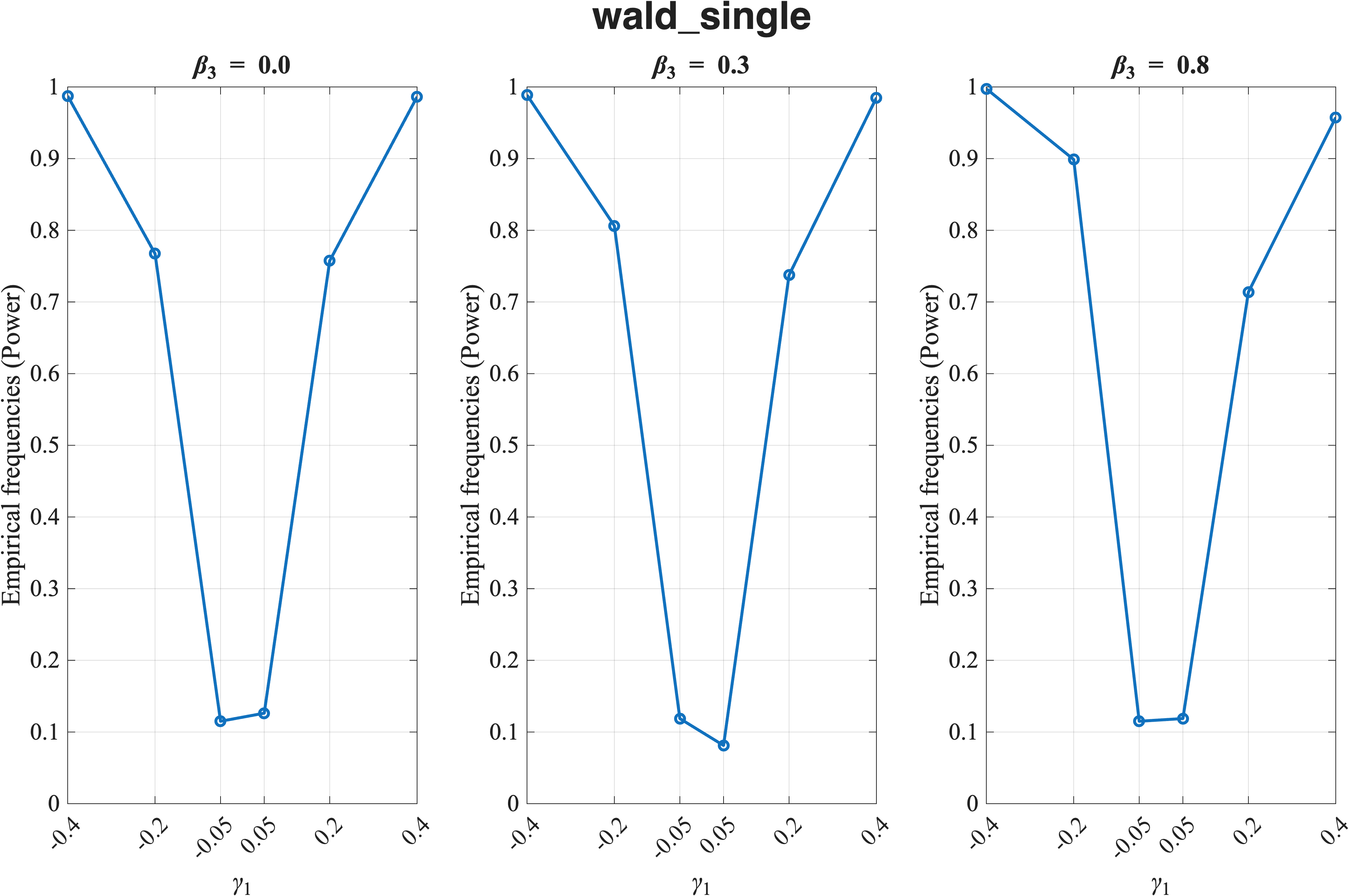}
  \caption{\footnotesize{\textsc{DGP 2b}, with $\beta=\{0,0.3,0.8\}$}}
\end{subfigure}
\par\medskip
\begin{subfigure}{0.65\textwidth}
 \centering
  \includegraphics[width=0.8\linewidth]{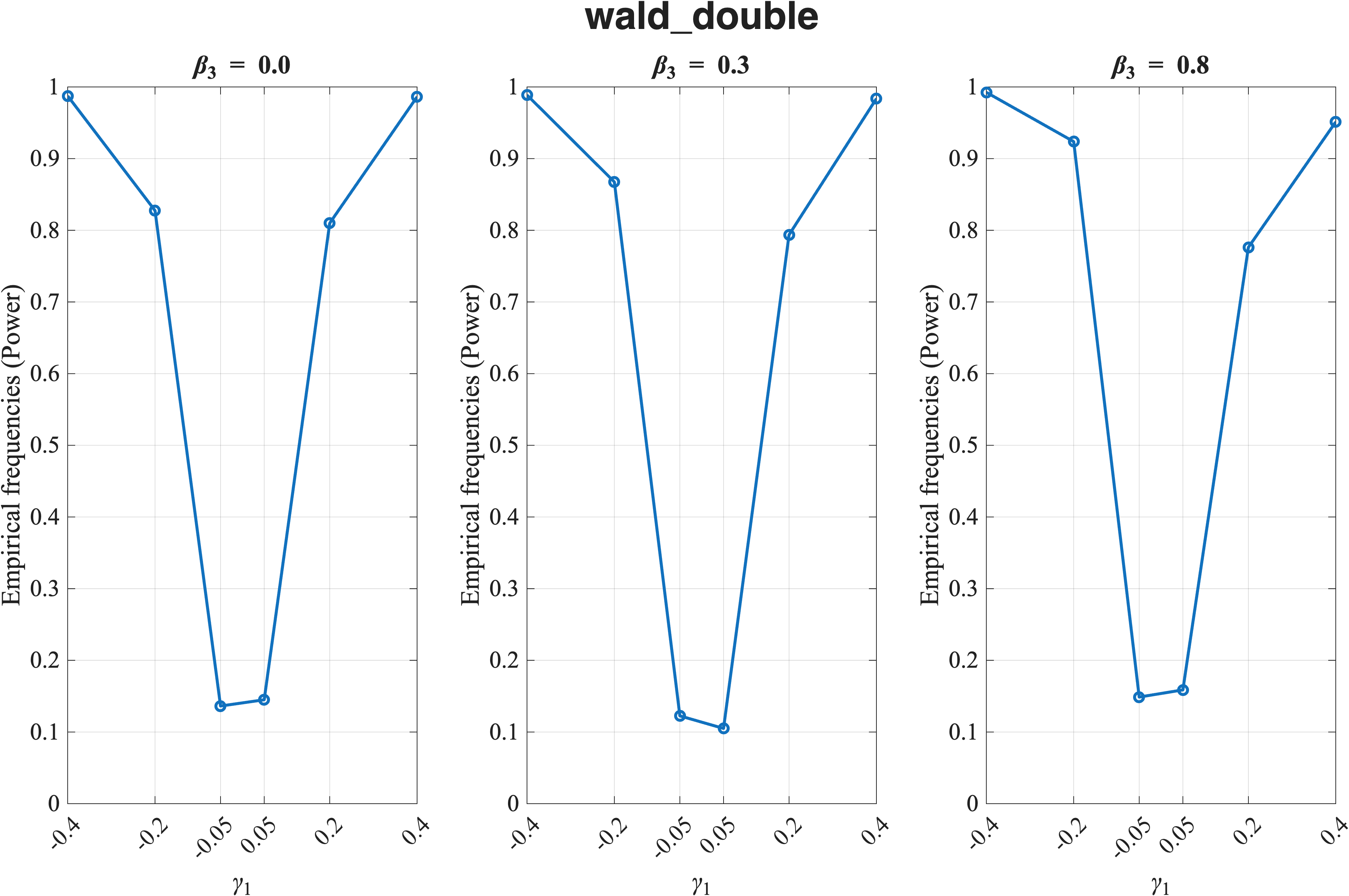}
  \caption{\footnotesize{\textsc{DGP 2b}, with $\beta=\{0,0.3,0.8\}$}}
\end{subfigure}
\caption[Empirical rates]{
				Power curves of the benchmark tests:  \footnotesize{These figures present the rejection rates of the testing procedure associated to three benchmark statistics (\textit{Hong, Wald Single, Wald Double}), under the alternatives (empirical power); sample size: $T=1000$; $1000$ iterations; the weighting function is the quadratic spectral kernel; the smoothing parameter: $M=\{12,24,36\}$; nominal significance level is 5\%.
			}}
\label{power_LiS_benchmark}
\end{figure}

\clearpage
\begin{figure}[h!]
\centering
\begin{subfigure}{0.75\textwidth}
 \centering
  \includegraphics[width=0.8\linewidth]{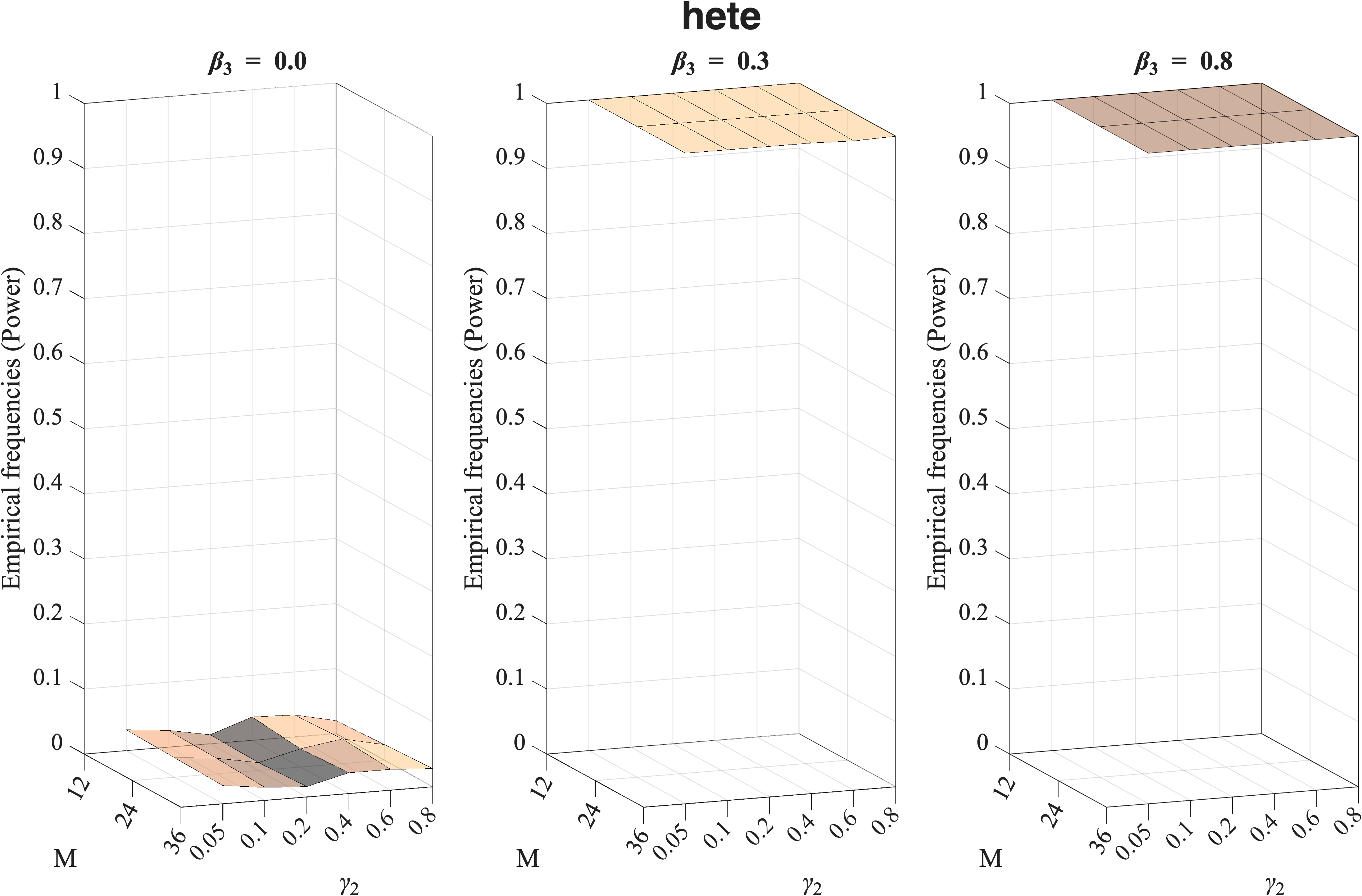}
  \caption{\footnotesize{\textsc{DGP 4b}, with $\beta=\{0,0.3,0.8\}$}}
\end{subfigure}
\par\medskip
\begin{subfigure}{0.75\textwidth}
 \centering
  \includegraphics[width=0.8\linewidth]{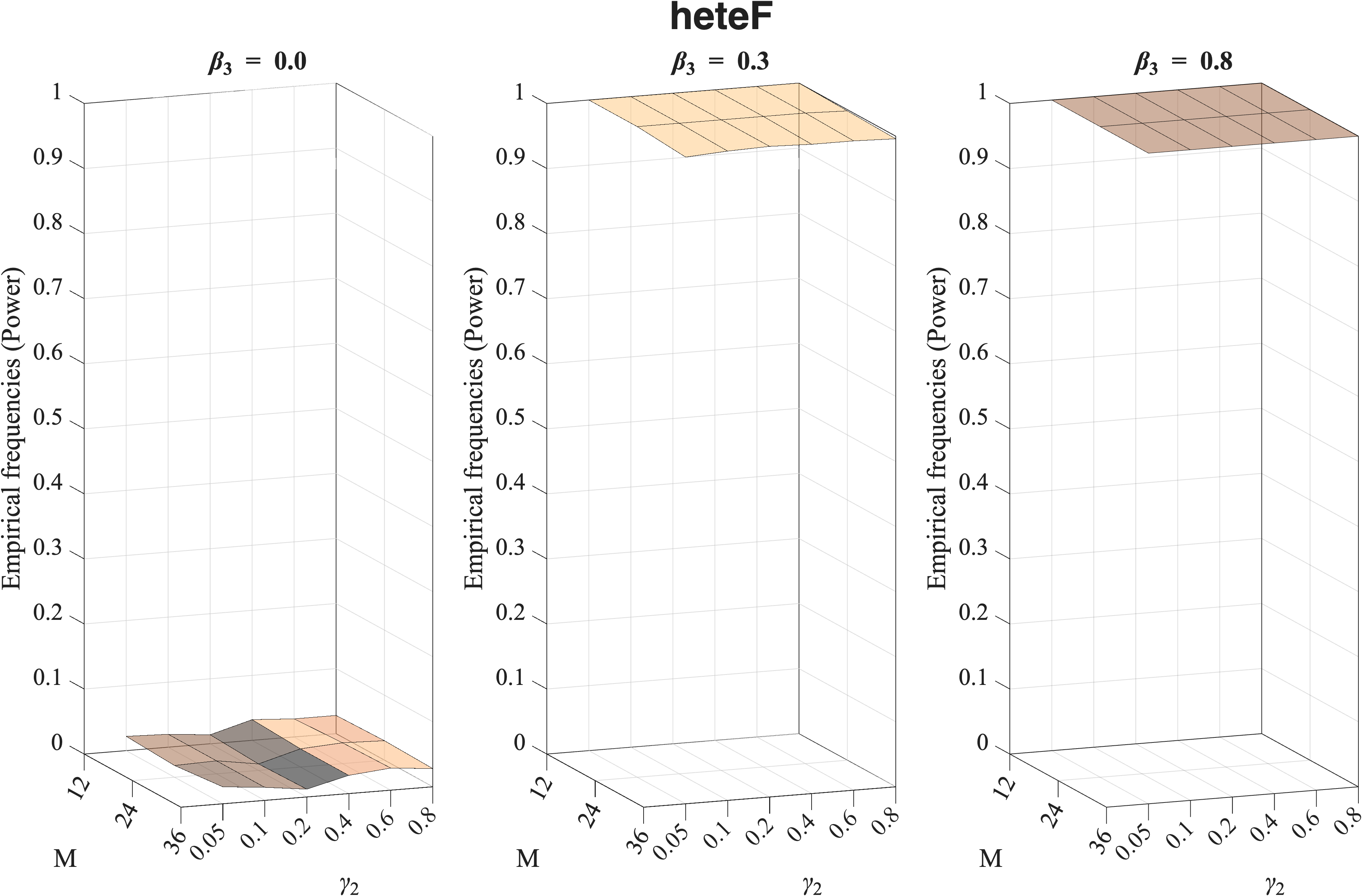}
  \caption{\footnotesize{\textsc{DGP 4b}, with $\beta=\{0,0.3,0.8\}$}}
\end{subfigure}
\caption[Empirical rates]{
				Power curves of the proposed tests:  \footnotesize{These figures present the rejection rates of the testing procedure associated to the proposed statistics (\textit{Hete, HeteF}), under the alternatives (empirical power); sample size: $T=1000$; $1000$ iterations; the weighting function is the quadratic spectral kernel; the smoothing parameter: $M=\{12,24,36\}$; nominal significance level is 5\%.
			}}
\label{power_ViS_corrected}
\end{figure}
\clearpage

\begin{figure}[h!]
\centering
\begin{subfigure}{0.65\textwidth}
 \centering 
  \includegraphics[width=0.8\linewidth]{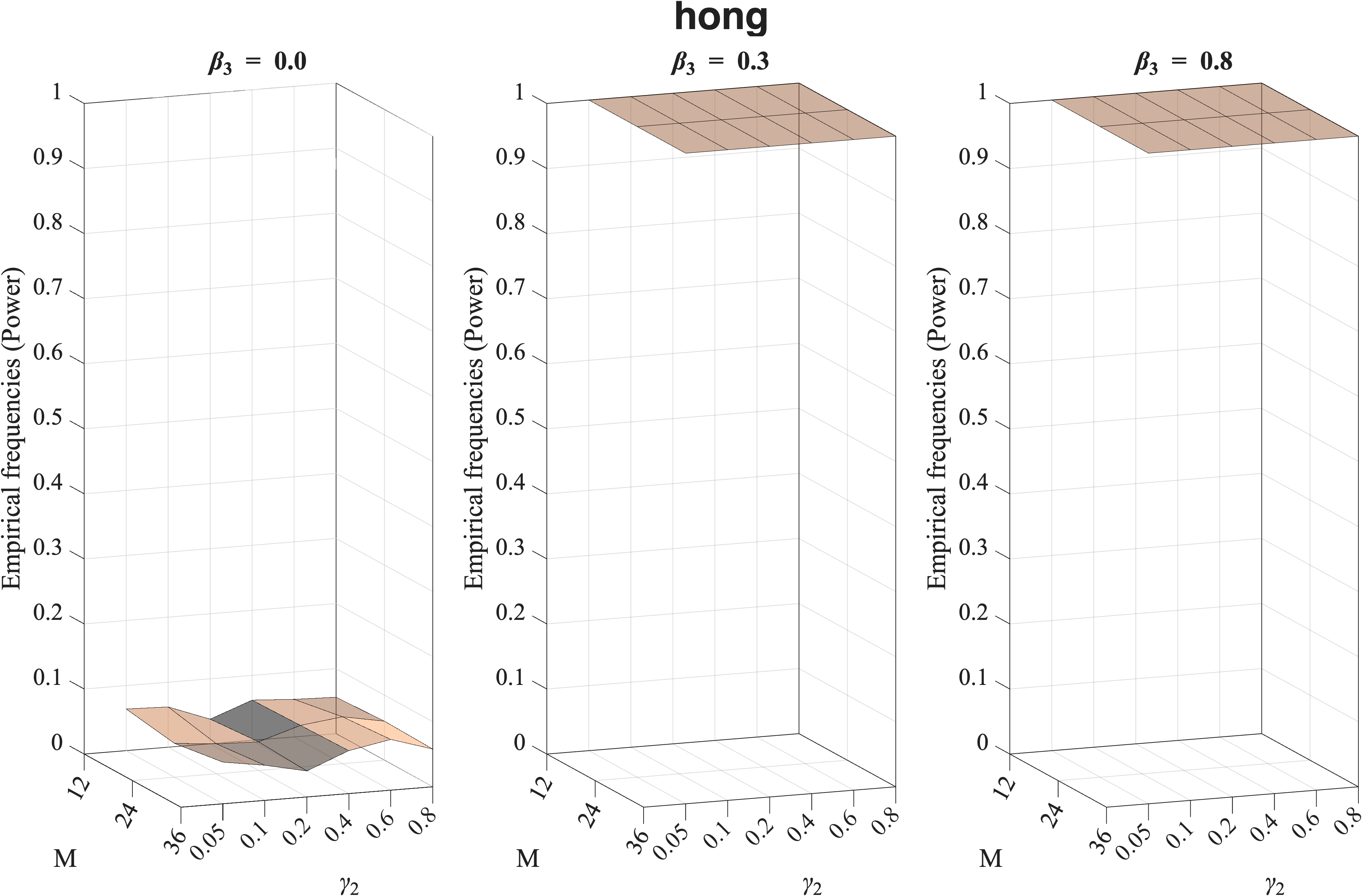}
  \caption{\footnotesize{\textsc{DGP 4b}, with $\beta=\{0,0.3,0.8\}$}}
\end{subfigure}
\par\medskip
\begin{subfigure}{0.65\textwidth}
 \centering
  \includegraphics[width=0.8\linewidth]{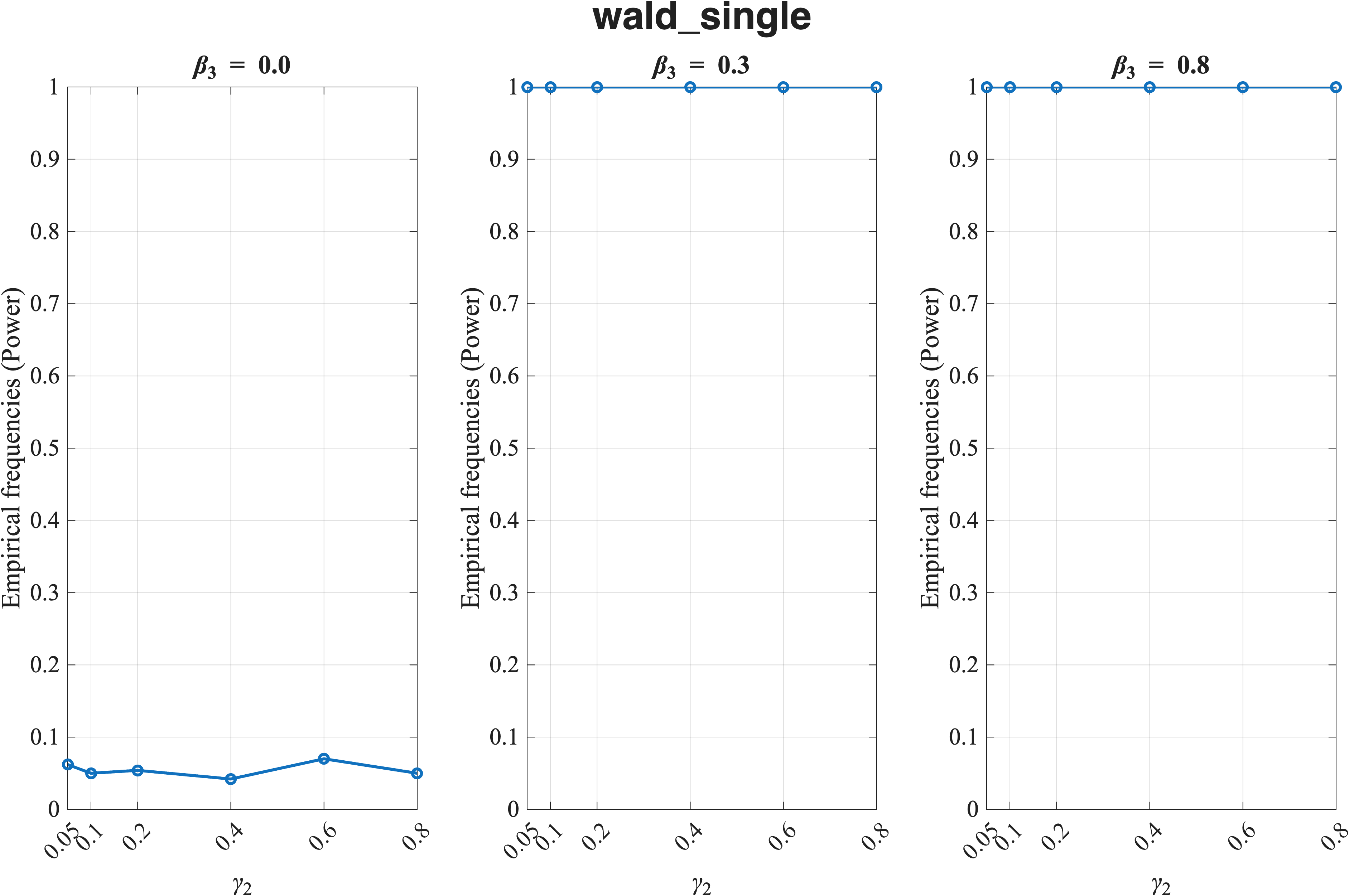}
  \caption{\footnotesize{\textsc{DGP 4b}, with $\beta=\{0,0.3,0.8\}$}}
\end{subfigure}
\par\medskip
\begin{subfigure}{0.65\textwidth}
 \centering
  \includegraphics[width=0.8\linewidth]{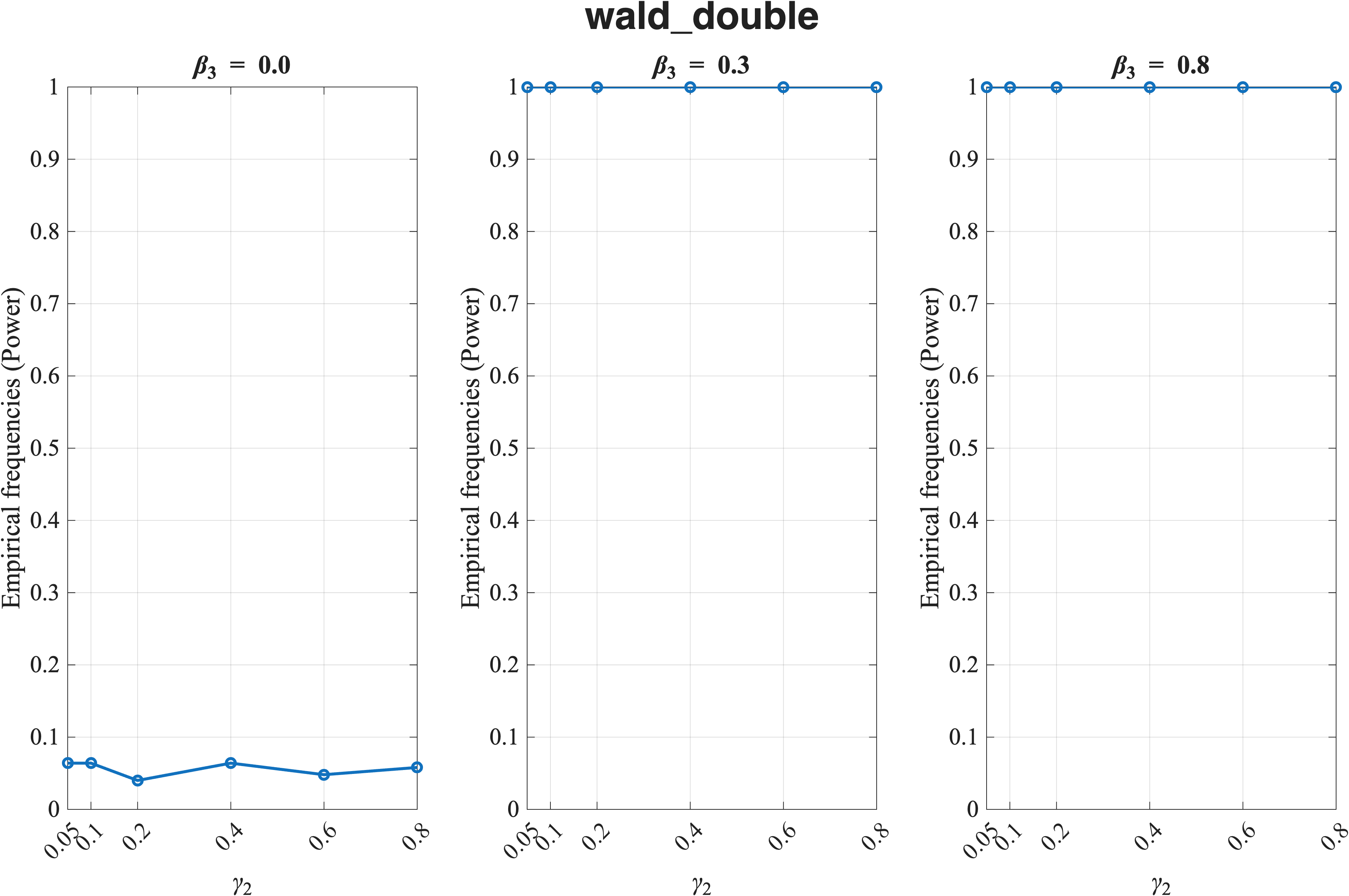}
  \caption{\footnotesize{\textsc{DGP 4b}, with $\beta=\{0,0.3,0.8\}$}}
\end{subfigure}
\caption[Empirical rates]{
				Power curves of the benchmark tests:  \footnotesize{These figures present the rejection rates of the testing procedure associated to three benchmark statistics (\textit{Hong, Wald Single, Wald Double}), under the alternatives (empirical power); sample size: $T=1000$; $1000$ iterations; the weighting function is the quadratic spectral kernel; the smoothing parameter: $M=\{12,24,36\}$; nominal significance level is 5\%.
			}}
\label{power_ViS_benchmark}
\end{figure}

\end{document}